\documentclass{fundam}

\publyear{22}
\papernumber{2102}
\volume{185}
\issue{1}

\usepackage{url} %
\usepackage{graphicx}%

\usepackage{amsmath, amssymb, amsfonts}
\usepackage{mathtools}

\usepackage{xparse}
\usepackage{hyperref}
\usepackage{cleveref}

\usepackage{thmtools}
\newtheorem{defi}{Definition}[section]
\crefname{defi}{Definition}{Definitions}
\newtheorem{thm}[defi]{Theorem}%
\crefname{thm}{Theorem}{Theorems}
\crefname{fct}{Fact}{Facts}
\newtheorem{lem}[defi]{Lemma}
\crefname{lem}{Lemma}{Lemmas}
\newtheorem{ex}[defi]{Example}
\crefname{ex}{Example}{Examples}

\crefname{assum}{Assumption}{Assumptions}
\newtheorem{prop}[defi]{Proposition}
\crefname{prop}{Proposition}{Propositions}
\newtheorem{rem}[defi]{Remark}
\crefname{rem}{Remark}{Remarks}
\newtheorem{cor}[defi]{Corollary}
\crefname{cor}{Corollary}{Corollaries}
\newtheorem{cla}[defi]{\textcolor{gray}{\scalebox{.3}[1]{$\blacktriangleright$}} Claim}
\crefname{cla}{Claim}{Claims}
\crefname{section}{Section}{Sections}
\crefname{figure}{Figure}{Figures}
\newtheorem{ques}{Open Question}
\crefname{ques}{Question}{Questions}

\definecolor{cLightGrey}{HTML}{d1d8e0}
\definecolor{cGreen}{HTML}{20bf6b}
\newcommand{\proofcase}[1]{\noindent{\hypersetup{linkcolor={cGreen!50!black!100}, citecolor={cGreen!50!black!100}, urlcolor={cGreen!50!black!100}, filecolor={cGreen!50!black!100}}\setlength{\fboxsep}{1.pt}\colorbox{cLightGrey}{\normalcolor\textsf{#1}\hspace{.2em}}}~~}
\newenvironment{claproof}{\trivlist\PRstyle\item[]{\bfseries Claim Proof:}\newline}{\CLAQED\endtrivlist}
\RenewDocumentEnvironment{proof}{O{Proof}}{\trivlist\PRstyle\item[]{\bfseries #1:}\newline}{\QED\endtrivlist}
\def\CLAQED{\unskip\nobreak\hfil
    \penalty50\hskip1em\null\nobreak\hfil\scalebox{.3}[1]{\textcolor{gray}{$\blacktriangleleft$}}
    \parfillskip=0pt\finalhyphendemerits=0\endgraf}

\usepackage{xcolor}
\definecolor[named]{ACMBlue}{cmyk}{1,0.1,0,0.1}
\definecolor[named]{ACMYellow}{cmyk}{0,0.16,1,0}
\definecolor[named]{ACMOrange}{cmyk}{0,0.42,1,0.01}
\definecolor[named]{ACMRed}{cmyk}{0,0.90,0.86,0}
\definecolor[named]{ACMLightBlue}{cmyk}{0.49,0.01,0,0}
\definecolor[named]{ACMGreen}{cmyk}{0.20,0,1,0.19}
\definecolor[named]{ACMPurple}{cmyk}{0.55,1,0,0.15}
\definecolor[named]{ACMDarkBlue}{cmyk}{1,0.58,0,0.21}
\hypersetup{colorlinks,
    linkcolor=ACMPurple,
    citecolor=ACMPurple,
    urlcolor=ACMDarkBlue,
    filecolor=ACMDarkBlue}

\MakeRobust{\ref}%

\makeatletter
\newcommand{\labeltext}[2]{%
    \@bsphack
    \csname phantomsection\endcsname %
    \def\@currentlabel{#1}{\label{#2}}%
    \@esphack
}
\makeatother

\usepackage{array}

\usepackage{tikz}
\usetikzlibrary{calc, trees, bending, backgrounds, shapes, shapes.geometric, fit, arrows, arrows.meta, positioning, graphs, quotes, babel}
\pgfdeclarelayer{background}
\pgfdeclarelayer{foreground}
\pgfsetlayers{background, main, foreground}
\usepackage{ebproof}

\usepackage{stackengine}

\makeatletter
\let\orgdescriptionlabel\descriptionlabel
\renewcommand*{\descriptionlabel}[1]{%
    \let\orglabel\label
    \let\label\@gobble
    \phantomsection
    \edef\@currentlabel{#1}%
    \let\label\orglabel
    \orgdescriptionlabel{#1}%
}
\makeatother

\usepackage{etoolbox}
\makeatletter
\let\original@footnote\footnote
\newcommand{\align@footnote}[1]{%
    \ifmeasuring@
        \chardef\@tempfn=\value{footnote}%
        \footnotemark
        \setcounter{footnote}{\@tempfn}%
    \else
        \iffirstchoice@
        \original@footnote{#1}%
        \fi
    \fi}
\pretocmd{\start@align}{\let\footnote\align@footnote}{}{}
\makeatother

\usepackage{subfig}
\usepackage{multirow}
\usepackage{makecell}%

\usepackage{scalerel}

\usepackage{tablefootnote}
\usepackage{imakeidx} %
\makeindex[noautomatic, title=Index of Notions] %
\indexsetup{noclearpage} %
\makeindex[noautomatic, name=notation, title=Index of Notation] %
\usepackage[xcolor, notion, quotation,paper]{knowledge}
\knowledgeconfigureenvironment{theorem,lemma,proof}{}
\knowledgeconfigure{diagnose line=true}
\makeatletter
\NewDocumentCommand\introalign{m}{\ifmeasuring@\expandafter\@firstoftwo\else\expandafter\@secondoftwo\fi{#1}{\phantomsection\intro*#1}}
\makeatother

\knowledgedefault{index}
\ExplSyntaxOn
\bool_new:N \g__yn_index_off_bool
\NewDocumentCommand \ynindexoff {} { \bool_gset_true:N \g__yn_index_off_bool }
\DeclareRobustCommand{\ynidxdesc}[1]{\ (#1)}
\prop_new:N \g__yn_index_desc_prop
\prop_gset_from_keyval:Nn \g__yn_index_desc_prop
  {
    range = {integer~interval}, card = {cardinality},
    domain = {universe~of~a~graph~or~an~algebra}, vsig = {variables}, vsigpm = {variables,~variable~complements,~and~the~difference~constant}, sig = {signature},
    PCoR = {positive~calculus~of~relations}, REL = {valuations~of~relational~models}, fin = {finite~valuations},
    DREL = {valuations~of~deterministic~relational~models},
    ttest = {hypotheses~for~primitive~tests}, tfunc = {functionality~hypotheses},
    tgrid = {grid~hypotheses},
    Sur = {valuations~with~a~surjective~graph~homomorphism},
    Tr = {translation~for~eliminating~loop~hypotheses},
    TrPWP = {translation~for~eliminating~loop~hypotheses~in~while~programs},
    Uorg = {auxiliary~primitive~program}, Kgo = {auxiliary~primitive~programs},
    TrOne = {translation~to~one~primitive~test}, sigmaAB = {substitution~to~two~actions},
    trunc = {subgraph~induced~by~the~source~and~target},
    quo = {equivalence~class},
    TrbarI = {translation~eliminating~the~difference~constant},
    GS = {guarded~string~language},
    termset = {set~of~$X$-terms}, size = {number~of~symbols},
    gval = {graph~of~a~valuation}, tcompl = {complement~of~a~test},
    unionb = {if-then-else}, starb = {while}, plusb = {do-while},
    sdF = {SDIPDL~formulas}, sdP = {SDIPDL~programs}, ctrue = {true}, cfalse = {false}
  }
\prop_new:N \g__yn_index_disp_prop
\prop_gset_from_keyval:Nn \g__yn_index_disp_prop
  {
    range = {\range{l, r}}, card = {\card{X}}, domain = {\domain{\graph}},
    fin = {\REL^{\fin}}, tfunc = {\tfunc^{W}},
    Sur = {\Sur _{\glangvar}(\algclass)},
    Tr = {\Tr _{\term[3]}(\term)}, TrPWP = {\TrPWP _{\term[3]}(\term)},
    Kgo = {\Kgo{i}}, TrOne = {\TrOne(\term)},
    trunc = {\trunc{\graph}},
    TrbarI = {\TrbarI _{\tQ}}, quo = {\quo{x}_{\sim _{\tQ}}},
    GS = {\GS(\term)},
    termset = {\termset _{X}}, size = {\size{\term}}, gval = {\gval(\val,~x,~y)},
    tcompl = {\tcompl{b}}, unionb = {\term[1]\unionb{b}\term[2]},
    starb = {\term^{\starb{b}}}, plusb = {\term^{\plusb{b}}}
  }
\prop_new:N \g__yn_index_merge_prop
\prop_gset_from_keyval:Nn \g__yn_index_merge_prop
  {
    adom = {{antidomain~operator}{\term^{\adom}}},
    dom = {{domain~operator}{\term^{\dom}}},
    lop = {{graph~loop~operator}{\term^{\lop}}},
    glang = {{graph~language}{\glang(\term)}},
    homo = {{graph~homomorphism}{\graph \homo \graph[2]}},
    src = {{source}{\src^{\graph}}},
    tgt = {{target}{\tgt^{\graph}}},
    psig = {{primitive~test}{\psig}},
    asig = {{action}{\asig}},
    Tests = {{test}{\Tests}},
    PWP = {{propositional~while~program}{\PWP}},
    PWPlop = {{propositional~while~programs~with~graph~loop}{\PWPlop}},
    PWPcap = {{propositional~while~programs~with~intersection}{\PWPcap}},
    interval = {{interval}{\interval{x_0}{y_0}_{\val, \term}}},
    restr = {{submodel}{\val\restr{B}}}
  }
\prop_new:N \g__yn_index_mword_prop
\prop_new:N \g__yn_index_mdisp_prop
\cs_new_protected:Npn \__yn_index_merged:nnn #1 #2 #3
  {
    \prop_gput:Nnn \g__yn_index_mword_prop {#1} {#2}
    \prop_gput:Nnn \g__yn_index_mdisp_prop {#1} {#3}
  }
\prop_map_inline:Nn \g__yn_index_merge_prop { \__yn_index_merged:nnn {#1} #2 }
\int_new:N \l__yn_index_args_int
\int_new:N \g__yn_index_order_int
\bool_new:N \l__yn_index_intro_bool
\tl_new:N \l__yn_index_entry_tl
\tl_new:N \l__yn_index_sort_tl
\tl_new:N \l__yn_index_rest_tl
\tl_new:N \l__yn_index_disp_tl
\tl_new:N \l__yn_index_desc_tl
\cs_set:Npn \knowledge_index:nn #1 #2
  { \bool_if:NF \g__yn_index_off_bool { \__yn_index:nw {#1} #2 \q_stop } }
\cs_new_protected:Npn \__yn_index:nw #1 #2 #3 \q_stop
  {
    \tl_set:Nn \l__yn_index_rest_tl {#3}
    \bool_set_false:N \l__yn_index_intro_bool
    \regex_match:nnT { \|\s*ynintro\s*\Z } {#3}
      {
        \bool_set_true:N \l__yn_index_intro_bool
        \regex_replace_once:nnN { ynintro\s*\Z } { } \l__yn_index_rest_tl
      }
    \bool_lazy_and:nnTF { \tl_if_single_token_p:n {#2} } { \token_if_cs_p:N #2 }
      {
        \__yn_index_symbol:N #2
        \bool_lazy_all:nT
          {
            { \l__yn_index_intro_bool }
            { \cs_if_free_p:c { g__yn_index_noted_ \cs_to_str:N #2 _tl } }
          }
          {
            \tl_gclear_new:c { g__yn_index_noted_ \cs_to_str:N #2 _tl }
            \int_gincr:N \g__yn_index_order_int
            \tl_set:Nx \l__yn_index_entry_tl
              {
                \int_eval:n { 1000 + \g__yn_index_order_int } @
                \l__yn_index_disp_tl
                \tl_if_empty:NF \l__yn_index_desc_tl
                  { \tl_to_str:n { \ynidxdesc } { \l__yn_index_desc_tl } }
                \l__yn_index_rest_tl
              }
            \__yn_index_emit:n { notation }
          }
      }
      {
        \exp_args:NnV \__yn_index_word:nn {#2} \l__yn_index_rest_tl
        \__yn_index_emit:n {#1}
      }
  }
\cs_new_protected:Npn \__yn_index_symbol:N #1
  {
    \int_zero:N \l__yn_index_args_int
    \exp_args:Nne \regex_count:nnN { \# } { \cs_parameter_spec:N #1 } \l__yn_index_args_int
    \tl_clear:N \l__yn_index_desc_tl
    \prop_get:NeNTF \g__yn_index_mdisp_prop { \cs_to_str:N #1 } \l_tmpb_tl
      { \prop_get:NeN \g__yn_index_mword_prop { \cs_to_str:N #1 } \l__yn_index_desc_tl }
      {
        \prop_get:NeNF \g__yn_index_disp_prop { \cs_to_str:N #1 } \l_tmpb_tl
          {
            \tl_set:Nx \l_tmpb_tl
              { \exp_not:N #1 \prg_replicate:nn { \l__yn_index_args_int } { { \exp_not:N \cdot } } }
          }
        \prop_get:NeNF \g__yn_index_desc_prop { \cs_to_str:N #1 } \l__yn_index_desc_tl
          { \tl_clear:N \l__yn_index_desc_tl }
      }
    \tl_set:Nx \l__yn_index_disp_tl { \tl_to_str:n { \ensuremath } { \tl_to_str:N \l_tmpb_tl } }
    \tl_set:Nx \l__yn_index_desc_tl { \tl_to_str:N \l__yn_index_desc_tl }
  }
\cs_new_protected:Npn \__yn_index_word:nn #1 #2
  {
    \tl_set:Nn \l__yn_index_entry_tl {#1#2}
    \regex_match:nnT { [\$\\] } {#1#2}
      {
        \tl_set:Nx \l__yn_index_sort_tl { \tl_to_str:n {#1#2} }
        \regex_replace_once:nnN { \|.* } { } \l__yn_index_sort_tl
        \regex_replace_all:nnN { [\\\$\{\}\s] } { } \l__yn_index_sort_tl
        \tl_put_left:Nx \l__yn_index_entry_tl { \l__yn_index_sort_tl @ }
      }
  }
\cs_new_protected:Npn \__yn_index_emit:n #1
  {
    \regex_replace_once:nnN { \|\s*\Z } { \|hyperpage } \l__yn_index_entry_tl
    \tl_if_empty:nTF {#1}
      { \exp_args:NV \index \l__yn_index_entry_tl }
      { \exp_args:NnV \__yn_index_named:nn {#1} \l__yn_index_entry_tl }
  }
\cs_new_protected:Npn \__yn_index_named:nn #1 #2 { \index [{#1}] {#2} }
\ExplSyntaxOff
 
\definecolor{Dark Ruby Red}{HTML}{580507}
\definecolor{Dark Blue Sapphire}{HTML}{053641}
\knowledgestyle{intro notion}{color={Dark Ruby Red}, emphasize, index style=ynintro}
\knowledgestyle{notion}{color={Dark Blue Sapphire}}
\DeclarePairedDelimiter\set{\{}{\}}
\DeclarePairedDelimiter\tuple{\langle}{\rangle}
\knowledgenewrobustcmd\range[1]{\cmdkl{[}#1\cmdkl{]}} %
\knowledgenewrobustcmd\card[1]{\cmdkl{\#}#1} %

\newcommand{\const}[1]{\mathsf{#1}}
\newcommand{\bl}{\_}
\newcommand{\defeq}{\mathrel{\ensurestackMath{\stackon[1pt]{=}{\scriptscriptstyle\Delta}}}}
\newcommand{\defiff}{\mathrel{\ensurestackMath{\stackon[1pt]{\Leftrightarrow}{\scriptscriptstyle\Delta}}}}

\newcommand{\nat}{\mathbb{N}}
\newcommand{\Z}{\mathbb{Z}}

\newcommand{\NC}{\textsc{NC}}
\newcommand{\PTIME}{\textsc{P}}
\newcommand{\NP}{\textsc{NP}}
\newcommand{\coNP}{\textsc{coNP}}
\newcommand{\PSPACE}{\textsc{PSpace}}
\newcommand{\NPSPACE}{\textsc{NPSpace}}
\newcommand{\EXPSPACE}{\textsc{ExpSpace}}
\newcommand{\TWOEXPTIME}{\textsc{2ExpTime}}
\newcommand{\PiOne}{\mathrm{\Pi}^{0}_{1}}
\newcommand{\SigmaOne}{\mathrm{\Sigma}^{0}_{1}}
\newcommand{\PiOneOne}{\mathrm{\Pi}^{1}_{1}}

\NewDocumentCommand\alg{O{1}}{%
    \ifcase#1 undefined
    \or \mathcal{A}
    \or \mathcal{B}
    \else undefined \fi
}
\NewDocumentCommand\algclass{O{1}}{%
    \ifcase#1 undefined
    \or \mathcal{C}
    \else undefined \fi
}
\NewDocumentCommand\graph{O{1}}{%
    \ifcase#1 undefined
    \or G
    \or H
    \or J
    \else undefined \fi
}
\NewDocumentCommand\glangvar{O{1}}{%
    \ifcase#1 undefined
    \or \mathcal{H}
    \or \mathcal{K}
    \else undefined \fi
}
\NewDocumentCommand\rel{O{1}}{%
    \ifcase#1 undefined
    \or R
    \or S
    \else undefined \fi
}
\NewDocumentCommand\term{O{1}}{%
    \ifcase#1 undefined
    \or t
    \or s
    \or u
    \else undefined \fi
}
\NewDocumentCommand\bterm{O{1}}{%
    \ifcase#1 undefined
    \or b
    \or c
    \or d
    \else undefined \fi
}
\NewDocumentCommand\fml{O{1}}{%
    \ifcase#1 undefined
    \or \varphi
    \or \psi
    \or \rho
    \else undefined \fi
}
\newcommand{\val}{\mathfrak{v}}
\newcommand{\Progs}{\mathbf{W}} %

\knowledgenewrobustcmd\domain[1]{\cmdkl{|}#1\cmdkl{|}} %
\knowledgenewrobustcmd\src{\cmdkl{\const{1}}} %
\knowledgenewrobustcmd\tgt{\cmdkl{\const{2}}} %
\knowledgenewrobustcmd\homo{\mathrel{\cmdkl{\longrightarrow}}} %

\knowledgenewrobustcmd\vsig{\cmdkl{\mathbf{V}}} %
\knowledgenewrobustcmd\vsigpm{\cmdkl{\mathbf{V}}^{\pm}} %
\knowledgenewrobustcmd\sig{\cmdkl{\mathsf{S}}} %
\newcommand{\union}{\mathbin{+}}
\newcommand{\compo}{\mathbin{;}}
\DeclareMathOperator*{\bigcompo}{\scalerel*{\compo}{\sum}}
\newcommand{\kstar}{*}
\newcommand{\emp}{0}
\newcommand{\eps}{\const{I}}
\newcommand{\compl}{-}
\newcommand{\com}[1]{\overline{#1}}
\knowledgenewrobustcmd\lop{\cmdkl{\circlearrowleft}} %
\knowledgenewrobustcmd\dom{\cmdkl{\mathtt{\mathop{d}}}} %
\knowledgenewrobustcmd\adom{\cmdkl{\mathtt{\mathop{a}}}} %
\knowledgenewrobustcmd\PCoR{\cmdkl{\mathrm{PCoR}}}
\knowledgenewrobustcmd\termset{\cmdkl{\mathcal{T}_{}\kern-\scriptspace}} %
\knowledgenewrobustcmd\size[1]{\cmdkl{\|}#1\cmdkl{\|}} %

\knowledgenewrobustcmd\REL{\cmdkl{\mathsf{REL}}}
\knowledgenewrobustcmd\fin{\cmdkl{\mathrm{fin}}}
\knowledgenewrobustcmd\restr{\mathrel{\cmdkl{\restriction}}} %
\knowledgenewrobustcmd\gval{\cmdkl{\const{G}}} %
\knowledgenewrobustcmd\glang{\cmdkl{\mathcal{G}}} %
\newcommand\reallywidehat[1]{%
\savestack{\tmpbox}{\stretchto{%
  \scaleto{%
    \scalerel*[\widthof{\ensuremath{#1}}]{\kern-.6pt\bigwedge\kern-.6pt}%
    {\rule[-\textheight/2]{1ex}{\textheight}}%
  }{\textheight}%
}{0.5ex}}%
\stackon[1pt]{#1}{\tmpbox}%
}

\knowledgenewrobustcmd\psig{\cmdkl{\mathbf{P}}} %
\knowledgenewrobustcmd\asig{\cmdkl{\mathbf{A}}} %
\knowledgenewrobustcmd\ttest{\cmdkl{\mathtt{test}}}
\knowledgenewrobustcmd\Tests{\cmdkl{\mathbf{B}}} %
\knowledgenewrobustcmd\tcompl[1]{\cmdkl{\widetilde{#1}}} %
\knowledgenewrobustcmd\unionb[1]{\mathbin{\cmdkl{+}_{#1}}} %
\knowledgenewrobustcmd\starb[1]{\cmdkl{*}_{#1}} %
\knowledgenewrobustcmd\plusb[1]{\cmdkl{+}_{#1}} %
\knowledgenewrobustcmd\PWP{\cmdkl{\mathrm{PWP}}}
\knowledgenewrobustcmd\PWPlop{\cmdkl{\mathrm{PWP}}_{\lop}}
\knowledgenewrobustcmd\PWPcap{\cmdkl{\mathrm{PWP}}_{\cap}}
\knowledgenewrobustcmd\tfunc{\cmdkl{\mathtt{func}}}
\knowledgenewrobustcmd\DREL{\cmdkl{\mathsf{DREL}}}

\newcommand{\aE}{\mathtt{E}}
\newcommand{\aN}{\mathtt{N}}
\newcommand{\aW}{\mathtt{W}}
\newcommand{\aS}{\mathtt{S}}
\knowledgenewrobustcmd\tgrid{\cmdkl{\mathtt{grid}}}

\knowledgenewrobustcmd\Sur{\cmdkl{\mathsf{Sur}}}
\knowledgenewrobustcmd\Tr{\cmdkl{\mathrm{T}}}
\knowledgenewrobustcmd\TrPWP{\cmdkl{\mathrm{T}}'}
\newcommand{\tgr}{\mathtt{gr}}
\newcommand{\tcyc}{\mathtt{cyc}}

\knowledgenewrobustcmd\sdF{\cmdkl{\mathrm{F}}} %
\knowledgenewrobustcmd\sdP{\cmdkl{\mathrm{P}}} %
\knowledgenewrobustcmd\ctrue{\cmdkl{\const{t}}}
\knowledgenewrobustcmd\cfalse{\cmdkl{\const{f}}}

\newcommand{\tP}{\mathtt{P}}
\newcommand{\aA}{\mathtt{A}}
\newcommand{\aB}{\mathtt{B}}
\knowledgenewrobustcmd\Uorg{\cmdkl{U}}
\knowledgenewrobustcmd\Kgo[1]{\cmdkl{K_{}\kern-\scriptspace}_{#1}}
\knowledgenewrobustcmd\TrOne{\cmdkl{\mathrm{T}^{(1)}}}
\knowledgenewrobustcmd\sigmaAB{\cmdkl{\sigma}_{\aA \aB}}

\knowledgenewrobustcmd\trunc[1]{\cmdkl{\lfloor} #1 \cmdkl{\rfloor}}

\knowledgenewrobustcmd\interval[2]{\cmdkl{[}#1, #2\cmdkl{]}} %
\newcommand{\tQ}{\mathtt{Q}} %
\knowledgenewrobustcmd\TrbarI{\cmdkl{\mathrm{T}}^{\com{\eps}}} %
\knowledgenewrobustcmd\quo[1]{\cmdkl{[}#1\cmdkl{]}} %
\knowledgenewrobustcmd\GS{\cmdkl{\mathrm{GS}}} %

\tikzstyle{mynode} = [inner sep = 1.5pt, fill= gray!20,  font = \scriptsize, draw, circle]
\tikzstyle{mysmallnode} = [inner sep = 1.pt, fill= gray!20, font = \scriptsize, draw, circle]
\tikzset{earrow/.style={>={{[flex] Latex[length=.08cm, width=2.5pt]}}}}
\tikzset{homoarrow/.style={earrow, line width = .8pt, densely dotted, color = olive, opacity=0.8}}

\stackMath

\knowledge{ignore, no index}
| letter
| letters

\knowledge{ignore, no index}
| word
| words

\knowledge{ignore, no index}
| empty word

\knowledge{ignore, no index}
| length

\knowledge{ignore, no index}
| concatenation

\knowledge{ignore, no index}
| Kleene star

\knowledge{notion}
| language
| languages

\knowledge{ignore, no index}
| binary relation
| binary relations

\knowledge{ignore, no index}
| converse

\knowledge{ignore, no index}
| composition
| (relational) composition

\knowledge{ignore, no index}
| reflexive transitive closure

\knowledge{ignore, no index}
| $n$-th iteration

\knowledge{notion}
| surjective

\knowledge{ignore, no index}
| variable
| variables

\knowledge{notion}
| term
| terms

\knowledge{notion}
| $\sig$-algebra
| $\sig$-algebras

\knowledge{notion}
| equation
| equations

\knowledge{notion}
| inequation
| inequations

\knowledge{notion}
| fresh

\knowledge{notion}
| relational model
| relational models

\knowledge{notion}
| valuation
| valuations

\knowledge{notion}
| equational theory (with respect to relations)
| equational theory with respect to relations
| equational theory
| equational theories

\knowledge{notion}
| submodel
| submodels

\knowledge{notion}
| submodel-closed

\knowledge{notion}
| variable complement
| variable complements

\knowledge{notion}
| difference constant

\knowledge{notion}
| domain operator

\knowledge{notion}
| antidomain operator

\knowledge{notion}
| hypothesis
| hypotheses

\knowledge{notion}
| inclusion problem

\knowledge{notion}
| emptiness problem

\knowledge{ignore, no index}
| validity problem

\knowledge{ignore, no index}
| finite validity problem

\knowledge{notion}
| graph
| graphs
| ($2$-pointed) graph
| ($2$-pointed) graphs
| Graphs

\knowledge{notion}
| vertex
| vertices

\knowledge{notion}
| source
| sources

\knowledge{notion}
| target
| targets

\knowledge{ignore, no index}
| series-composition

\knowledge{ignore, no index}
| parallel-composition

\knowledge{notion}
| graph language
| graph languages

\knowledge{notion}
| graph homomorphism
| graph homomorphisms
| graph homomorphic

\knowledge{notion}
| domino system
| domino systems

\knowledge{notion}
| tiling
| tilings

\knowledge{notion}
| periodic tiling
| periodic tilings
| periodic
| $(\tuple{h,v}$-$)$periodic
| $\tuple{h,v}$-periodic
| $\tuple{h,v}$-periodic tiling
| $(\tuple{h,v}$-$)$periodic tiling

\knowledge{notion}
| domino problem
| the domino problem
| The domino problem

\knowledge{notion}
| periodic domino problem
| the periodic domino problem
| The periodic domino problem

\knowledge{notion}
| bounded periodic domino problem

\knowledge{notion}
| formula
| formulas
| (quantifier-free) formula
| (quantifier-free) formulas
| quantifier-free formula
| quantifier-free formulas

\knowledge{notion}
| test
| tests
| Test
| Tests

\knowledge{notion}
| primitive test
| primitive tests

\knowledge{notion}
| action
| actions

\knowledge{notion}
| propositional while program
| propositional while programs

\knowledge{notion}
| graph loop operator
| (graph) loop operator
| graph loop
| (graph) loop
| graph loops
| (graph) loops
| loop operator
| loop
| loops

\knowledge{notion}
| propositional while programs with graph loop
| propositional while programs with loop
| $\PWPlop$ terms
| $\PWPlop$ term

\knowledge{notion}
| propositional while programs with intersection

\knowledge{notion}
| hypothesis elimination
| hypothesis eliminations
| Hypothesis elimination
| Hypothesis eliminations

\knowledge{notion}
| witness-basis

\knowledge{notion}
| Kleene algebra
| Kleene algebras
| KA

\knowledge{notion}
| Hoare hypothesis
| Hoare hypotheses

\knowledge{notion}
| substitution hypothesis
| substitution hypotheses

\knowledge{notion}
| loop hypothesis
| loop hypotheses

\knowledge{notion}
| loop extension
| loop extensions
| (graph) loop extension

\knowledge{notion}
| SDIPDL

\knowledge{notion}
| primitive

\knowledge{notion}
| primitive propositional while programs with graph loop
| primitive propositional while programs with loop
| primitive $\PWPlop$ terms
| primitive $\PWPlop$ term

\knowledge{notion}
| pseudoforest

\knowledge{notion}
| treewidth

\knowledge{notion}
| consistent

\knowledge{notion}
| Kleene lattice
| Kleene lattices
| KL

\knowledge{notion}
| Kleene allegory
| Kleene allegories
| KAll

\knowledge{notion}
| interval
| intervals

\knowledge{notion}
| interval-closed

\knowledge{notion}
| compatible

\knowledge{notion}
| quotient
| quotients

\knowledge{notion}
| regular program
| regular programs

\usetikzlibrary{automata, positioning, arrows, decorations.pathreplacing}

\begin{document}

\title{Undecidability of the Emptiness Problem of Deterministic Propositional While Programs with Graph Loop:
Hypothesis Elimination Using Loops}

\address{nakamura.yoshiki.ny{@}gmail.com}

\author{Yoshiki Nakamura\thanks{%
This work was supported by JSPS KAKENHI Grant Numbers JP21K13828, JP25K14985.}\\
  Department of Mathematics and Informatics\\
  Chiba University\\
  nakamura.yoshiki.ny{@}gmail.com
}

\maketitle

\runninghead{Y. Nakamura}{Undecidability of Deterministic Propositional While Programs with Loop}

\begin{abstract}
We show that the emptiness (unsatisfiability) problem is undecidable and
$\mathrm{\Pi}^{0}_{1}$-complete for deterministic propositional while programs with (graph) loop.
To this end, we introduce a hypothesis elimination procedure in propositional while programs using loops.
Using this, we give reductions from the complement of the periodic domino problem,
where the numbers of actions and primitive tests can be reduced by two and one, respectively.

Moreover, as a corollary via hypothesis eliminations,
we also show that the equational theory is $\mathrm{\Pi}^{0}_{1}$-complete for the positive calculus of relations with transitive closure and difference.
Additionally, we show that the emptiness problem (on the class of all relational models) is \textsc{PSpace}-complete for the existential calculus of relations with transitive closure.

 \end{abstract}


\section{Introduction}\label{section:introduction}
\emph{The calculus of relations} (CoR, for short) \cite{tarskiCalculusRelations1941} is an algebraic system with operators on binary relations.
Here, the operators consist of identity ($\eps$), empty ($\emp$), universality ($\top$), composition ($\compo$), converse ($\bl^{\smile}$), union ($\union$), intersection ($\cap$), and complement ($\bl^{-}$).
While the "equational theory" of CoR is \emph{undecidable} \cite{tarskiCalculusRelations1941,tarskiFormalizationSetTheory1987} (see also, e.g., \cite{hirschDecidabilityEquationalTheories2018,nakamuraUndecidabilityFO3Calculus2019,hirschUndecidabilityAlgebrasBinary2021} for more undecidable fragments),
one approach to avoid the undecidability is to consider its (existential) \emph{positive} fragment \cite{andrekaEquationalTheoryUnionfree1995,pousPositiveCalculusRelations2018} ($\PCoR$, for short) by excluding complement.
The "equational theory" of $\PCoR$ is \emph{decidable} \cite{andrekaEquationalTheoryUnionfree1995} and this decidability result also holds when adding the reflexive transitive closure operator ($\bl^{*}$) \cite{brunetPetriAutomataKleene2015,brunetPetriAutomata2017,nakamuraPartialDerivativesGraphs2017,nakamuraDerivativesGraphsPositive2025}.

In \cite{nakamuraExistentialCalculiRelations2023},
the \emph{existential}\footnote{With respect to binary relations, $\PCoR$ and $\PCoR_{\set{\com{\eps}, \com{x}}}$ have the same expressive power as the three-variable fragment of existential \emph{positive} logic with equality and that of \emph{existential} logic with equality \cite{nakamuraExpressivePowerSuccinctness2022,nakamuraExistentialCalculiRelations2023}, respectively.} fragment $\PCoR_{\set{\com{\eps}, \com{x}}}$ is considered: $\PCoR$ with complements of term variables ($\com{x}$) and the difference constant ($\com{\eps}$), i.e., the complement operator only applies to term variables or constants.
In this fragment, complement is allowed only in a restricted way, so that the complexity does not increase drastically.
The "equational theory" of $\PCoR_{\set{\com{\eps}, \com{x}}}$ with transitive closure (denoted by $\PCoR_{\set{\bl^{*}, \com{\eps}, \com{x}}}$) is \emph{undecidable} and $\PiOne$-complete in general (but decidable for its $\cap$-free fragments) \cite{nakamuraExistentialCalculiRelations2023}.
The undecidability result still holds for the $\com{\eps}$-free fragment $\PCoR_{\set{\bl^{*}, \com{x}}}$ \cite[Theorem 50]{nakamuraExistentialCalculiRelations2023}.
However, to our knowledge, it was open whether the "equational theory" is decidable for the $\com{x}$-free fragment $\PCoR_{\set{\bl^{*}, \com{\eps}}}$, cf.\ \cite[Table I]{nakamuraExistentialCalculiRelations2023}.
The known bounds were as follows: the "equational theory" is $\EXPSPACE$-hard \cite{nakamuraPartialDerivativesGraphs2017,brunetPetriAutomata2017} (by a reduction from the universality problem of regular expressions with intersection \cite{furerComplexityInequivalenceProblem1980}) and in $\PiOne$ \cite[Lemma 22]{nakamuraExistentialCalculiRelations2023} (by the finite model property).

\subsection*{Contributions}
We show that the "equational theory" of $\PCoR_{\set{\bl^{*}, \com{\eps}}}$ is undecidable and $\PiOne$-complete (\Cref{cor:undecidable-PCoR*-difference}).
More strongly,
we show that the emptiness problem of \emph{deterministic} \emph{"propositional while programs with loop"} is also $\PiOne$-complete (\Cref{thm:undecidable-while-emptiness}).
Notably, this result also shows that the "validity problem" is $\PiOne$-complete for formulas of the form $[\term]\cfalse$ in "SDIPDL" \cite{goldblattWellstructuredProgramEquivalence2012} (strict deterministic propositional dynamic logic (SDPDL) \cite{halpernPropositionalDynamicLogic1983} with intersection),
where $\term$ is a "$\PWPlop$ term" containing no modal operators (\Cref{cor:undecidable-PDL} and \Cref{table:complexity-SDIPDL}).
This refines the undecidability result of "SDIPDL" \cite{goldblattWellstructuredProgramEquivalence2012} with respect to the nesting of modal operators (\Cref{rem:nested-modal-operators,rem:difference}).

To this end, we give reductions from the complement of "the periodic domino problem" (of $\Z^2$).
Our reductions are based on the reduction of Goldblatt and Jackson \cite{goldblattWellstructuredProgramEquivalence2012},
where we refine this reduction by a novel hypothesis elimination using "graph loops" for hypotheses of the form $\eps \le \term[3]$.
Furthermore, the "emptiness problem" remains $\PiOne$-complete even with $2$ "actions" and $1$ "primitive test" (\Cref{cor:reducing-primitive-tests-conv}),
whereas it is decidable when there is at most one "action" or no "primitive test" (\Cref{cor:reducing-primitive-tests-2}).

Moreover, we pinpoint the complexity of the "emptiness problem" with respect to $\REL$ and $\DREL$ for fragments of $\PCoR_{\set{\bl^{*}, \com{\eps}, \com{x}}}$ "terms"; they are the classes of all relational models and of all \emph{deterministic} relational models, i.e., those in which every "variable" is interpreted as a partial function (see \Cref{section:semantics,section:PWP-with-loop}).
In particular, we show that the "emptiness problem" is decidable and $\PSPACE$-complete for $\PCoR_{\set{\bl^{*}, \com{\eps}, \com{x}}}$ (on $\REL$) (\Cref{thm:decidable-ECoR*-emptiness}).

\subsection*{Differences from the Conference Version}
This paper is an extended and revised version of the paper presented at the 21st International Conference on Relational and Algebraic Methods in Computer Science (RAMiCS 2024) \cite{nakamuraUndecidabilityPositiveCalculus2024}.
The main differences from the conference version are as follows.
\begin{enumerate}
    \item
    In \Cref{section:hypothesis-eliminations},
    we give a collection of "hypothesis eliminations" (under $\REL$).

    \item 
    In \Cref{section:Hypothesis-elimination-using-graph-loops},
    we newly give some sufficient conditions for hypothesis eliminations using "loops".
    Using them, we will show undecidability results.

    \item
    In this version,
    we show the $\PiOne$-hardness of the "equational theory" of $\PCoR_{\set{\bl^{*}, \com{\eps}}}$ (\Cref{cor:undecidable-PCoR*-difference}) 
    from that of the "emptiness problem" of "propositional while programs with loop" (\Cref{thm:undecidable-while-emptiness}) via hypothesis eliminations.
    (In the conference version, they were shown independently, although similar reductions were used.)
     
    \item
    \Cref{section:on-the-number-of-primitive-tests-and-actions} is new.
    We strengthen the undecidability result of the "emptiness problem" of "propositional while programs with loop" by fixing the numbers of "primitive tests" and "actions" (\Cref{thm:reducing-primitive-tests,cor:reducing-primitive-tests-conv}).

    \item
    \Cref{section:emptiness} is new.
    We pinpoint the complexity of the "emptiness problem" with respect to $\REL$ and $\DREL$ for fragments of $\PCoR_{\set{\bl^{*}, \com{\eps}, \com{x}}}$.
\end{enumerate}

\subsection*{Related Work}
\paragraph*{SDIPDL}
Propositional dynamic logic (PDL) with intersection is $\TWOEXPTIME$-complete \cite{daneckiNondeterministicPropositionalDynamic1984,lange2ExpTimeLowerBounds2005}
and strict deterministic PDL (SDPDL) is $\PSPACE$-complete \cite{halpernPropositionalDynamicLogic1983}.
The decidability of "SDIPDL" (SDPDL with intersection) was posed as an open problem in the Stanford Encyclopedia of Philosophy \cite{balbianiPropositionalDynamicLogic2008}
and was settled by Goldblatt and Jackson \cite[Theorem 4.1]{goldblattWellstructuredProgramEquivalence2012}.
The "validity problem" is $\PiOneOne$-complete and the "finite validity problem" is $\PiOne$-complete (see \Cref{rem:difference}).
Their reduction uses formulas with nested box and diamond operators.
In contrast, our undecidability result uses only formulas of the form $[\term]\cfalse$ given above (\Cref{cor:undecidable-PDL}).

\paragraph*{GKAT}
GKAT (guarded Kleene algebra with tests) expressions \cite{smolkaGuardedKleeneAlgebra2019} whose tests are in negation normal form are exactly "propositional while programs" (by reading $\tilde{p}$ as the negation of $p$);
since every test can be transformed into this form by De Morgan's laws, we only consider such tests for simplicity of the presentation.
Their "equational theory" on $\DREL_{\ttest}$ coincides with that on $\REL_{\ttest}$ (\Cref{section:DREL-KAT}) and is decidable
(in nearly linear time when the number of "primitive tests" is fixed \cite{smolkaGuardedKleeneAlgebra2019}; see also \Cref{table:complexity-SDIPDL}).
With the "graph loop operator", the "emptiness problem" on $\DREL_{\ttest}$ is $\PiOne$-complete,
even when the numbers of "actions" and "primitive tests" are fixed to $2$ and $1$, respectively (\Cref{thm:reducing-primitive-tests,cor:reducing-primitive-tests-conv}).

\paragraph*{Hypothesis elimination}
In CoR, hypotheses can be eliminated using complement \cite[p.\ 86]{tarskiCalculusRelations1941} (\Cref{prop:Hoare-hypotheses}); see also \cite{nakamuraUndecidabilityFO3Calculus2019,nakamuraExistentialCalculiRelations2023} for "hypothesis eliminations" of CoR fragments.
"Hypothesis eliminations" have also been studied for Kleene algebras and Kleene algebras with tests \cite{cohenHypothesesKleeneAlgebra1994,doumaneKleeneAlgebraHypotheses2019,kozenHoareLogicKleene2000,hardinEliminationHypothesesKleene2002,hardinProofTheoryKleene2005,hardinModularizingElimination$r0$2005,kozenKleeneAlgebraEquations2014};
see \cite{pousToolsCompletenessKleene2024} for a general framework.
Some of them are also applicable under $\REL$ (\Cref{rem:Hoare-KA}), but some are not (\Cref{rem:id}).
A translation similar to $\Tr_{\term[3]}$ (\Cref{defi:loop-transformation}) is used in \cite{cohenHypothesesKleeneAlgebra1994} for eliminating $x \le \eps$ in general Kleene algebras.
For the encoding of "tests" in automata constructions, see \cite{nakamuraDerivativesGraphsPositive2025} (cf.\ \Cref{ex:tests}).
The elimination of "Hoare hypotheses" used in \cite[Theorem 50]{nakamuraExistentialCalculiRelations2023} is also useful to show the undecidability with respect to language models \cite{nakamuraFiniteRelationalSemantics2025}.

\subsection*{Organization}
In \Cref{section:preliminaries}, we give basic definitions.
In \Cref{section:grids}, we give an encoding of periodic "tilings" based on \cite{goldblattWellstructuredProgramEquivalence2012}.
In \Cref{section:hypothesis-eliminations},
we give a collection of "hypothesis eliminations" under $\REL$.
In \Cref{section:Hypothesis-elimination-using-graph-loops},
we introduce a "hypothesis elimination" using "loops" and give its two sufficient conditions.
In \Cref{section:while}, we show that the emptiness problem is $\PiOne$-complete for deterministic "propositional while programs with loop" (\Cref{thm:undecidable-while-emptiness}).
In \Cref{section:on-the-number-of-primitive-tests-and-actions}, we show that the emptiness problem is still $\PiOne$-complete even if the numbers of "actions" and "primitive tests" are fixed.
In \Cref{section:undecidable-PCoR*-with-difference}, we show that the "equational theory" of $\PCoR_{\set{\bl^{*}, \com{\eps}}}$ is $\PiOne$-complete (\Cref{cor:undecidable-PCoR*-difference}).
In \Cref{section:emptiness}, we consider the "emptiness problem" on $\REL$ and $\DREL$ for fragments of $\PCoR_{\set{\bl^{*}, \com{\eps}, \com{x}}}$.
In \Cref{section:conclusion}, we conclude this paper.
\section{Preliminaries}\label{section:preliminaries}
We write $\Z$ and $\nat$ for the sets of integers and non-negative integers, respectively.
For $l, r \in \Z$, we write $\intro*\range{l, r}$ for the set $\{i \in \Z \mid l \le i \le r\}$.
For a set $X$, we write
$\intro*\card{X}$ for the cardinality of $X$ and
$\wp(X)$ for the power set of $X$.

We write $\triangle_{A}$ for the identity relation on a set $A$: $\triangle_{A} \defeq \set{\tuple{x, x} \mid x \in A}$.
For "binary relations" $\rel[1]$ and $\rel[2]$ on a set $B$, the ""converse"" $\rel[1]^{\smile}$, the ""composition"" $\rel[1] \compo \rel[2]$, the ""$n$-th iteration"" $\rel[1]^{n}$ (where $n \in \nat$), and the ""reflexive transitive closure"" $\rel[1]^*$ are defined by:
\begin{align*}
    \rel[1] \compo \rel[2] & \defeq \set{\tuple{x, z} \mid \exists y,\ \tuple{x, y} \in \rel[1] \ \land \ \tuple{y, z} \in \rel[2]}, & \rel[1]^{\smile} & \defeq \set{\tuple{x, y} \mid \tuple{y, x} \in \rel[1]}, \\
    \rel[1]^n             & \defeq \begin{cases}
                                       \rel[1] \compo \rel[1]^{n-1} & (n \ge 1) \\
                                       \triangle_{B}               & (n = 0)
                                   \end{cases},                                                             & \rel[1]^*        & \defeq \bigcup_{n \in \nat} \rel[1]^n.
\end{align*}

A ""($2$-pointed) graph"" $\graph$ over a set $L$ (of labels) is a tuple $\tuple{\intro*\domain{\graph}, \set{a^{\graph}}_{a \in L}, \intro*\src^{\graph}, \intro*\tgt^{\graph}}$, where
$\domain{\graph}$ is a non-empty set (of ""vertices""),
each $a^{\graph} \subseteq \domain{\graph}^{2}$ is a binary relation (the set of $a$-labelled edges), and
$\src^{\graph}, \tgt^{\graph} \in \domain{\graph}$ are the ""source"" and ""target"" "vertices".
Let $\graph[1]$ and $\graph[2]$ be "graphs" over a set $L$.
For a map $f \colon \domain{\graph[1]} \to \domain{\graph[2]}$,
we say that $f$ is a ""graph homomorphism"" from $\graph[1]$ to $\graph[2]$, written $f \colon \graph \intro*\homo \graph[2]$, if
$f(\src^{\graph[1]}) = \src^{\graph[2]}$, $f(\tgt^{\graph[1]}) = \tgt^{\graph[2]}$, and
for all $x$, $y$, and $a$, $\tuple{x, y} \in a^{\graph[1]}$ implies $\tuple{f(x), f(y)} \in a^{\graph[2]}$.
We say that $f$ is ""surjective"" if $f$ is surjective as a map from $\domain{\graph[1]}$ to $\domain{\graph[2]}$.

\subsection{Syntax: PCoR with Transitive Closure and Complement}\label{section:syntax}
Let $\intro*\vsig$ be a set of ""variables"".
We consider ""terms"" over the signature
\[\intro*\sig \defeq \set{\eps_{(0)}, \emp_{(0)}, \compo_{(2)}, \union_{(2)}, {\bl^{\smile}}_{(1)}, {\bl^{\kstar}}_{(1)}, {\bl^{\compl}}_{(1)}}.\]
We use parentheses in ambiguous situations.
We often abbreviate $\term[1] \compo \term[2]$ to $\term[1] \term[2]$.
For a sequence $\term[1]_1, \dots, \term[1]_n$ of "terms",
we write $\sum_{i = 1}^{n} \term[1]_i$ (or we write $\sum X$ where $X = \set{\term[1]_1, \dots, \term[1]_n}$ as the ordering is not important)
for the "term" $\term[1]_1 \union \dots \union \term[1]_n$ (or $\emp$ if $n = 0$)
and write $\bigcompo_{i = 1}^{n} \term[1]_i$
for the "term" $\term[1]_1 \compo \dots \compo \term[1]_n$ (or $\eps$ if $n = 0$).
Similarly, for $n \ge 1$, we write $\bigcap_{i = 1}^{n} \term[1]_i$ for the "term" $\term[1]_1 \cap \dots \cap \term[1]_n$.
For $n \in \nat$, we write $\term[1]^{n}$ for the "term" $\bigcompo_{i = 1}^{n} \term[1]$.
For a "term" $\term$ over $\sig$, let $\com{\term}$ be $\term[2]$ if $\term = \term[2]^{-}$ for some $\term[2]$ and be $\term[1]^{-}$ otherwise.
We use the following abbreviations:
\begin{align*}
    \top                   & \defeq \emp^{-},                               &
    \term[1] \cap \term[2] & \defeq (\term[1]^{-} \union \term[2]^{-})^{-}, &
    \term[1]^{+}           & \defeq \term \term[1]^*,                       &
    \term^{\lop}           & \defeq \term \cap \eps,                        &
    \term^{\dom}           & \defeq (\term \top)^{\lop},&
    \term^{\adom}           & \defeq ((\term \top)^{-})^{\lop}.
\end{align*}
Here, $\term^{\intro*\lop}$ expresses the ""(graph) loop operator"" (strong loop predicate) \cite{daneckiPropositionalDynamicLogic1984},
$\term^{\intro*\dom}$ expresses the ""domain operator"" \cite{desharnaisKleeneAlgebraDomain2006},
and $\term^{\intro*\adom}$ expresses the ""antidomain operator"" \cite{desharnaisInternalAxiomsDomain2011}.
We call $\com{\eps}$ the ""difference constant"" and call $\com{x}$ (where $x \in \vsig$) ""variable complements"".

For a set $X \subseteq \set{\eps, \emp, \top, \compo, \union, {\bl^{\smile}}, {\bl^{*}}, \cap, {\bl^{\lop}}, \com{\eps}, \com{x}, \bl^{-}, {\bl^{\dom}}, {\bl^{\adom}}}$,
the set of $X$-"terms", written $\intro*\termset_{X}$, is defined as the minimal subset of the set of "terms" over $\sig$ satisfying the following:
\begin{gather*}
    \begin{prooftree}
        \hypo{y \in \vsig}
        \infer1{y \in \termset_{X}}
    \end{prooftree}
    \quad
    \begin{prooftree}
        \hypo{\eps \in X}
        \infer1{\eps \in \termset_{X}}
    \end{prooftree}
    \quad
    \begin{prooftree}
        \hypo{\emp \in X}
        \infer1{\emp \in \termset_{X}}
    \end{prooftree}
    \quad
    \begin{prooftree}
        \hypo{\top \in X}
        \infer1{\top \in \termset_{X}}
    \end{prooftree}
    \quad
    \begin{prooftree}
        \hypo{\compo \in X}
        \hypo{\term[1] \in \termset_{X}}
        \hypo{\term[2] \in \termset_{X}}
        \infer3{\term[1] \compo \term[2] \in \termset_{X}}
    \end{prooftree}
    \quad
    \begin{prooftree}
        \hypo{\union \in X}
        \hypo{\term[1] \in \termset_{X}}
        \hypo{\term[2] \in \termset_{X}}
        \infer3{\term[1] \union \term[2] \in \termset_{X}}
    \end{prooftree}\\
    \begin{prooftree}
        \hypo{\bl^{\smile} \in X}
        \hypo{\term[1] \in \termset_{X}}
        \infer2{\term[1]^{\smile} \in \termset_{X}}
    \end{prooftree}
    \quad
    \begin{prooftree}
        \hypo{\bl^{*} \in X}
        \hypo{\term[1] \in \termset_{X}}
        \infer2{\term[1]^{\kstar} \in \termset_{X}}
    \end{prooftree}
    \quad
    \begin{prooftree}
        \hypo{\cap \in X}
        \hypo{\term[1] \in \termset_{X}}
        \hypo{\term[2] \in \termset_{X}}
        \infer3{\term[1] \cap \term[2] \in \termset_{X}}
    \end{prooftree}
    \quad
    \begin{prooftree}
        \hypo{\bl^{\lop} \in X}
        \hypo{\term[1] \in \termset_{X}}
        \infer2{\term[1]^{\lop} \in \termset_{X}}
    \end{prooftree}
    \\
    \begin{prooftree}
        \hypo{\com{\eps} \in X}
        \infer1{\com{\eps} \in \termset_{X}}
    \end{prooftree}
    \quad
    \begin{prooftree}
        \hypo{\com{x} \in X}
        \hypo{y \in \vsig}
        \infer2{\com{y} \in \termset_{X}}
    \end{prooftree}
    \quad
    \begin{prooftree}
        \hypo{\bl^{-} \in X}
        \hypo{\term[1] \in \termset_{X}}
        \infer2{\term[1]^{-} \in \termset_{X}}
    \end{prooftree}
    \quad
    \begin{prooftree}
        \hypo{\bl^{\dom} \in X}
        \hypo{\term[1] \in \termset_{X}}
        \infer2{\term[1]^{\dom} \in \termset_{X}}
    \end{prooftree}
    \quad
    \begin{prooftree}
        \hypo{\bl^{\adom} \in X}
        \hypo{\term[1] \in \termset_{X}}
        \infer2{\term[1]^{\adom} \in \termset_{X}}
    \end{prooftree}
\end{gather*}
(In particular, $\termset_{\sig}$ coincides with the set of "terms" over $\sig$.)
We also let $\intro*\PCoR \defeq \termset_{\set{\eps, \emp, \top, \compo, \union, \cap, \bl^{\smile}}}$,
and let $\PCoR_{X} \defeq \termset_{X \cup \set{\eps, \emp, \top, \compo, \union, \cap, \bl^{\smile}}}$ for $X \subseteq \set{\bl^{*}, \com{\eps}, \com{x}, \bl^{-}}$.

An ""equation"" $\term[1] = \term[2]$ is a pair of "terms".
An ""inequation"" $\term[1] \le \term[2]$ is an abbreviation of the "equation" $\term[1] \union \term[2] = \term[2]$.
The set of ""(quantifier-free) formulas"" of a "term" class $\mathcal{T}$ is defined by:
\begin{align*}
    \fml[1], \fml[2] \quad::=\quad \term[1] = \term[2] \mid \fml[1] \land \fml[2] \mid \lnot \fml[1]. \tag{$\term[1], \term[2] \in \mathcal{T}$} 
\end{align*}
We use the following abbreviations, as usual: $\fml[1] \lor \fml[2] \defeq \lnot (\lnot \fml[1] \land \lnot \fml[2])$,
$\fml[1] \to \fml[2] \defeq (\lnot \fml[1]) \lor \fml[2]$, and $\fml[1] \leftrightarrow \fml[2] \defeq (\fml[1] \to \fml[2]) \land (\fml[2] \to \fml[1])$.

Additionally, we write $\intro*\size{\term[1]}$ for the number of symbols occurring in a "term" $\term[1]$.

\subsection{Semantics: Relational Models}\label{section:semantics}
An ""$\sig$-algebra"" $\alg$ is a tuple $\tuple{\domain{\alg}, \set{f^{\alg}}_{f_{(k)} \in \sig}}$, where $\domain{\alg}$ is a non-empty set and $f^{\alg} \colon \domain{\alg}^{k} \to \domain{\alg}$ is a $k$-ary map for each $f_{(k)} \in \sig$.
A ""valuation"" $\val$ of an "$\sig$-algebra" $\alg$ is a map $\val \colon \vsig \to \domain{\alg}$.
For a "valuation" $\val$, we write $\hat{\val} \colon \termset_{\sig} \to \domain{\alg}$ for the unique homomorphism extending $\val$.
For a "formula" $\fml$, we define $\hat{\val}(\fml) \in \set{\const{false}, \const{true}}$ as follows:
\begin{align*}
    \hat{\val}(\term[1] = \term[2]) & \hspace{.15em} \defiff\hspace{.15em} (\hat{\val}(\term[1]) = \hat{\val}(\term[2])),&
    \hat{\val}(\fml[1] \land \fml[2]) &\hspace{.15em} \defiff\hspace{.15em}  (\hat{\val}(\fml[1]) \mbox{ and } \hat{\val}(\fml[2])), &
    \hat{\val}(\lnot \fml[1]) &\hspace{.15em} \defiff\hspace{.15em}  (\mbox{not } \hat{\val}(\fml[1])).
\end{align*}

For a "formula" $\fml$ and a class of "valuations" (of "$\sig$-algebras") $\algclass$, we write
\begin{align*}
    \algclass \models \fml & \quad \defiff \quad \hat{\val}(\fml) \mbox{ holds for all "valuations" $\val \in \algclass$}.
\end{align*}
We abbreviate $\set{\val} \models \fml$ to $\val \models \fml$.
We write $\algclass_{\fml}$ for the class $\set{\val \in \algclass \mid \val \models \fml}$.
For a set $\Gamma$ of "formulas", we write $\algclass_{\Gamma}$ for the class $\bigcap_{\fml \in \Gamma} \algclass_{\fml}$.
When we consider a class of the form $\algclass_{\Gamma}$, we often call the "formulas" in $\Gamma$ ""hypotheses""
(i.e., $\algclass_{\Gamma} \models \fml$ means that $\fml$ holds in $\algclass$ under the "hypotheses" $\Gamma$).
We also write $\algclass_{\Gamma, \fml}$ for $\algclass_{\Gamma \cup \set{\fml}}$ and write $\algclass_{\Gamma, \Delta}$ for $\algclass_{\Gamma \cup \Delta}$.

The ""relational model"" $\alg$ over a non-empty set $B$ is the "$\sig$-algebra" such that
$\domain{\alg} = \wp(B^2)$,
$\eps^{\alg} = \triangle_{B}$,
$\emp^{\alg} = \emptyset$,
and for all $\rel[1], \rel[2] \subseteq B^2$,
\begin{flalign*}
               \rel[1] \compo^{\alg} \rel[2] & = \rel[1] \compo \rel[2], &\hspace{-.4em}
               \rel[1] \union^{\alg} \rel[2] & = \rel[1] \cup \rel[2], &\hspace{-.4em}
               \rel[1]^{\kstar^{\alg}}       & = \rel[1]^{\kstar},       &\hspace{-.4em}
               R^{\smile^{\alg}} & = R^{\smile}, &\hspace{-.4em}
               \rel[1]^{\compl^{\alg}}       & = B^2 \setminus \rel[1].
\end{flalign*}
We write $\intro*\REL$ $[\REL^{\intro*\fin}]$ for the class of all "valuations" of "relational models" [over a finite set].
For $\algclass \subseteq \REL$, we write $\algclass^{\fin}$ for the class $\algclass \cap \REL^{\fin}$.
For a class $\algclass \subseteq \REL$ and $X \subseteq \termset_{\sig}$,
the ""equational theory"" of $X$ with respect to $\algclass$ is the set
$\set{\term[1] = \term[2]  \in X^2 \mid \algclass \models \term[1] = \term[2]}$.
Additionally, for a "valuation" $\val \colon \vsig \to \wp(A^2)$ in $\REL$ and a non-empty subset $B \subseteq A$,
the ""submodel"" of $\val$ with respect to $B$ is the "valuation" in $\REL$ defined by:
\begin{align*}
    \val \introalign\restr B \colon \vsig &\to \wp(B^2)\\
                             x &\mapsto \val(x) \cap B^2.
\end{align*}
We say that $\algclass \subseteq \REL$ is ""submodel-closed"" if $\val \restr B \in \algclass$ for all $\val \colon \vsig \to \wp(A^2)$ in $\algclass$ and non-empty $B \subseteq A$.
For a "term" $\term$, we write $\vsig(\term)$ for the set of "variables" occurring in $\term$.
Similarly, for a set $\Gamma$ of "formulas", we write $\vsig(\Gamma)$ for the set of "variables" occurring in $\Gamma$.
For a class $\algclass \subseteq \REL$, we say that a "variable" $z$ is ""fresh"" in $\algclass$ if,
for all $\val \colon \vsig \to \wp(A^2)$ in $\algclass$ and $\val' \colon \vsig \to \wp(A^2)$ such that $\val'(y) = \val(y)$ for all $y \neq z$, we have $\val' \in \algclass$
(for example, if $z \not\in \vsig(\Gamma)$, then $z$ is "fresh" in $\REL_{\Gamma}$).

\subsubsection{Alternative semantics using graph languages}\label{section:alternative-semantics}
For $\PCoR_{\set{\bl^{*}, \com{\eps}, \com{x}}}$, we can use an alternative semantics using "graph homomorphisms" \cite{nakamuraExistentialCalculiRelations2023} (see also, e.g., \cite{andrekaEquationalTheoryUnionfree1995,brunetPetriAutomata2017,pousPositiveCalculusRelations2018} for fragments of $\PCoR_{\set{\bl^{*}, \com{\eps}, \com{x}}}$).
Let $\intro*\vsigpm \defeq \set{x, \com{x} \mid x \in \vsig} \cup \set{\com{\eps}}$ and let $\vsigpm_{\eps} \defeq \vsigpm \cup \set{\eps}$.
The ""graph language"" $\intro*\glang(\term)$ of a $\PCoR_{\set{\bl^{*}, \com{\eps}, \com{x}}}$ "term" $\term$ is the set of "graphs" over $\vsigpm_{\eps}$ defined by:
\kldeactivatequotation
\begin{align*}
    \hspace{-1.em} \glang(x)                                                                                                      & \defeq \set{
        \begin{tikzpicture}[baseline = -.5ex]
            \graph[grow right = 1.cm, branch down = 6ex, nodes={mynode}]{
            {0/{}[draw, circle]}-!-{1/{}[draw, circle]}
            };
            \node[left = .5em of 0](l){};
            \node[right = .5em of 1](r){};
            \graph[use existing nodes, edges={color=black, pos = .5, earrow}, edge quotes={fill=white, inner sep=1pt,font= \scriptsize}]{
            0 ->["$x$"] 1;
            l -> 0; 1 -> r;
            };
        \end{tikzpicture}
    }  \mbox{ \mbox{ for $x \in \vsigpm$}},                                          & \glang(\emp)  & \defeq \emptyset, \quad \glang(\top) \defeq \set{
        \begin{tikzpicture}[baseline = -.5ex]
            \graph[grow right = .6cm, branch down = 6ex, nodes={mynode}]{
            {0/{}[draw, circle]}-!-{1/{}[draw, circle]}
            };
            \node[left = .5em of 0](l){};
            \node[right = .5em of 1](r){};
            \graph[use existing nodes, edges={color=black, pos = .5, earrow}, edge quotes={fill=white, inner sep=1pt,font= \scriptsize}]{
            l -> 0; 1 -> r;
            };
        \end{tikzpicture}
    },                                                                                       \\
    \glang(\term[1]^{\smile})                                                                                                     & \defeq \set{
        \begin{tikzpicture}[baseline = -.5ex]
            \graph[grow right = 1.cm, branch down = 6ex, nodes={mynode}]{
            {0/{}[draw, circle]}-!-{1/{}[draw, circle]}
            };
            \node[left = .5em of 0](l){};
            \node[right = .5em of 1](r){};
            \graph[use existing nodes, edges={color=black, pos = .5, earrow}, edge quotes={fill=white, inner sep=1pt,font= \scriptsize}]{
            1 ->["$\graph$"] 0;
            l -> 0; 1 -> r;
            };
        \end{tikzpicture} \mid \graph \in \glang(\term)
    },                                                                                                                             & \glang(\eps)                                                                                                             & \defeq \set{\begin{tikzpicture}[baseline = -.5ex]
                                                                                                                                                                                                                                                                                        \graph[grow right = 1.cm, branch down = 6ex, nodes={mynode}]{
                                                                                                                                                                                                                                                                                        {0/{}[draw, circle]}
                                                                                                                                                                                                                                                                                        };
                                                                                                                                                                                                                                                                                        \node[left = .5em of 0](l){};
                                                                                                                                                                                                                                                                                        \node[right = .5em of 0](r){};
                                                                                                                                                                                                                                                                                        \graph[use existing nodes, edges={color=black, pos = .5, earrow}, edge quotes={fill=white, inner sep=1pt,font= \scriptsize}]{
                                                                                                                                                                                                                                                                                            l -> 0; 0 -> r;
                                                                                                                                                                                                                                                                                        };
                                                                                                                                                                                                                                                                                    \end{tikzpicture}}, \\
    \hspace{-1.em}\glang(\term[1] \cap \term[2])                                                                                  & \defeq
    \set{\begin{tikzpicture}[baseline = -.5ex]
                 \graph[grow right = 1.cm, branch down = 2.5ex, nodes={mynode, font = \scriptsize}]{
                 {s1/{}[draw, circle]}
                 -!- {t1/{}[draw, circle]}
                 };
                 \node[left = 4pt of s1](s1l){} edge[earrow, ->] (s1);
                 \node[right = 4pt of t1](t1l){}; \path (t1) edge[earrow, ->] (t1l);
                 \graph[use existing nodes, edges={color=black, pos = .5, earrow}, edge quotes={fill=white, inner sep=1pt,font= \scriptsize}]{
                 s1 ->["$G$", bend left = 25] t1;
                 s1 ->["$H$", bend right = 25] t1;
                 };
             \end{tikzpicture} \mid G \in \glang(\term[1]) \land  H \in \glang(\term[2])},&
             \glang(\term[1] \union \term[2])                                                                                                & \defeq \glang(\term[1]) \cup \glang(\term[2]),                                                                                                      \\
    \hspace{-1.em}\glang(\term[1] \compo \term[2])                                                                                 & \defeq \set{\begin{tikzpicture}[baseline = -.5ex]
                                                                                                                                                        \graph[grow right = 1.cm, branch down = 2.5ex, nodes={mynode, font = \scriptsize}]{
                                                                                                                                                        {s1/{}[draw, circle]}
                                                                                                                                                        -!- {c/{}[draw, circle]}
                                                                                                                                                        -!- {t1/{}[draw, circle]}
                                                                                                                                                        };
                                                                                                                                                        \node[left = 4pt of s1](s1l){} edge[earrow, ->] (s1);
                                                                                                                                                        \node[right = 4pt of t1](t1l){}; \path (t1) edge[earrow, ->] (t1l);
                                                                                                                                                        \graph[use existing nodes, edges={color=black, pos = .5, earrow}, edge quotes={fill=white, inner sep=1pt,font= \scriptsize}]{
                                                                                                                                                        s1 ->["$G$"] c;
                                                                                                                                                        c ->["$H$"] t1;
                                                                                                                                                        };
                                                                                                                                                    \end{tikzpicture} \mid G \in \glang(\term[1]) \land H \in \glang(\term[2])}, &
    \glang(\term[1]^{*})                                                                                                          & \defeq
    \bigcup_{n \in \nat} \glang(\term[1]^n).                                                                                             
\end{align*}
\klactivatequotation
Here, we use the "series-composition" for ($\compo$) (gluing the target of the first "graph" with the source of the second), the "parallel-composition" for ($\cap$) (gluing the sources and the targets of the two "graphs", respectively), and the swapping of the source and the target for ($\smile$).
Below are some instances:
\kldeactivatequotation
\begin{align*}
    \glang((a b \com{\eps}) \cap (a \eps c))                 & = \set*{\hspace{-.5em} \begin{tikzpicture}[baseline = -3.ex]
                                                                              \graph[grow right = .6cm, branch down = 2.5ex]{
                                                                              {/, 1/{}[draw, circle, mynode, font = \scriptsize]} -!- {21/{}[draw, circle, mynode, font = \scriptsize], /, 22/{}} -!- {31/{}[draw, circle, mynode, font = \scriptsize], /, 32/{}} -!- {/, 4/{}[draw, circle, mynode, font = \scriptsize]}
                                                                              };
                                                                              \path (22) edge [draw = white, opacity = 0] node[pos= 0.5, opacity = 1, draw = black, circle, mynode, font= \scriptsize](3){}(32);
                                                                              \node[left = 4pt of 1](s1l){} edge[earrow, ->] (1);
                                                                              \node[right = 4pt of 4](t4l){}; \path (4) edge[earrow, ->] (t4l);
                                                                              \graph[use existing nodes, edges={color=black, pos = .5, earrow}, edge quotes={fill=white, inner sep=1pt,font= \scriptsize}]{
                                                                              1 ->["$a$", pos = .4] 21 -> ["$b$", pos = .4] 31 -> ["$\com{\eps}$", pos = .4] 4;
                                                                              1 ->["$a$", pos = .4] 3 -> ["$c$"] 4;
                                                                              };
                                                                          \end{tikzpicture} \hspace{-.5em}}, &
    \glang((ab)^{\lop} c^{\smile} \com{a}^{\lop}b^{\lop}) & = \set*{\begin{tikzpicture}[baseline = 2.ex]
                                                                              \graph[grow right = 1.cm, branch down = 2.5ex, nodes={mynode, font = \scriptsize}]{
                                                                              {1/{}[draw, circle]} -!- {2/{}[draw, circle]}
                                                                              };
                                                                              \node[left = 4pt of 1](s1l){} edge[earrow, ->] (1);
                                                                              \node[right = 4pt of 2](t1l){}; \path (2) edge[earrow, ->] (t1l);
                                                                              \node[above = .6cm of 1, mynode, draw, circle, font =\scriptsize](3){};
                                                                              \graph[use existing nodes, edges={color=black, pos = .5, earrow}, edge quotes={fill=white, inner sep=1pt,font= \scriptsize}]{
                                                                              1 ->["$a$", bend right = 25] 3 -> ["$b$", bend right = 25] 1;
                                                                              2 ->["$c$"] 1;
                                                                              2 ->["$\com{a}$", out = 100, in = 135, looseness = 25] 2;
                                                                              2 ->["$b$", out = 45, in = 80, looseness = 25] 2;
                                                                              };
                                                                          \end{tikzpicture}},                                                                                                                                                                                                  \\
    \glang(a^+)                                              & =  \set*{\begin{tikzpicture}[baseline = -.5ex]
                                                                               \graph[grow right = .8cm, branch down = 2.5ex]{
                                                                               {1/{}[draw, circle, mynode, font = \scriptsize]} -!- {2/{}[draw, circle, mynode, font = \scriptsize]}
                                                                               };
                                                                               \node[left = 4pt of 1](s1l){} edge[earrow, ->] (1);
                                                                               \node[right = 4pt of 2](t2l){}; \path (2) edge[earrow, ->] (t2l);
                                                                               \graph[use existing nodes, edges={color=black, pos = .5, earrow}, edge quotes={fill=white, inner sep=1pt,font= \scriptsize}]{
                                                                               1 ->["$a$"] 2;
                                                                               };
                                                                           \end{tikzpicture}\hspace{-.5em} , \begin{tikzpicture}[baseline = -.5ex]
                                                                                                  \graph[grow right = .8cm, branch down = 2.5ex]{
                                                                                                  {1/{}[draw, circle, mynode, font = \scriptsize]} -!- {2/{}[draw, circle, mynode, font = \scriptsize]}  -!- {3/{}[draw, circle, mynode, font = \scriptsize]}
                                                                                                  };
                                                                                                  \node[left = 4pt of 1](s1l){} edge[earrow, ->] (1);
                                                                                                  \node[right = 4pt of 3](t3l){}; \path (3) edge[earrow, ->] (t3l);
                                                                                                  \graph[use existing nodes, edges={color=black, pos = .5, earrow}, edge quotes={fill=white, inner sep=1pt,font= \scriptsize}]{
                                                                                                  1 ->["$a$"] 2 ->["$a$"] 3;
                                                                                                  };
                                                                                              \end{tikzpicture}\hspace{-.5em} , \hspace{.5em} \dots}. \span \span
\end{align*}
\klactivatequotation
For a "valuation" $\val \colon \vsig \to \wp(A^2)$ and $x, y \in A$,
let $\intro*\gval(\val, x, y)$ be the "graph" $\tuple{A, \set{\hat{\val}(a)}_{a \in \vsigpm_{\eps}}, x, y}$.
We extend $\hat{\val}$ to "graphs" $\graph[2]$ and sets $\glangvar$ of "graphs", as the relations $\hat{\val}(\graph[2]), \hat{\val}(\glangvar) \subseteq A^2$ defined by:
\begin{align*}
    \hat{\val}(\graph[2]) & \defeq \set{\tuple{x, y} \mid \graph[2] \homo \gval(\val, x, y)}, & \hat{\val}(\glangvar) \defeq \bigcup_{\graph[2] \in \glangvar} \hat{\val}(\graph[2]).
\end{align*}
We use the same symbol $\hat{\val}$ for "terms", "formulas", "graphs", and sets of "graphs", as it always denotes the interpretation under $\val$,
so that the following proposition can be simply stated as $\hat{\val}(\term) = \hat{\val}(\glang(\term))$.
We then have the following alternative semantics using "graph languages".
Intuitively, $\tuple{x, y} \in \hat{\val}(\term)$ iff some "graph" in $\glang(\term)$ can be mapped into $\val$ so that its source and target are mapped to $x$ and $y$, respectively.
For example, $\tuple{x, y} \in \hat{\val}((a b \com{\eps}) \cap (a \eps c))$ iff there are $z_1, z_2, z_3$ such that
$\tuple{x, z_1}, \tuple{x, z_3} \in \val(a)$, $\tuple{z_1, z_2} \in \val(b)$, $\tuple{z_2, y} \in \hat{\val}(\com{\eps})$, and $\tuple{z_3, y} \in \val(c)$,
which corresponds to the unique "graph" in $\glang((a b \com{\eps}) \cap (a \eps c))$ above.
\begin{prop}[{\cite[Proposition 11]{nakamuraExistentialCalculiRelations2023}}]\label{prop:glang}
    Let $\val \in \REL$.
    For $\PCoR_{\set{\bl^{*}, \com{\eps}, \com{x}}}$ "terms" $\term$,
    we have $\hat{\val}(\term) = \hat{\val}(\glang(\term))$.
\end{prop}
Note that, for every "submodel" $\val \restr B$ of $\val$, we have $\widehat{(\val \restr B)}(\term) \subseteq \hat{\val}(\term)$ for all $\PCoR_{\set{\bl^{*}, \com{\eps}, \com{x}}}$ "terms" $\term$,
by \Cref{prop:glang} with the canonical "graph homomorphism" (i.e., the inclusion map) $\gval(\val \restr B, x, y) \homo \gval(\val, x, y)$.
Using \Cref{prop:glang}, we have the following bounded model property.
\begin{prop}[bounded model property]\label{prop:bounded-model-property}
    Suppose that $\algclass \subseteq \REL$ is "submodel-closed".
    Let $\term[1]$ and $\term[2]$ be $\PCoR_{\set{\bl^{*}, \com{\eps}, \com{x}}}$ "terms"
    such that $\algclass \not\models \term[1] \le \term[2]$.
    We then have the following.
    \begin{enumerate}
        \item \label{prop:bounded-model-property-1} 
        There is some "valuation" $\val \colon \vsig \to \wp(X^2)$ in $\algclass$, where $X$ is finite, such that
        $\hat{\val}(\term[1]) \not\subseteq \hat{\val}(\term[2])$.

        \item \label{prop:bounded-model-property-2}
        If $\term[1]$ is a $\PCoR_{\set{\com{\eps}, \com{x}}}$ "term" (i.e., without $\bl^{*}$),
        then there is some "valuation" $\val \colon \vsig \to \wp(X^2)$ in $\algclass$, where $\card{X} \le 1 + \size{\term[1]}$, such that
        $\hat{\val}(\term[1]) \not\subseteq \hat{\val}(\term[2])$.
    \end{enumerate}
\end{prop}
\begin{proof}
    \proofcase{For \ref{prop:bounded-model-property-1}}
    Suppose that $\tuple{x, y} \in \hat{\val}(\term[1]) \setminus \hat{\val}(\term[2])$ for some "valuation" $\val \in \algclass$.
    By \Cref{prop:glang},
    there is some "graph homomorphism" $f$ from some "graph" $\graph[2] \in \glang(\term[1])$ into $\gval(\val, x, y)$.
    Let $\val' \defeq \val \restr f(\domain{\graph[2]})$, which is in $\algclass$ as $\algclass$ is "submodel-closed".
    We then have $\tuple{x, y} \in \hat{\val}'(\term[1])$ by the same map $f$.
    Also, we have $\tuple{x, y} \not\in \hat{\val}'(\term[2])$ by $\tuple{x, y} \not\in \hat{\val}(\term[2])$, as $\val'$ is a "submodel" of $\val$.
    We thus have $\hat{\val}'(\term[1]) \not\subseteq \hat{\val}'(\term[2])$.
    Because the number of "vertices" of $\graph[2]$ is finite, this completes the proof.

    \proofcase{For \ref{prop:bounded-model-property-2}}
    By an easy induction on $\term[1]$, we can show that the number of "vertices" of $\graph[2] \in \glang(\term)$ is bounded by $1 + \size{\term[1]}$.
    Hence, in the same way as above, this completes the proof.
\end{proof}
From \Cref{prop:bounded-model-property}, we have the following complexity upper bounds, which apply to both $\REL$ and $\DREL$ (see \Cref{section:PWP-with-loop}; cf.\ \cite[Lemma 22 and Theorem 24]{nakamuraExistentialCalculiRelations2023} for $\REL$).
\begin{prop}\label{prop:upper-bound}
    Suppose that $\algclass \subseteq \REL$ is "submodel-closed" and that the membership $\val \in \algclass$ is in $\PTIME$
    (for finite "valuations" restricted to the "variables" occurring in the input).
    We then have the following.
    \begin{enumerate}
        \item \label{prop:upper-bound-1}
        The "equational theory" of $\PCoR_{\set{\bl^{*}, \com{\eps}, \com{x}}}$ with respect to $\algclass$ is in $\PiOne$.
        \item \label{prop:upper-bound-2}
        The "equational theory" of $\PCoR_{\set{\com{\eps}, \com{x}}}$ with respect to $\algclass$ is in $\coNP$.
        More precisely, the ""inclusion problem"" $\term[1] \le \term[2]$ with respect to $\algclass$ (i.e., given $\term[1]$ and $\term[2]$, does $\algclass \models \term[1] \le \term[2]$ hold?),
        where $\term[1]$ is a $\PCoR_{\set{\com{\eps}, \com{x}}}$ "term" and $\term[2]$ is a $\PCoR_{\set{\bl^{*}, \com{\eps}, \com{x}}}$ "term", is in $\coNP$.
    \end{enumerate}
\end{prop}
\begin{proof}
    \proofcase{For \ref{prop:upper-bound-1}}
    We use the finite model property (\Cref{prop:bounded-model-property}~(\ref{prop:bounded-model-property-1}))
    and the fact that the model checking problem of $\PCoR_{\set{\bl^{*}, \com{\eps}, \com{x}}}$ on finite "valuations" is decidable (in $\PTIME$; see, e.g., \cite[Proposition 20]{nakamuraExistentialCalculiRelations2023}).
    By enumerating all finite "valuations" $\val \in \algclass$ (up to the "variables" occurring in the input) and checking whether $\hat{\val}(\term[1]) \not\subseteq \hat{\val}(\term[2])$,
    we have that the problem is co-recursively enumerable, i.e., in $\PiOne$.

    \proofcase{For \ref{prop:upper-bound-2}}
    We use the linearly bounded model property (\Cref{prop:bounded-model-property}~(\ref{prop:bounded-model-property-2}))
    and the fact that the model checking problem of $\PCoR_{\set{\bl^{*}, \com{\eps}, \com{x}}}$ on finite "valuations" is decidable in $\PTIME$.
    By nondeterministically guessing a "valuation" $\val \in \algclass$ (on the "variables" occurring in the input) of size polynomial in the input size and checking whether $\hat{\val}(\term[1]) \not\subseteq \hat{\val}(\term[2])$,
    we have that the complement of the problem is in $\NP$, i.e., the problem is in $\coNP$.
\end{proof}

\subsection{Deterministic Propositional While Programs with Loop}\label{section:PWP-with-loop}
We define "propositional while programs with loop" ($\intro*\PWPlop$ terms, for short) as a syntactic fragment of $\PCoR_{\set{\bl^{*}}}$ "terms",
which are interpreted under the "hypotheses" $\ttest$ given below (i.e., on the class $\REL_{\ttest}$; recall that $\algclass_{\Gamma}$ denotes the class of "valuations" in $\algclass$ satisfying all "formulas" in $\Gamma$, \Cref{section:semantics}).

Let $\intro*\psig$ be a set of ""primitive tests"" with a fixpoint-free involution $\tilde{\bl} \colon \psig \to \psig$
and let $\intro*\asig$ be a set of ""actions"", where $\psig, \asig \subseteq \vsig$ are disjoint.
Let $\ttest$ be the set of "equations" given by:
\[\intro*\ttest \defeq \set{p \cap \tilde{p} = \emp, p \union \tilde{p} = \eps \mid p \in \psig}.\]
For every $\val \in \REL_{\ttest}$,
each $p \in \psig$ satisfies $\hat{\val}(p) \subseteq \hat{\val}(\eps)$ and $\hat{\val}(\tilde{p}) = \hat{\val}(\eps) \setminus \hat{\val}(p)$.
Thus $p$ behaves as a test (a subset of the identity relation) and $\tilde{p}$ behaves as its complement, as expected.

\subsubsection{Syntax}\label{section:PWP-syntax}
The set of ""tests"", written $\intro*\Tests$, is defined as the minimal set of $\PCoR_{\set{\bl^{*}}}$ "terms" satisfying the following:
\begin{gather*}
    \begin{prooftree}
        \hypo{p \in \psig}
        \infer1{p \in \Tests}
    \end{prooftree}
    \quad
    \begin{prooftree}
        \hypo{\mathstrut}
        \infer1{\eps \in \Tests}
    \end{prooftree}
    \quad
    \begin{prooftree}
        \hypo{\mathstrut}
        \infer1{\emp \in \Tests}
    \end{prooftree}
    \quad
    \begin{prooftree}
        \hypo{\bterm[1] \in \Tests}
        \hypo{\bterm[2] \in \Tests}
        \infer2{\bterm[1] \compo \bterm[2] \in \Tests}
    \end{prooftree}
    \quad
    \begin{prooftree}
        \hypo{\bterm[1] \in \Tests}
        \hypo{\bterm[2] \in \Tests}
        \infer2{\bterm[1] \union \bterm[2] \in \Tests}
    \end{prooftree}.
\end{gather*}
For a "test" $b \in \Tests$, we define the complement $\intro*\tcompl{b} \in \Tests$ as follows:
\begin{gather*}
    \tcompl{p} \defeq \tilde{p}, \quad \tcompl{\eps} \defeq \emp, \quad \tcompl{\emp} \defeq \eps, \quad \tcompl{b \union c} \defeq \tcompl{b} \compo \tcompl{c}, \quad \tcompl{b \compo c} \defeq \tcompl{b} \union \tcompl{c}.
\end{gather*}
Note that $\REL_{\ttest} \models \tcompl{b} = \eps \cap \com{b}$; thus $\tcompl{b}$ expresses the complement of $b$ with respect to $\eps$, as expected.

For "terms" $\term[1], \term[2]$ and a "test" $b$, we define the following operators, which encode ``$\mathbf{if}$-$\mathbf{then}$-$\mathbf{else}$'', ``$\mathbf{while}$'', and ``$\mathbf{do}$-$\mathbf{while}$'', respectively.
\begin{align*}
    \term[1] \introalign\unionb{b} \term[2] & \quad\defeq\quad b \term[1] \union \tcompl{b} \term[2]  \tag{$\mathbf{if}$ $b$ $\mathbf{then}$ $\term[1]$ $\mathbf{else}$ $\term[2]$} \\
    \term[1]^{\introalign\starb{b}}           & \quad\defeq\quad (b \term[1])^{*} \tcompl{b}           \tag{$\mathbf{while}$ $b$ $\mathbf{do}$ $\term[1]$}                           \\
    \term[1]^{\introalign\plusb{b}}           & \quad\defeq\quad \term[1] \term[1]^{\starb{b}}          \tag{$\mathbf{do}$ $\term[1]$ $\mathbf{while}$ $b$}
\end{align*}
These encodings are well known in propositional dynamic logic \cite{fischerPropositionalDynamicLogic1979,halpernPropositionalDynamicLogic1983} and Kleene algebra with tests (KAT) \cite{kozenKleeneAlgebraTests1996,smolkaGuardedKleeneAlgebra2019}.
The set of ""$\PWPlop$ terms"" is the minimal subset $\Progs$ of $\PCoR_{\set{\bl^{*}}}$ satisfying the following:
\begin{gather*}
    \begin{prooftree}[separation = .5em]
        \hypo{b \in \Tests}
        \infer1{b \in \Progs}
    \end{prooftree}
    \quad
    \begin{prooftree}[separation = .5em]
        \hypo{x \in \asig}
        \infer1{x \in \Progs}
    \end{prooftree}
    \quad
    \begin{prooftree}[separation = .5em]
        \hypo{\term[1] \in \Progs}
        \hypo{\term[2] \in \Progs}
        \infer2{\term[1] \compo \term[2] \in \Progs}
    \end{prooftree}
    \quad
    \begin{prooftree}[separation = .5em]
        \hypo{\term[1] \in \Progs}
        \hypo{\term[2] \in \Progs}
        \hypo{b \in \Tests}
        \infer3{\term[1] \unionb{b} \term[2] \in \Progs}
    \end{prooftree}
    \quad
    \begin{prooftree}[separation = .5em]
        \hypo{\term[1] \in \Progs}
        \hypo{b \in \Tests}
        \infer2{\term[1]^{\starb{b}} \in \Progs}
    \end{prooftree}
    \quad
    \begin{prooftree}[separation = .5em]
        \hypo{\term[1] \in \Progs}
        \infer1{\term[1]^{\lop} \in \Progs}
    \end{prooftree}.
\end{gather*}

The set of ""propositional while programs"" ($\intro*\PWP$ terms, for short) is the set of "$\PWPlop$ terms" in which $\bl^{\lop}$ does not occur.
Similarly,
the set of ""propositional while programs with intersection"" ($\intro*\PWPcap$ terms, for short) is the minimal subset $\Progs \subseteq \PCoR_{\set{\bl^{*}}}$ satisfying the rules for "$\PWPlop$ terms" together with
the rule \begin{prooftree}[separation = .5em]
    \hypo{\term[1] \in \Progs}
    \hypo{\term[2] \in \Progs}
    \infer2{\term[1] \cap \term[2] \in \Progs}
\end{prooftree}.

The notations above are similar to those in \cite{smolkaGuardedKleeneAlgebra2019} (where we use $\bl^{\starb{b}}$ instead of ``$\bl^{(b)}$''), which are useful in that we can give an over-approximation for "$\PWPlop$ terms" by eliminating annotations, as follows:
\begin{prop}\label{prop:annotation-elimination}
    For all "$\PWPlop$ terms" $\term$ and $\term[2]$ and "tests" $b$, we have:
    \begin{align*}
        \REL_{\ttest} &\models \term[1] \unionb{b} \term[2] \le \term[1] \union \term[2],&
        \REL_{\ttest} &\models \term[1]^{\starb{b}} \le \term[1]^{*},&
        \REL_{\ttest} &\models \term[1]^{\plusb{b}} \le \term[1]^{+}.
    \end{align*}
\end{prop}
For the relationship with GKAT \cite{smolkaGuardedKleeneAlgebra2019}, see Related Work in \Cref{section:introduction}.

\subsubsection{Semantics: deterministic structures}
In this paper, we mainly consider "$\PWPlop$ terms" on 
deterministic structures \cite[Section 10.1]{harelDynamicLogic2000}.
For a set $W \subseteq \vsig$,
let $\intro*\tfunc^{W}$ be the set of "equations" $\set{a^{\smile} a \le \eps \mid a \in W}$.
In particular, we let $\tfunc \defeq \tfunc^{\vsig}$.
Note that, for every $\val \in \REL$, we have that
$\val \models a^{\smile} a \le \eps$ iff $\val(a)$ is functional (i.e., for all $x, y, z$, if $\tuple{x, y}, \tuple{x, z} \in \val(a)$, then $y = z$).
We then define $\intro*\DREL \defeq \REL_{\tfunc}$.
Namely, $\DREL$ is the class of all "valuations" of \emph{deterministic} "relational models".

In \Cref{section:while}, we consider "$\PWPlop$ terms" on $\DREL_{\ttest}$, which we call \emph{deterministic} "$\PWPlop$ terms".
This follows the terminology of DPDL \cite{ben-ariDeterministicPropositionalDynamic1982}, where actions are interpreted as partial functions (the terms themselves are unchanged).

\subsection{Domino Problems}
A ""domino system""\footnote{This system is essentially equivalent to the system using Wang tiles \cite[p.\ 89]{borgerClassicalDecisionProblem1997}.} is a tuple $\mathcal{D} = \tuple{C, H, V}$, where
\begin{itemize}
    \item $C$ is a finite set (of colors/dominoes),
    \item $H \subseteq C^2$ is a binary relation on $C$ (for constraints on pairs of horizontally adjacent colors),
    \item $V \subseteq C^2$ is a binary relation on $C$ (for constraints on pairs of vertically adjacent colors).
\end{itemize}
Given a "domino system" $\mathcal{D} = \tuple{C, H, V}$, for a map $\tau \colon \Z^2 \to C$, we say that
\begin{itemize}
    \item $\tau$ is a ""tiling"" if, for all $\tuple{x, y} \in \Z^2$,
    $\tuple{\tau(x, y), \tau(x+1, y)} \in H$ and $\tuple{\tau(x, y), \tau(x, y+1)} \in V$,
    \item for $h, v \ge 1$, $\tau$ is ""$\tuple{h,v}$-periodic"" if, for all $\tuple{x, y} \in \Z^2$,
$\tau(x, y) = \tau(x + h, y) = \tau(x, y + v)$,
    \item $\tau$ is \reintro*"periodic" if $\tau$ is "$\tuple{h,v}$-periodic" for some $h, v \ge 1$.
\end{itemize}
We may assume without loss of generality that
(1) $C$ is nonempty; and
(2) for every $c \in C$ there are $d, e \in C$ with $\tuple{c, d} \in H$ and $\tuple{c, e} \in V$.
As no color violating (2) occurs in any "tiling", such colors can be deleted repeatedly (until none remain) in polynomial time.
(If no colors remain, replace the input with the fixed system $C = \set{0, 1}$, $H = C \times \set{0}$, and $V = C \times \set{1}$, which satisfies the stated conditions and has no "tiling".)

\Cref{figure:tiling} illustrates a ($\tuple{4, 2}$-)"periodic tiling".
We recall the following undecidability results on "periodic tilings" and "tilings".
\begin{prop}[{\cite{bergerUndecidabilityDominoProblem1966,gurevichRemarksBergerPaper1972}}; see also {\cite[Theorem 3.1.7]{borgerClassicalDecisionProblem1997}}]
\label{prop:domino-tiling-undecidable}
    \begin{itemize}
        \item ""The periodic domino problem""---given a "domino system" $\mathcal{D}$,
        does $\mathcal{D}$ have a "periodic tiling"?---is $\SigmaOne$-complete.

        \item ""The domino problem""---given a "domino system" $\mathcal{D}$,
        does $\mathcal{D}$ have a "tiling"?---is $\PiOne$-complete.
    \end{itemize}
\end{prop}
Additionally, we use the following $\NP$-complete problem on "periodic tilings".
\begin{prop}[a minor change from {\cite[B1]{harelEffectiveTransformationsInfinite1986}, \cite[Theorem 7.2.1]{lewisElementsTheoryComputation1997}}]\label{prop:bounded-periodic-tiling}
    The ""bounded periodic domino problem""---given a "domino system" $\mathcal{D}$ and a unary-encoded $n \ge 1$,
    does $\mathcal{D}$ have an $\tuple{n,n}$-"periodic tiling"?---is $\NP$-complete.
\end{prop}
\begin{proof}
    \proofcase{(In $\NP$)} The problem is in $\NP$ by guessing a map $\tau \colon \range{0, n-1}^2 \to C$ and checking the conditions of $\tuple{n,n}$-"periodic tilings" (which is in polynomial time, as $n$ is unary-encoded).

    \proofcase{($\NP$-hardness)}
    From \cite[Theorem 7.2.1]{lewisElementsTheoryComputation1997}, the following problem is $\NP$-complete:
    given a "domino system" $\mathcal{D} = \tuple{C, H, V}$, a unary-encoded $n \ge 1$, and $c_0, \dots, c_{n-1} \in C$,
    does $\mathcal{D}$ have a map $\tau \colon \range{0, n-1}^2 \to C$ satisfying the following conditions?
    \begin{itemize}
        \item $\tuple{\tau(x, y),\tau(x+1, y)} \in H$ for $x \in \range{0, n-2}$ and $y \in \range{0,n-1}$,
        \item $\tuple{\tau(x, y),\tau(x, y+1)} \in V$ for $x \in \range{0, n-1}$ and $y \in \range{0,n-2}$,
        \item $\tau(x, 0) = c_x$ for $x \in \range{0, n-1}$.
    \end{itemize}
    Given an input, we define the "domino system" $\mathcal{D}' = \tuple{D, H', V'}$ as follows:
    \begin{align*}
        D &\defeq (C \times \range{0, n-1} \times \range{1, n-1}) \cup \set{\tuple{c_x, x, 0} \mid x \in \range{0, n-1}},\\
        H' &\defeq \set{\tuple{\tuple{c, x, y}, \tuple{c', x', y'}} \in D^2 \mid x' = (x + 1) \bmod n, y' = y, \mbox{ and } (\tuple{c, c'} \in H \mbox{ or } x' = 0)},\\
        V' &\defeq \set{\tuple{\tuple{c, x, y}, \tuple{c', x', y'}} \in D^2 \mid x' = x, y' = (y + 1) \bmod n, \mbox{ and } (\tuple{c, c'} \in V \mbox{ or } y' = 0)}.
    \end{align*}
    We then have that $\mathcal{D}$ has a map satisfying the conditions above
    iff $\mathcal{D}'$ has an $\tuple{n, n}$-"periodic tiling".
    Hence, this completes the proof.
\end{proof}
\section{Encoding of Periodic Tilings}\label{section:grids}
Our encoding of "periodic tilings" is based on the reduction of Goldblatt and Jackson \cite[Section 4]{goldblattWellstructuredProgramEquivalence2012},
where grids are enforced by modal formulas.
Instead, we use the following "hypotheses" $\tgrid$, such that the grid $\Z^2$ can be mapped periodically into each of their finite models (\Cref{prop:grid}).

\kldeactivatequotation
\begin{center}
\begin{minipage}{\linewidth}
    \captionsetup{type=figure}
    \centering
    \subfloat[][{A \kl[$\tuple{h,v}$-periodic tiling]{$\tuple{4, 2}$-periodic tiling} $\tau$.}]{%
    \begin{minipage}[c][3cm][c]{.47\linewidth}
        \centering
        \begin{tikzpicture}[baseline = -.5cm]
            \graph[grow up= .5cm, branch right = .5cm, nodes={font=\tiny}]{
            {m2_m1/{}, m2_0/{}, m2_1/{}, m2_2/{}, m2_3/{}, m2_4/{}, m2_5/{}, m2_6/{}, m2_7/{}, m2_8/{}, m2_9/{}, m2_10/{}} -!-
            {m1_m1/{}, m1_0/{}[mynode,fill = black], m1_1/{}[mynode, fill = red!60], m1_2/{}[mynode, fill = white], m1_3/{}[mynode, fill = blue!25], m1_4/{}[mynode, fill = black], m1_5/{}[mynode, fill = red!60], m1_6/{}[mynode, fill = white], m1_7/{}[mynode, fill = blue!25], m1_8/{}[mynode, fill = black], m1_9/{}[mynode, fill = red!60], m1_10/{}[mynode, fill = white], m1_11/{}} -!-
            {0_m1/{}, 0_0/{$o$}[mynode, fill = white], 0_1/{}[mynode, fill = blue!25], 0_2/{}[mynode, fill = black], 0_3/{}[mynode, fill = red!60], 0_4/{}[mynode, fill = white], 0_5/{}[mynode, fill = blue!25], 0_6/{}[mynode, fill = black], 0_7/{}[mynode, fill = red!60], 0_8/{}[mynode, fill = white], 0_9/{}[mynode, fill = blue!25], 0_10/{}[mynode, fill = black], 0_11/{}} -!-
            {1_m1/{}, 1_0/{}[mynode, fill = black], 1_1/{}[mynode, fill = red!60], 1_2/{}[mynode, fill = white], 1_3/{}[mynode, fill = blue!25], 1_4/{}[mynode, fill = black], 1_5/{}[mynode, fill = red!60], 1_6/{}[mynode, fill = white], 1_7/{}[mynode, fill = blue!25], 1_8/{}[mynode, fill = black], 1_9/{}[mynode, fill = red!60], 1_10/{}[mynode, fill = white], 1_11/{}} -!-
            {2_m1/{}, 2_0/{}[mynode, fill = white], 2_1/{}[mynode, fill = blue!25], 2_2/{}[mynode, fill = black], 2_3/{}[mynode, fill = red!60], 2_4/{}[mynode, fill = white], 2_5/{}[mynode, fill = blue!25], 2_6/{}[mynode, fill = black], 2_7/{}[mynode, fill = red!60], 2_8/{}[mynode, fill = white], 2_9/{}[mynode, fill = blue!25], 2_10/{}[mynode, fill = black], 2_11/{}} -!-
            {3_m1/{}, 3_0/{}[mynode, fill = black], 3_1/{}[mynode, fill = red!60], 3_2/{}[mynode, fill = white], 3_3/{}[mynode, fill = blue!25], 3_4/{}[mynode, fill = black], 3_5/{}[mynode, fill = red!60], 3_6/{}[mynode, fill = white], 3_7/{}[mynode, fill = blue!25], 3_8/{}[mynode, fill = black], 3_9/{}[mynode, fill = red!60], 3_10/{}[mynode, fill = white], 3_11/{}} -!-
            {4_m1/{}, 4_0/{}, 4_1/{}, 4_2/{}, 4_3/{}, 4_4/{}, 4_5/{}, 4_6/{}, 4_7/{}, 4_8/{}, 4_9/{}, 4_10/{}, 4_11/{}} -!-
            };
            \graph[use existing nodes, edges={color=black, pos = .5, earrow}, edge quotes={fill = white, inner sep=1pt,font= \scriptsize}]{
            \foreach \ya/\yb in {m2/m1, m1/0, 0/1, 1/2, 2/3, 3/4}{\foreach \xa in {0, 1, 2, 3, 4, 5, 6, 7, 8, 9, 10}{(\ya_\xa) -- (\yb_\xa);};};
            \foreach \ya in {m1, 0, 1, 2, 3}{\foreach \xa/\xb in {m1/0, 0/1, 1/2, 2/3, 3/4, 4/5, 5/6, 6/7, 7/8, 8/9, 9/10, 10/11}{(\ya_\xa) -- (\ya_\xb);};};
            };
            \begin{pgfonlayer}{foreground}
                \node[fit=(0_0)(0_3)(1_3)(1_0), draw = orange, rectangle, inner sep = 0.1cm, line width = 1pt](back) {};
            \end{pgfonlayer}
        \end{tikzpicture}
        \label{figure:tiling}
    \end{minipage}%
    }\hfill
    \subfloat[][The induced finite \kl{valuation} $\val_{\mathcal{D}, \tau}$.]{%
    \begin{minipage}[c][3cm][c]{.47\linewidth}
        \centering
        \begin{tikzpicture}[baseline = .0cm]
            \graph[grow up= 1.cm, branch right = 1.cm, nodes={font=\tiny}]{
            {0_0/{}[mynode, fill = white], 0_1/{}[mynode, fill = blue!25], 0_2/{}[mynode, fill = black], 0_3/{}[mynode, fill = red!60]} -!-
            {1_0/{}[mynode, fill = black], 1_1/{}[mynode, fill = red!60], 1_2/{}[mynode, fill = white], 1_3/{}[mynode, fill = blue!25]} -!-
            };
            \node[left = .3cm of 0_0](src){};
            \node[below = .3cm of 0_0](tgt){};
            \graph[use existing nodes, edges={color=black, pos = .5, earrow}, edge quotes={fill = white, inner sep=1pt,font= \scriptsize}]{
            \foreach \ya/\yb in {0/1}{\foreach \xa in {0, 1, 2, 3}{(\ya_\xa) ->["$\aN$"] (\yb_\xa); (\yb_\xa) ->["$\aN$", bend right = 30] (\ya_\xa);};};
            \foreach \ya in {0, 1}{\foreach \xa/\xb in {0/1, 1/2, 2/3}{(\ya_\xa) ->["$\aE$"] (\ya_\xb);};};
            \foreach \ya in {0}{\foreach \xa/\xb in {3/0}{(\ya_\xa) ->["$\aE$", bend left = 20] (\ya_\xb);};};
            \foreach \ya in {1}{\foreach \xa/\xb in {3/0}{(\ya_\xa) ->["$\aE$", bend right = 20] (\ya_\xb);};};
            };
        \end{tikzpicture}
        \label{figure:tiling-encoding}
    \end{minipage}%
    }
    \par
    \captionsetup{justification=raggedright}
    \caption{The \kl{domino system} used in both panels is $\mathcal{D} = \tuple{C, H, V}$,
    where $C = \set{\begin{tikzpicture}[baseline = -.5ex]\graph{/{}[mynode, fill = white]};\end{tikzpicture}, \begin{tikzpicture}[baseline = -.5ex]\graph{/{}[mynode, fill = black]};\end{tikzpicture},
\begin{tikzpicture}[baseline = -.5ex]\graph{/{}[mynode, fill = blue!25]};\end{tikzpicture}, \begin{tikzpicture}[baseline = -.5ex]\graph{/{}[mynode, fill = red!60]};\end{tikzpicture}}$,
\mbox{$H = \set{\begin{tikzpicture}[baseline = -.5ex]
    \graph[grow up= .5cm, branch right = .5cm, nodes={font=\tiny}]{{0/{}[mynode, fill = white], 1/{}[mynode, fill = blue!25]}};
    \graph[use existing nodes, edges={color=black, pos = .5, earrow}, edge quotes={fill = white, inner sep=1pt,font= \scriptsize}]{0 -- 1};
\end{tikzpicture}, \begin{tikzpicture}[baseline = -.5ex]
    \graph[grow up= .5cm, branch right = .5cm, nodes={font=\tiny}]{{0/{}[mynode, fill = blue!25], 1/{}[mynode, fill = black]}};
    \graph[use existing nodes, edges={color=black, pos = .5, earrow}, edge quotes={fill = white, inner sep=1pt,font= \scriptsize}]{0 -- 1};
\end{tikzpicture}, \begin{tikzpicture}[baseline = -.5ex]
    \graph[grow up= .5cm, branch right = .5cm, nodes={font=\tiny}]{{0/{}[mynode, fill = black], 1/{}[mynode, fill = red!60]}};
    \graph[use existing nodes, edges={color=black, pos = .5, earrow}, edge quotes={fill = white, inner sep=1pt,font= \scriptsize}]{0 -- 1};
\end{tikzpicture}, \begin{tikzpicture}[baseline = -.5ex]
    \graph[grow up= .5cm, branch right = .5cm, nodes={font=\tiny}]{{0/{}[mynode, fill = red!60], 1/{}[mynode, fill = white]}};
    \graph[use existing nodes, edges={color=black, pos = .5, earrow}, edge quotes={fill = white, inner sep=1pt,font= \scriptsize}]{0 -- 1};
\end{tikzpicture}}$}, and
\mbox{$V = \set*{\begin{tikzpicture}[baseline = 1.ex]
    \graph[grow up= .5cm, branch right = .5cm, nodes={font=\tiny}]{{0/{}[mynode, fill = white]} -!- {1/{}[mynode, fill = black]}};
    \graph[use existing nodes, edges={color=black, pos = .5, earrow}, edge quotes={fill = white, inner sep=1pt,font= \scriptsize}]{0 -- 1};
\end{tikzpicture}, \begin{tikzpicture}[baseline = 1.ex]
    \graph[grow up= .5cm, branch right = .5cm, nodes={font=\tiny}]{{0/{}[mynode, fill = black]} -!- {1/{}[mynode, fill = white]}};
    \graph[use existing nodes, edges={color=black, pos = .5, earrow}, edge quotes={fill = white, inner sep=1pt,font= \scriptsize}]{0 -- 1};
\end{tikzpicture}, \begin{tikzpicture}[baseline = 1.ex]
    \graph[grow up= .5cm, branch right = .5cm, nodes={font=\tiny}]{{0/{}[mynode, fill = blue!25]} -!- {1/{}[mynode, fill = red!60]}};
    \graph[use existing nodes, edges={color=black, pos = .5, earrow}, edge quotes={fill = white, inner sep=1pt,font= \scriptsize}]{0 -- 1};
\end{tikzpicture}, \begin{tikzpicture}[baseline = 1.ex]
    \graph[grow up= .5cm, branch right = .5cm, nodes={font=\tiny}]{{0/{}[mynode, fill = red!60]} -!- {1/{}[mynode, fill = blue!25]}};
    \graph[use existing nodes, edges={color=black, pos = .5, earrow}, edge quotes={fill = white, inner sep=1pt,font= \scriptsize}]{0 -- 1};
\end{tikzpicture}}$}.
    The origin of the tiling is denoted by $o$.}
\end{minipage}
\end{center}
\klactivatequotation

Let $\aE, \aN, \aW, \aS$ be pairwise distinct "variables" not in $\psig$.
Let $\intro*\tgrid$ be the set of the following "hypotheses":
\begin{align*}
    a b&= \eps & \mbox{ for $\tuple{a,b} \in \set{\tuple{\aE, \aW}, \tuple{\aW, \aE}, \tuple{\aN, \aS}, \tuple{\aS, \aN} }$} \tag{bijective-inverse}\label{equation:bijective} \\
    \aS \aE \aN \aW &\ge \eps \tag{4-cycle} \label{equation:cycle}\\
    (\aE \union \aN \union \aW \union \aS)^* &\ge \top \tag{strongly connected} \label{equation:strongly-connected}
\end{align*}
\begin{prop}\label{prop:grid}
    Let $\val \in \REL^{\fin}_{\tgrid}$ where $\val \colon \vsig \to \wp(X^2)$.
    Then, $\val(\aE)$ and $\val(\aN)$ are commuting bijections,
    and there are $h, v \ge 1$ and a map $f \colon \Z^2 \to X$ such that, for all $x, y \in \Z$,
    \[\tuple{f(x, y), f(x + 1, y)} \in \val(\aE), \quad \tuple{f(x, y), f(x, y + 1)} \in \val(\aN), \quad f(x, y) = f(x + h, y) = f(x, y + v).\]
\end{prop}
\begin{proof}
    By (\ref{equation:bijective}), $\val(\aE)$ and $\val(\aN)$ are bijections (and $\val(\aW)$ and $\val(\aS)$ are their inverses).
    By (\ref{equation:cycle}), we have $\val \models \aS \aE \aN \aW = \eps$, by 
    \(\val \models \eps \le \aS \aE \aN \aW \le (\aS \aE \aN \aW)^{\smile} \aS \aE \aN \aW = \aE \aS \aW \aN \aS \aE \aN \aW = \eps\).
    Thus, $\val(\aE)$ and $\val(\aN)$ commute: $\val \models \aE \aN = \aN (\aS \aE \aN \aW) \aE = \aN \aE$.
    As $X$ is finite, there are $h, v \ge 1$ such that $\val(\aE)^{h}$ and $\val(\aN)^{v}$ are the identity map.
    Hence, for any $o \in X$, the map $f(x, y) \defeq \val(\aE)^{x}(\val(\aN)^{y}(o))$ satisfies the conditions
    (where, for $k \ge 1$, $\val(\aE)^{-k}$ and $\val(\aN)^{-k}$ denote $\val(\aW)^{k}$ and $\val(\aS)^{k}$, respectively, i.e., the powers of the inverse bijections).
\end{proof}

We consider the following "valuations" $\val_{\mathcal{D}, \tau}$;
\Cref{figure:tiling-encoding} illustrates $\val_{\mathcal{D}, \tau}$ for the "periodic tiling" in \Cref{figure:tiling} (where $\aW$ and $\aS$ are omitted).
\begin{defi}[from {\cite{goldblattWellstructuredProgramEquivalence2012}}]\label{defi:valuation-of-periodic-tiling}
    Let $\mathcal{D} = \tuple{C, H, V}$ be a "domino system" where $C \subseteq \psig$ and $\tilde{c} \not\in C$ for all $c \in C$.
    For an $\tuple{h,v}$-"periodic tiling" $\tau \colon \Z^2 \to C$,
    let $\val_{\mathcal{D}, \tau} \colon \vsig \to \wp((\range{0, h-1} \times \range{0,v-1})^2)$ be the "valuation" defined by:
    \begin{align*}
        \val_{\mathcal{D}, \tau}(\aE) & \defeq \set{\tuple{\tuple{x \bmod h, y \bmod v}, \tuple{(x+1) \bmod h, y \bmod v}} \mid \tuple{x, y} \in \Z^2},
        \val_{\mathcal{D}, \tau}(\aW) \defeq \val_{\mathcal{D}, \tau}(\aE)^{\smile}, \\
        \val_{\mathcal{D}, \tau}(\aN) & \defeq \set{\tuple{\tuple{x \bmod h, y \bmod v}, \tuple{x \bmod h, (y+1) \bmod v}} \mid \tuple{x, y} \in \Z^2}, 
        \val_{\mathcal{D}, \tau}(\aS) \defeq \val_{\mathcal{D}, \tau}(\aN)^{\smile},\\
        \val_{\mathcal{D}, \tau}(c)   & \defeq \set{\tuple{\tuple{x \bmod h, y \bmod v}, \tuple{x \bmod h, y \bmod v}} \mid \tuple{x, y} \in \Z^2 \mbox{ and } \tau(x, y) = c} \mbox{ for $c \in C$}
    \end{align*}
    (and $\val_{\mathcal{D}, \tau}(\tilde{c}) \defeq \hat{\val}_{\mathcal{D}, \tau}(\eps) \setminus \val_{\mathcal{D}, \tau}(c)$ for $c \in C$,
    the other "primitive tests" are interpreted so that $\ttest$ holds, and $\val_{\mathcal{D}, \tau}(z) \defeq \emptyset$ for the other "variables" $z$).
\end{defi}
\section{Hypothesis Eliminations under REL}\label{section:hypothesis-eliminations}
In this paper, we will use ""hypothesis eliminations"".
More precisely,
given a class $\algclass$ ($\subseteq \REL$, in this paper) of "valuations", a set $\Gamma$ of "formulas", and an "equation" $\term[1] = \term[2]$,
we give an "equation" $\term[1]' = \term[2]'$ such that the following equivalence holds:
\[\algclass_{\Gamma} \models \term[1] = \term[2] \qquad\Longleftrightarrow\qquad \algclass \models \term[1]' = \term[2]'.\]
Note that $\algclass_{\Gamma} \models \term[1] = \term[2]$ iff $\algclass \models (\bigwedge \Gamma) \to \term[1] = \term[2]$ (when $\Gamma$ is finite).
Below we prepare some (mostly known) "hypothesis eliminations" under $\REL$.
We will use these "hypothesis eliminations" mainly for showing undecidability results in this paper.

\subsection{Eliminating Hoare Hypotheses}\label{section:Hoare}
A ""Hoare hypothesis"" \cite{hardinEliminationHypothesesKleene2002,kozenKleeneAlgebraEquations2014,nakamuraExistentialCalculiRelations2023,pousToolsCompletenessKleene2024,nakamuraFiniteRelationalSemantics2025} is an "equation" of the form $\term[3] \le \emp$.
We can eliminate "Hoare hypotheses" as follows.
\begin{prop}[Cor.\ of {\cite{tarskiCalculusRelations1941}}]\label{prop:Hoare-hypotheses}
    Let $\algclass \subseteq \REL$ be \emph{arbitrary}.
    For all "terms" $\term[1], \term[2], \term[3]$, we have:
    \[\algclass_{\term[3] \le \emp} \models \term[1] \le \term[2] \quad\Longleftrightarrow\quad \algclass \models \term[1] \le \term[2] \union \top \term[3] \top.\]
\end{prop}
\begin{proof}
    For all $\val \in \REL$,
    we have $\hat{\val}(\top \term[3] \top) = \begin{cases}
        \hat{\val}(\emp) &  (\val \models \term[3] \le \emp)\\
        \hat{\val}(\top) &  (\val \not\models \term[3] \le \emp)
    \end{cases}$.
    Thus, 
    $\REL \models (\term[3] \le \emp \to \term[1] \le \term[2]) \leftrightarrow \term[1] \le \term[2] \union \top \term[3] \top$.
    (This equivalence can be viewed as a variant of the Schr{\"o}der--Tarski translation theorem \cite[p.\ 86]{tarskiCalculusRelations1941}.)
    Hence, this proposition holds.
\end{proof}
Additionally, multiple "Hoare hypotheses" can be combined into one "Hoare hypothesis" by the following equivalence:
$\REL \models (\term[1] \le \emp \land \term[2] \le \emp) \leftrightarrow \term[1] \union \term[2] \le \emp$.

Below we give some examples of \Cref{prop:Hoare-hypotheses}.
\begin{ex}[{\cite[Theorem 50]{nakamuraExistentialCalculiRelations2023}}]\label{ex:u-le-x}
    For hypotheses of the form $\term[3] \le x$, where $x$ is a "variable",
    by $\REL \models \term[3] \le x \leftrightarrow \term[3] \cap \com{x} = \emp$,
    we have:
    \[\algclass_{\term[3] \le x} \models \term[1] \le \term[2] \hspace{.5em}\Longleftrightarrow\hspace{.5em}
    \algclass \models \term[1] \le \term[2] \union \top (\term[3] \cap \com{x}) \top.\]
\end{ex}
\begin{ex}[{\cite[Section 2.1]{nakamuraUndecidabilityFO3Calculus2019}}]\label{ex:converse}
    For hypotheses of the form $z = x^{\smile}$,
    by $\REL \models z = x^{\smile} \leftrightarrow
    ((x \com{z})^{\lop} = \emp \land (\com{z} x)^{\lop} = \emp \land
    (z \com{x})^{\lop} = \emp \land (\com{x} z)^{\lop} = \emp)$,
    we have:
    \begin{align*}
       \algclass_{z = x^{\smile}} \models \term[1] \le \term[2]
       &\hspace{.5em}\Longleftrightarrow\hspace{.5em} \algclass \models \term[1] \le \term[2] \union \top ((x \com{z})^{\lop} \union (\com{z} x)^{\lop} \union
       (z \com{x})^{\lop} \union (\com{x} z)^{\lop}) \top.
    \end{align*}
    After pushing the converse operator inside as much as possible,
    we can eliminate the converse operator using such hypotheses.
    For instance, when $\algclass \subseteq \REL$ and $z$ is "fresh" in $\algclass$ and $z \neq x$,
    we have:
    \begin{align*}
        \algclass \models x \le (x^{\smile} x x^{\smile})^{\smile}
        &\;\Leftrightarrow\; \algclass \models x \le x x^{\smile} x \tag{pushing $\bl^{\smile}$ inside}\\
        &\;\Leftrightarrow\; \algclass_{z = x^{\smile}} \models x \le x x^{\smile} x \tag{$\Rightarrow$: Trivial; $\Leftarrow$: replace $\val(z)$ with $\hat{\val}(x^{\smile})$ ($z$ is \kl{fresh})}\\
        &\;\Leftrightarrow\; \algclass_{z = x^{\smile}} \models x \le x z x \tag{By $z = x^{\smile}$}\\
        &\;\Leftrightarrow\; \algclass \models x \le x z x \union \top ((x \com{z})^{\lop} \union (\com{z} x)^{\lop} \union
        (z \com{x})^{\lop} \union (\com{x} z)^{\lop}) \top. \tag{\Cref{prop:Hoare-hypotheses}}
    \end{align*}
    Thus, for $\set{\union, \compo, \top, \com{x}, \bl^{\lop}} \subseteq S$,
    we can reduce the "equational theory" of $S$-"terms" to "that@equational theory" for $S \setminus \set{\bl^{\smile}}$.
\end{ex}
\begin{ex}\label{ex:u-le-x-2}
    For hypotheses of the form $\term[3] \le x$,
    by $\REL \models \term[3] \le x \leftrightarrow (\term[3] \com{x}^{\smile})^{\lop} = \emp$ with the "hypothesis elimination" for $\bl^{\smile}$ in \Cref{ex:converse},
    we have:
    \begin{align*}
        \algclass_{\term[3] \le x} \models \term[1] \le \term[2]
        &\;\Longleftrightarrow\; \algclass \models \term[1] \le \term[2] \union \top (\term[3] \com{x}^{\smile})^{\lop} \top\\
        &\;\Longleftrightarrow\; \algclass_{z = \com{x}^{\smile}} \models \term[1] \le \term[2] \union \top (\term[3] z)^{\lop} \top \tag{$z$ is \kl{fresh}}\\
        &\;\Longleftrightarrow\; \algclass \models \term[1] \le \term[2] \union \top ((\term[3] z)^{\lop} \union (\com{x}\, \com{z})^{\lop} \union (\com{z}\, \com{x})^{\lop} \union (z x)^{\lop} \union (x z)^{\lop}) \top, \tag{\Cref{ex:converse} with $\com{x}$ in place of $x$}
    \end{align*}
    where $z$ is "fresh" in $\algclass$ and $z \not\in \vsig(\term[1]) \cup \vsig(\term[2]) \cup \vsig(\term[3]) \cup \set{x}$.
\end{ex}
\begin{ex}\label{ex:u-le-I}
    For hypotheses of the form $\term[3] \le \eps$,
    similarly to \Cref{ex:u-le-x,ex:u-le-x-2}, we have:
    \[\algclass_{\term[3] \le \eps} \models \term[1] \le \term[2]
    \quad\Longleftrightarrow\quad \algclass \models \term[1] \le \term[2] \union \top (\term[3] \cap \com{\eps}) \top
    \quad\Longleftrightarrow\quad \algclass \models \term[1] \le \term[2] \union \top (\term[3] \com{\eps})^{\lop} \top. \]
\end{ex}

\begin{rem}\label{rem:Hoare-KA}
    In ""Kleene algebras"" (e.g., \cite{cohenHypothesesKleeneAlgebra1994} for general "Kleene algebras", \cite[Theorem 4.1]{kozenHoareLogicKleene2000} for $\REL$, \cite[Theorem 27]{hardinProofTheoryKleene2005}, \cite[Theorem 3.2 (3.3)]{hardinModularizingElimination$r0$2005} for $\REL_{\Gamma}$), 
    when the "variables" of $\term[1]$, $\term[2]$, and $\term[3]$ (and $\Gamma$) are in a finite set $\Sigma \subseteq \vsig$,
    "Hoare hypotheses" are usually eliminated in the following style (i.e., $\top$ from \Cref{prop:Hoare-hypotheses} has been replaced with $\Sigma^{*}$):
    \[\algclass_{\term[3] \le \emp} \models \term[1] \le \term[2] \quad\Longleftrightarrow\quad \algclass \models \term[1] \le \term[2] \union \Sigma^{*} \term[3] \Sigma^{*},\]
    where $\Sigma^{*}$ denotes the "term" $(\sum_{a \in \Sigma} a)^{*}$.
    For a "KL" "term" $\term[1]$ and $\PCoR_{\set{\bl^{*}, \com{\eps}, \com{x}}}$ "terms" $\term[2]$ and $\term[3]$, the equivalence follows from \Cref{prop:Hoare-hypotheses} by taking "submodels" whenever $\algclass \subseteq \REL$ is "interval-closed" with respect to $\Sigma^*$ (\Cref{thm:Hoare-hypothesis-elimination-interval}).
    This includes $\REL$ and $\REL_{\Gamma}$ when each "inequation" in $\Gamma$ has a $\PCoR_{\set{\bl^{*}, \com{\eps}, \com{x}}}$ left-hand side and a "KL" right-hand side (\Cref{cor:Hoare-hypothesis-elimination}).
    For an example where the equivalence fails, take $\algclass = \REL_{\Gamma}$ with $\Gamma = \set{\top a \top = \top}$, $\term[1] = \eps$, $\term[2] = \emp$, $\term[3] = a$, and $\Sigma = \set{a}$.
\end{rem}

\subsection{Hypothesis Elimination via Graph Languages}\label{section:graph-languages}
Let $\algclass'$ and $\algclass$ be such that $\algclass' \subseteq \algclass \subseteq \REL$.
For a "term" $\term$, we say that $\algclass'$ is a ""witness-basis"" of $\algclass$ for $\term$ if,
for every $\val \colon \vsig \to \wp(A^2)$ in $\algclass$ and $\tuple{x, y} \in \hat{\val}(\term)$,
there is some non-empty $B \subseteq A$ such that $(\val \restr B) \in \algclass'$ and $\tuple{x, y} \in \widehat{\val \restr B}(\term)$.
\begin{prop}\label{prop:submodel-cover}
    Let $\algclass'$ and $\algclass$ be such that $\algclass' \subseteq \algclass \subseteq \REL$.
    For all "terms" $\term[1]$ and $\PCoR_{\set{\bl^{*}, \com{\eps}, \com{x}}}$ "terms" $\term[2]$,
    if $\algclass'$ is a "witness-basis" of $\algclass$ for $\term[1]$,
    then we have:
    \[\algclass' \models \term[1] \le \term[2] \quad\Longleftrightarrow\quad \algclass \models \term[1] \le \term[2].\]
\end{prop}
\begin{proof}
    \proofcase{($\Longleftarrow$)}
    By $\algclass' \subseteq \algclass$.

    \proofcase{($\Longrightarrow$)}
    Suppose that $\tuple{x, y} \in \hat{\val}(\term[1])$ and $\val \in \algclass$.
    As $\algclass'$ is a "witness-basis" of $\algclass$ for $\term[1]$,
    we have that $(\val \restr B) \in \algclass'$ and $\tuple{x, y} \in \widehat{\val \restr B}(\term[1])$ for some $B$.
    By $\algclass' \models \term[1] \le \term[2]$,
    we have $\tuple{x, y} \in \widehat{\val \restr B}(\term[2])$.
    As $\val \restr B$ is a "submodel" of $\val$, we have $\tuple{x, y} \in \hat{\val}(\term[2])$ (see the note after \Cref{prop:glang}).
\end{proof}

\begin{ex}\label{ex:surjective}
    (See also \cite[Lemma 21]{nakamuraExistentialCalculiRelations2023} when $\algclass = \REL$.)
    For a class $\algclass \subseteq \REL$ and a set $\glangvar$ of "graphs", let
    \[\intro*\Sur_{\glangvar}(\algclass) \defeq \set*{ \val \in \algclass \;\middle|\; \begin{aligned}
    &\mbox{there are some $\graph[2] \in \glangvar$, $x$, $y$, and $h$ such that}\\
    &\mbox{\quad $h \colon \graph[2] \homo \gval(\val, x, y)$ and $h$ is surjective}
    \end{aligned}}.\]
    When $\algclass \subseteq \REL$ is "submodel-closed",
    for all $\PCoR_{\set{\bl^{*}, \com{\eps}, \com{x}}}$ "terms" $\term[1]$ and $\term[2]$,
    we have that
    $\Sur_{\glang(\term)}(\algclass)$ is a "witness-basis" of $\algclass$ for $\term[1]$
    (Proof: Let $\tuple{x, y} \in \hat{\val}(\term[1])$ where $\val \in \algclass$.
    By \Cref{prop:glang}, there is a "graph homomorphism" $h \colon \graph[2] \homo \gval(\val, x, y)$ for some $\graph[2] \in \glang(\term[1])$.
    Then, $(\val \restr h(\domain{\graph[2]})) \in \Sur_{\glang(\term)}(\algclass)$,
    as $\algclass$ is "submodel-closed" and $h$ is a \emph{surjective} "graph homomorphism" from $\graph[2]$ to $\gval(\val \restr h(\domain{\graph[2]}), x, y)$.
    Also, $\tuple{x, y} \in \reallywidehat{\val \restr h(\domain{\graph[2]})}(\term[1])$ holds by the same $h$).
    We thus have:
    \[\Sur_{\glang(\term)}(\algclass) \models \term[1] \le \term[2] \quad\Longleftrightarrow\quad \algclass \models \term[1] \le \term[2].\]
    Note that, in this argument, $\glang(\term[1])$ is used only via $\hat{\val}(\term[1]) = \hat{\val}(\glang(\term[1]))$ (\Cref{prop:glang});
    hence, the same holds when $\glang(\term[1])$ is replaced with any set $\glangvar$ of "graphs" such that $\hat{\val}(\term[1]) = \hat{\val}(\glangvar)$ for all $\val \in \algclass$.
\end{ex}
In the sequel, we use \Cref{prop:submodel-cover} in the style of \Cref{ex:surjective}.

In \Cref{prop:submodel-cover}, for every "formula" $\fml$ such that $\algclass' \models \fml$, by $\algclass' \subseteq \algclass_{\fml} \subseteq \algclass$, we have:
\[\algclass_{\fml} \models \term[1] \le \term[2] \quad\Longleftrightarrow\quad \algclass \models \term[1] \le \term[2].\]
Hence, \Cref{prop:submodel-cover} can be viewed as a hypothesis elimination.
In particular, when $\algclass$ is "submodel-closed", we have 
$\algclass^{\fin} \models \term[1] \le \term[2]$ iff $\algclass \models \term[1] \le \term[2]$,
because each "graph" in $\glang(\term)$ is finite (and hence, $\Sur_{\glang(\term)}(\algclass) \subseteq \algclass^{\fin} \subseteq \algclass$).
The finite model property (\Cref{prop:bounded-model-property}~(\ref{prop:bounded-model-property-1})) can also be viewed as a corollary of this hypothesis elimination.

\subsection{Eliminating Substitution Hypotheses}\label{section:elimination-substitution}
We recall the hypothesis elimination via substitution.
The idea is already found in, e.g., \cite{hardinEliminationHypothesesKleene2002}
for KAT and
\cite[Lemma 3.8 (iv) (v)]{pousToolsCompletenessKleene2024} for a framework of Kleene algebra with hypotheses (see also \Cref{section:substitution-examples} for more details).
Based on substitutions, we give a hypothesis elimination for hypotheses of the form $x = \term[3]$ (called ""substitution hypotheses"").
Given a "variable" $x$ and a "term" $\term[3]$, for a "term" $\term[1]$,
we write $\term[1][\term[3]/x]$ for the "term" $\term$ in which each occurrence of the "variable" $x$ has been replaced with the "term" $\term[3]$.
For a "valuation" $\val$, let $\val[\term[3]/x](y) \defeq \begin{cases}
        \hat{\val}(\term[3]) & (y = x)                         \\
        \hat{\val}(y)        & (\mbox{otherwise})
    \end{cases}$.
By easy induction on $\term[1]$, we have $\widehat{\val[\term[3]/x]}(\term[1]) = \hat{\val}(\term[1][\term[3]/x])$,
so $\val[\term[3]/x] \models \term[1] = \term[2]$ iff $\val \models \term[1][\term[3]/x] = \term[2][\term[3]/x]$.
For a class $\algclass[1]$ of "valuations", we write $\algclass[1][\term[3]/x] \defeq \set{\val[\term[3]/x] \mid \val \in \algclass}$.
We then have:
\[\algclass[1][\term[3]/x] \models \term[1] = \term[2] \quad\Longleftrightarrow\quad \algclass[1] \models \term[1][\term[3]/x] = \term[2][\term[3]/x].\]
Under certain conditions, we have $\algclass[1][\term[3]/x] = \algclass_{x = \term[3]}$,
and thus we can eliminate the "hypothesis" $x = \term[3]$, as follows.
\begin{prop}\label{prop:REL-axiom-=-substitution}
    Let $\algclass[1]$ be a class of "valuations".\footnote{More generally, this proposition holds for classes of valuations of any algebras (of any signature), as the proof uses no property specific to relations.}
    Let $x$ be a "variable" and $\term[3]$ be a "term".
    If $\algclass[1] \models \term[3] = \term[3][\term[3]/x]$ and $\algclass[1][\term[3]/x] \subseteq \algclass[1]$,
    then $\algclass[1]_{x = \term[3]} = \algclass[1][\term[3]/x]$.
    Consequently, we have the following: 
    \[\algclass[1]_{x = \term[3]} \models \term[1] = \term[2] \quad\Longleftrightarrow\quad \algclass \models \term[1][\term[3]/x] = \term[2][\term[3]/x].\]
\end{prop}
\begin{proof}
    \proofcase{($\subseteq$)}
    Trivial, as $\val \in \algclass[1]_{x = \term[3]}$ implies $\val = \val[\term[3]/x]$.

    \proofcase{($\supseteq$)}
    Let $\val \in \algclass$.
    By $\val[\term[3]/x] \in \algclass[1][\term[3]/x] \subseteq \algclass$
    and $\widehat{\val[\term[3]/x]}({x = \term[3]}) \,=\, \hat{\val}({\term[3] = \term[3][\term[3]/x]}) \,=\, \const{true}$,
    we have $\val[\term[3]/x] \in \algclass_{x = \term[3]}$. %
\end{proof}
The condition ``$\algclass[1][\term[3]/x] \subseteq \algclass[1]$'' is crucial in \Cref{prop:REL-axiom-=-substitution}
(for example, when $\algclass = \REL_{x y = \emp}$ and $\val \in \algclass$ satisfies $\val(x) = \emptyset$ and $\val(y) = \hat{\val}(\eps)$,
we have $\val[y/x] \in \algclass[1][y/x] \setminus \algclass$ by $\hat{\val}(y y) \neq \emptyset$;
indeed, $\REL_{x y = \emp, x = y} \models x x \le \emp$ but $\REL_{x y = \emp} \not\models y y \le \emp$ by this $\val$).
In particular, when $x$ is "fresh" in $\algclass[1]$, we have $\algclass[1][\term[3]/x] \subseteq \algclass[1]$.

Below we give some examples of \Cref{prop:REL-axiom-=-substitution} (see also \Cref{section:substitution-examples} for more examples).
\begin{ex}\label{ex:substitution-id-by-loop}
    For hypotheses of the form $x \le \eps$ (where $x$ is a "variable"),
    when $x \not\in \vsig(\Gamma)$,
    by $\REL \models x \le \eps \leftrightarrow x = x^{\lop}$ and $\REL \models x^{\lop} = x^{\lop}[x^{\lop}/x]$ (and $\REL_{\Gamma}[x^{\lop}/x] \subseteq \REL_{\Gamma}$ as $x$ is "fresh" in $\REL_{\Gamma}$),
    we have:
    \[\REL_{\Gamma, x \le \eps} \models \term[1] = \term[2] \quad\Longleftrightarrow\quad \REL_{\Gamma} \models (\term[1] = \term[2])[x^{\lop}/x].\]
    Moreover, when $\term[2]$ is a $\PCoR_{\set{\bl^{*}, \com{\eps}, \com{x}}}$ "term" and $\com{x}$ does not occur in $\term[2]$, we also have:
    \[\REL_{\Gamma, x \le \eps} \models \term[1] \le \term[2] \quad\Longleftrightarrow\quad \REL_{\Gamma} \models (\term[1] \le \term[2])[x^{\lop}/x] \quad\Longleftrightarrow\quad \REL_{\Gamma} \models \term[1][x^{\lop}/x] \le \term[2].\]
    The first equivalence is by the one above.
    For the second one, we use $\REL \models \term[2][x^{\lop}/x] \le \term[2]$ and $\REL_{x \le \eps} \models \term[1] = \term[1][x^{\lop}/x]$ (by easy induction on $\term[2]$ and $\term[1]$, respectively).

    \proofcase{($\Longrightarrow$)}
    By the former.

    \proofcase{($\Longleftarrow$)}
    By the latter, applied to $\val[x^{\lop}/x] \in \REL_{\Gamma, x \le \eps}$ for each $\val \in \REL_{\Gamma}$.
\end{ex}
\begin{ex}\label{ex:substitution-id-by-dom}
    For hypotheses of the form $x \le \eps$,
    when $x \not\in \vsig(\Gamma)$,
    by $\REL \models x \le \eps \leftrightarrow x = x^{\dom}$ and $\REL \models x^{\dom} = x^{\dom}[x^{\dom}/x]$ (and $\REL_{\Gamma}[x^{\dom}/x] \subseteq \REL_{\Gamma}$ as $x$ is "fresh" in $\REL_{\Gamma}$),
    we have:
    \[\REL_{\Gamma, x \le \eps} \models \term[1] = \term[2] \quad\Longleftrightarrow\quad \REL_{\Gamma} \models (\term[1] = \term[2])[x^{\dom}/x].\]
\end{ex}

\begin{ex}[implicitly used in {\cite[Section 5]{nakamuraUndecidabilityPositiveCalculus2024}}]\label{ex:substitution-tests}
    Based on \Cref{ex:substitution-id-by-loop,ex:substitution-id-by-dom}, we can encode "tests" using the "loop operator" ($\lop$) or the "domain operator" ($\dom$).
    (See also, e.g., \Cref{ex:tests} (for $\REL$) 
    and \cite[Section 4.2]{pousToolsCompletenessKleene2024} (for general "Kleene algebras"), for other encodings of tests.)
    We recall the sets $\psig$ and $\ttest$ in \Cref{section:PWP-with-loop}.
    As the "primitive tests" not occurring in $\term[1] \le \term[2]$ do not affect its satisfaction, we may assume that $\psig = \set{p_0, \dots, p_{n-1}, \tilde{p}_0, \dots, \tilde{p}_{n-1}}$ (consisting of $2n$ distinct elements).
    To eliminate $\ttest$,
    let $q_0, \dots, q_{n-1}$ be pairwise distinct "variables" not in $\psig$ and not occurring in $\term[1] \le \term[2]$, and let $\Theta$ be the set of "equations" given by:
    \[\Theta \defeq \set{p_i = q_i^{\lop}, \tilde{p}_i = \com{q}_i^{\lop} \mid i \in \range{0, n-1}}.\]
    We then have:
    \begin{align*}
        &\REL_{\ttest} \models \term[1] \le \term[2]\\
        &\Leftrightarrow \REL_{\ttest, \Theta} \models \term[1] \le \term[2] \tag{$\Rightarrow$: Trivial. $\Leftarrow$: By replacing $\val(q_i)$ with $\val(p_i)$ for each $\val \in \REL_{\protect\ttest}$}\\
        &\Leftrightarrow \REL_{\Theta} \models \term[1] \le \term[2] \tag{$\Leftarrow$: Trivial. $\Rightarrow$: By $\REL_{\Theta} \models \protect\ttest$ ($\bigstar$)} \\
        &\Leftrightarrow \REL \models (\term[1] \le \term[2])[q_0^{\lop}, \dots, q_{n-1}^{\lop}, \com{q}_0^{\lop}, \dots, \com{q}_{n-1}^{\lop}/p_0, \dots, p_{n-1}, \tilde{p}_0, \dots, \tilde{p}_{n-1}]. \tag{\Cref{prop:REL-axiom-=-substitution}, applied to each $p_i$ and $\tilde{p}_i$} 
    \end{align*}
    Here, ($\bigstar$) is shown by $\REL \models q^{\lop} \cap \com{q}^{\lop} = \emp \land q^{\lop} \union \com{q}^{\lop} = \eps$.

    Additionally, some minor variants of the reduction above also work.
    For example, we can replace $q_i^{\lop}$ and $\com{q}_i^{\lop}$ with $q_i^{\dom}$ and $q_i^{\adom}$ in the reduction above, by $\REL \models q^{\dom} \cap q^{\adom} = \emp \land q^{\dom} \union q^{\adom} = \eps$.
    Also, when $\term[2]$ is a $\PCoR_{\set{\bl^{*}, \com{\eps}, \com{x}}}$ "term" and $\com{p}$ does not occur in $\term[2]$ for $p \in \psig$, the reduction works even if the right-hand side term has been replaced with $\term[2][q_0, \dots, q_{n-1}, \com{q}_0^{\lop}, \dots, \com{q}_{n-1}^{\lop}/p_0, \dots, p_{n-1}, \tilde{p}_0, \dots, \tilde{p}_{n-1}]$ (i.e., only $p_0, \dots, p_{n-1}$ are replaced without $\bl^{\lop}$), as we can replace $\val(q_i)$ with $\val(q_i) \cap \hat{\val}(\eps)$ for each $\val$.
    The same equivalences also hold for $\DREL$ (as $\val(p_i)$ is a partial function for $\val \in \DREL_{\ttest}$, and $\hat{\val}(q_i^{\lop})$ and $\hat{\val}(\com{q}_i^{\lop})$ are partial functions for every $\val$).
\end{ex}

\begin{rem}\label{rem:id}
    For general "Kleene algebras", other hypothesis eliminations for $x \le \eps$ are given in \cite[Section 4]{cohenHypothesesKleeneAlgebra1994}, \cite[Theorem 2]{kozenKleeneAlgebraEquations2014}, \cite[Lemma 3.8 (i) (and (ii))]{pousToolsCompletenessKleene2024}.
    However, it seems that they (explicitly or implicitly) rely on the soundness of "Kleene algebras" with hypotheses with respect to a closure model of (word) languages \cite[Theorem 2.4]{pousToolsCompletenessKleene2024}.
    This soundness fails with respect to $\REL$, unfortunately.
    For instance, $\REL_{x \le \eps} \models x = xx$ holds,
    but the $\set{x \le \eps}$-closures \cite[Definition 2.2]{pousToolsCompletenessKleene2024} of $x$ and $xx$ are the languages given by regular expressions $x^* x x^*$ and $x^* x x^* x x^*$, respectively (they are clearly different sets), cf.\ \cite[pp.\ 3--4]{hardinEliminationHypothesesKleene2002}.
    Nevertheless, \Cref{ex:substitution-tests} (and \Cref{ex:tests}) work with respect to $\REL$,
    by using extra operators.
\end{rem}

\section{Eliminating Loop Hypotheses under Certain Conditions}\label{section:Hypothesis-elimination-using-graph-loops}
We say that a ""loop hypothesis"" is an "equation" of the form $\term[3] \ge \eps$.
In this section, we consider eliminating "loop hypotheses".
Multiple "loop hypotheses" can be combined into one "loop hypothesis" by the equivalence:
$\REL \models (\term[1]\ge \eps \land \term[2]\ge \eps) \leftrightarrow \term[1]^{\lop} \term[2]^{\lop} \ge \eps$
(also note that $\REL \models \term[3]^{\lop \lop} = \term[3]^{\lop}$).
We consider the following translation.
\begin{defi}\label{defi:loop-transformation}
    For $\PCoR_{\set{\bl^{*}, \com{\eps}, \com{x}}}$ "terms" $\term$ and $\term[3]$,
    let $\intro*\Tr_{\term[3]}(\term)$ be the $\PCoR_{\set{\bl^{*}, \com{\eps}, \com{x}}}$ "term" defined as follows:
    \begin{align*}
        \Tr_{\term[3]}(x)                       & = \term[3] x \term[3]  \ \mbox{ for $x \in \vsigpm_{\eps} \cup \set{\top}$}, &
        \Tr_{\term[3]}(\term[1]^{\smile})       & = \Tr_{\term[3]}(\term[1])^{\smile},& 
        \Tr_{\term[3]}(\emp)                   & = \emp,                     \\
        \Tr_{\term[3]}(\term[1] \cap \term[2])  & = \Tr_{\term[3]}(\term[1]) \cap \Tr_{\term[3]}(\term[2]), &
        \Tr_{\term[3]}(\term[1] \union \term[2]) & = \Tr_{\term[3]}(\term[1]) \union \Tr_{\term[3]}(\term[2]),      \\
        \Tr_{\term[3]}(\term[1] \compo \term[2]) & = \Tr_{\term[3]}(\term[1]) \compo \Tr_{\term[3]}(\term[2]), &
        \Tr_{\term[3]}(\term[1]^*)             & = \Tr_{\term[3]}(\term[1])^{*} \term[3], \\
        \Tr_{\term[3]}(\term[1]^{\lop})        & = \term[3] \Tr_{\term[3]}(\term[1])^{\lop} \term[3].
    \end{align*}
    (Here, we regard $\bl^{\lop}$ as a primitive operator.)
\end{defi}
Below we list some easy facts.
\begin{prop}\label{prop:loop-transformation}
    For all $\PCoR_{\set{\bl^{*}, \com{\eps}, \com{x}}}$ "terms" $\term[1]$, $\term[3]$, $\term[3]_1$, $\term[3]_2$, we have:
    \begin{enumerate}
        \item \label{prop:loop-transformation-monotone} $\REL \models \term[3]_1 \le \term[3]_2 \to \Tr_{\term[3]_1}(\term[1]) \le \Tr_{\term[3]_2}(\term[1])$,
        \item \label{prop:loop-transformation-id} $\REL \models \Tr_{\eps}(\term[1]) = \term[1]$,
        \item \label{prop:loop-transformation-sub} $\REL \models \Tr_{\term[3]^{\lop}}(\term[1]) \le \term[1]$,
        \item \label{prop:loop-transformation-sup} $\REL \models \term[3] \ge \eps \to \term[1] \le \Tr_{\term[3]^{\lop}}(\term[1])$.
    \end{enumerate}
\end{prop}
\begin{proof}
\ref{prop:loop-transformation-monotone}, \ref{prop:loop-transformation-id}:
By easy induction on $\term$.
\ref{prop:loop-transformation-sub}, \ref{prop:loop-transformation-sup}:
By \ref{prop:loop-transformation-monotone} and \ref{prop:loop-transformation-id}.
\end{proof}
We recall the "graph languages" for $\PCoR_{\set{\bl^{*}, \com{\eps}, \com{x}}}$ (\Cref{section:alternative-semantics}).
In this view, the translation $\Tr_{\term[3]^{\lop}}(\term)$ extends each "vertex" of "graphs" of $\glang(\term[1])$ with some "graph" of $\glang(\term[3]^{\lop})$, as follows.
\begin{defi}\label{defi:loop-transformation-graph}
    For a "graph" $\graph[2]$ and a map $f$ from $\domain{\graph[2]}$ to "graphs",
    the ""(graph) loop extension"" of $\graph[2]$ with respect to $f$, written $\graph[2][f]$, is the "graph" $\graph[2]$ in which $f(z)$ is glued to $\graph[2]$ by merging $\src^{f(z)}$, $\tgt^{f(z)}$, and $z$, for each $z \in \domain{\graph[2]}$.
    \Cref{figure:graph-loop} gives an illustrative example where edge labels are omitted.
    \begin{figure}[ht]
        \begin{align*}
            \graph[2]    & = \begin{tikzpicture}[baseline = .5ex]
                                 \tikzstyle{mynodeg} = [mynode, fill= gray!20, draw, circle]
                                 \graph[grow right = .8cm, branch down = 2.5ex, nodes={font = \scriptsize}]{
                                 {1/{$1$}[mynodeg]}
                                 -!- {2/{$2$}[mynodeg, yshift = 4ex]}
                                 -!- {3/{$3$}[mynodeg]}
                                 };
                                 \node[left = 4pt of 1](1l){} edge[earrow, ->] (1);
                                 \node[right = 4pt of 3](3l){}; \path (3) edge[earrow, ->] (3l);
                                 \graph[use existing nodes, edges={color=black, pos = .5, earrow}, edge quotes={fill=white, inner sep=1pt,font= \scriptsize}]{
                                     1 -> 2 -> 3; 1 -> 3;
                                 };
                             \end{tikzpicture},                                   &
            \graph[2][f] & = \begin{tikzpicture}[baseline = .5ex]
                                 \tikzstyle{mynodeg} = [mynode, fill= gray!20, draw, circle]
                                 \graph[grow right = .8cm, branch down = 2.5ex, nodes={font = \scriptsize}]{
                                 {1/{$1$}[mynodeg]}
                                 -!- {2/{$2$}[mynodeg, yshift = 4ex]}
                                 -!- {3/{$3$}[mynodeg]}
                                 };
                                 \node[left = 4pt of 1](1l){} edge[earrow, ->] (1);
                                 \node[right = 4pt of 3](3l){}; \path (3) edge[earrow, ->] (3l);
                                 \node[above = .7em of 1, inner sep = 1pt](1e1){\scriptsize $f(1)$};
                                 \node[above = .7em of 2, inner sep = 1pt](2e2){\scriptsize $f(2)$};
                                 \node[above = .7em of 3, inner sep = 1pt](3e3){\scriptsize $f(3)$};
                                 \graph[use existing nodes, edges={color=black, pos = .5, earrow}, edge quotes={fill=white, inner sep=1pt,font= \scriptsize}]{
                                 1 -> 2 -> 3; 1 -> 3;
                                 2 -- [bend right] 2e2 ->[bend right] 2;
                                 1 -- [bend right] 1e1 ->[bend right] 1;
                                 3 -- [bend right] 3e3 ->[bend right] 3;
                                 };
                             \end{tikzpicture}.
        \end{align*}
        \vspace{-2ex}
        \caption{Illustrative example of \kl{loop extensions}.}
    \label{figure:graph-loop}
    \end{figure}
\end{defi}
Let $\glang_{\term[3]}(\term[1]) \defeq \set{\graph[2][f] \mid \mbox{$\graph[2] \in \glang(\term[1]) \ \land\  f \colon \domain{\graph[2]} \to \glang(\term[3]^{\lop})$}}$.
We have the following:
\begin{prop}\label{prop:glang-u}
    Let $\val \in \REL$.
    For all $\PCoR_{\set{\bl^{*}, \com{\eps}, \com{x}}}$ "terms" $\term[1]$ and $\term[3]$,
    we have $\hat{\val}(\Tr_{\term[3]^{\lop}}(\term[1])) = \hat{\val}(\glang_{\term[3]}(\term[1]))$.
\end{prop}
\begin{proof}
    By induction on $\term[1]$, in the same way as \Cref{prop:glang} \cite[Proposition 11]{nakamuraExistentialCalculiRelations2023}.
    (See \Cref{section:prop:glang-u} for details.)
\end{proof}
Intuitively, the translation $\Tr_{\term[3]^{\lop}}(\term[1])$ encodes the "loop hypothesis" $\term[3] \ge \eps$ ``\emph{partially}''.
Under \emph{certain conditions}, this translation corresponds to adding $\term[3] \ge \eps$ exactly,
in that we have the following:
\begin{align*}
    \label{equation:loop-HE} \algclass_{\term[3] \ge \eps} \models \term[1] \le \term[2] &\quad\Longleftrightarrow\quad \algclass \models \Tr_{\term[3]^{\lop}}(\term[1]) \le \term[2]. \tag{Loop HE}
\end{align*}
The direction $\Longleftarrow$ always holds by \Cref{prop:loop-transformation}~(\ref{prop:loop-transformation-sup}).
For the direction $\Longrightarrow$, we need some conditions.
Below are counterexamples for which the direction $\Longrightarrow$ fails:
\kldeactivatequotation
\begin{itemize}
    \item $\REL_{a b \ge \eps} \models \eps \le a a b b$, but $\REL \not\models \Tr_{(a b)^{\lop}}(\eps) \le a a b b$
    by the \kl{relational model}: \begin{tikzpicture}[baseline = -.5ex]
        \tikzstyle{mynodeg} = [mynode, fill= gray!20, draw, circle]
        \graph[grow right = .8cm, branch down = 4.ex, nodes={font = \scriptsize}]{
        {1/{}[mynodeg]} -!- {2/{}[mynodeg]}
        };
        \graph[use existing nodes, edges={color=black, pos = .5, earrow}, edge quotes={fill=white, inner sep=1pt,font= \scriptsize}]{
            1 ->["$a$", bend left] 2 ->["$b$", bend left] 1;
        };
    \end{tikzpicture}.
    \item $\REL_{\top \com{\eps} = \top, a \ge \eps} \models \eps \le \com{\eps} a \com{\eps}$, but $\REL_{\top \com{\eps} = \top} \not\models \Tr_{a^{\lop}}(\eps) \le \com{\eps} a \com{\eps}$ by the \kl{relational model}:
    \begin{tikzpicture}[baseline = -.5ex]
        \tikzstyle{mynodeg} = [mynode, fill= gray!20, draw, circle]
        \graph[grow right = .8cm, branch down = 4.ex, nodes={font = \scriptsize}]{
        {1/{}[mynodeg]} -!- {2/{}[mynodeg]}
        };
        \graph[use existing nodes, edges={color=black, pos = .5, earrow}, edge quotes={fill=white, inner sep=1pt,font= \scriptsize}]{
            1 ->["$a$", out = -150, in = 150, looseness = 15] 1;
        };
    \end{tikzpicture}.
\end{itemize}
\klactivatequotation
Below, we give some sufficient conditions for the equivalence \eqref{equation:loop-HE}.

\subsection{Condition 1: When All Graph Homomorphisms Are Surjective}\label{subsection:condition-1}
First, we give the following sufficient condition.
\begin{thm}\label{thm:hypothesis-elimination-using-graph-loops}
    Let $\algclass \subseteq \REL$.
    Let $\term[1]$, $\term[2]$, $\term[3]$ be $\PCoR_{\set{\bl^{*}, \com{\eps}, \com{x}}}$ "terms".
    Suppose that $\term[1]$ satisfies the following condition for $\algclass$ and $\term[3]$:
    \begin{description}
        \item[\label{cond:surj}(Surj)] for all $\val \colon \vsig \to \wp(A^2)$ in $\algclass$, $o, o' \in A$, $\graph[2] \in \glang(\term[1])$, and "graph homomorphisms" $f \colon \graph[2] \homo \gval(\val, o, o')$ such that $\triangle_{f(\domain{\graph[2]})} \subseteq \hat{\val}(\term[3])$,
        $f$ is "surjective".
    \end{description}
    We then have:
    \[\algclass_{\term[3] \ge \eps} \models \term[1] \le \term[2] \quad\Longleftrightarrow\quad \algclass \models \Tr_{\term[3]^{\lop}}(\term[1]) \le \term[2].\]
\end{thm}
\begin{proof}
    \proofcase{($\Longleftarrow$)}
    By \Cref{prop:loop-transformation}~(\ref{prop:loop-transformation-sup}).

    \proofcase{($\Longrightarrow$)}
    Let $\tuple{o, o'} \in \hat{\val}(\Tr_{\term[3]^{\lop}}(\term))$ where $\val \colon \vsig \to \wp(A^2)$ is in $\algclass$ and $o, o' \in A$.
    By \Cref{prop:glang-u}, there are a "graph" $\graph[2] \in \glang(\term[1])$ and a "graph homomorphism" $f \colon \graph[2] \homo \gval(\val, o, o')$ such that $\triangle_{f(\domain{\graph[2]})} \subseteq \hat{\val}(\term[3])$.
    By assumption, $f$ is \emph{"surjective"}, and thus $\val \models \term[3] \ge \eps$.
    By $\val \in \algclass_{\term[3] \ge \eps}$, we have $\val \models \Tr_{\term[3]^{\lop}}(\term) \le \term \le \term[2]$.
\end{proof}
Note that \ref{cond:surj} is (anti-)monotone with respect to $\term[3]$:
if $\algclass \models \term[3]' \le \term[3]$ and \ref{cond:surj} holds for $\term[3]$, then it holds also for $\term[3]'$.
It is also (anti-)monotone with respect to $\algclass$: if $\algclass_0 \subseteq \algclass$ and \ref{cond:surj} holds for $\algclass$, then it holds also for $\algclass_0$.
Moreover, \ref{cond:surj} is inherited as follows.
\begin{cor}\label{cor:condition-inheritance}
    Let $\algclass \subseteq \REL$.
    Let $\term[1]_0$, $\term[1]'$, $\term[2]$, $\term[3]$ be $\PCoR_{\set{\bl^{*}, \com{\eps}, \com{x}}}$ "terms".
    Suppose that $\term[1]_0$ satisfies \ref{cond:surj} for $\algclass$ and $\term[3]$
    and that $\algclass' \models \term[1]' \le \term[1]_0$ for some "submodel-closed" $\algclass' \supseteq \algclass$.
    We then have:
    \[\algclass_{\term[3] \ge \eps} \models \term[1]' \le \term[2] \quad\Longleftrightarrow\quad \algclass \models \Tr_{\term[3]^{\lop}}(\term[1]') \le \term[2].\]
\end{cor}
\begin{proof}
    By \Cref{thm:hypothesis-elimination-using-graph-loops},
    it suffices to show that
    $\term[1]'$ satisfies \ref{cond:surj} for $\algclass$ and $\term[3]$.
    Let $\val \colon \vsig \to \wp(A^2)$ be in $\algclass$, $o, o' \in A$, $\graph[2]' \in \glang(\term[1]')$, and $f' \colon \graph[2]' \homo \gval(\val, o, o')$ be a "graph homomorphism" such that $\triangle_{B} \subseteq \hat{\val}(\term[3])$, where $B \defeq f'(\domain{\graph[2]'})$.
    We have $\tuple{o, o'} \in \widehat{\val \restr B}(\term[1]')$ by $f'$ and $\val \restr B \in \algclass'$ (as $\algclass'$ is "submodel-closed"), so $\tuple{o, o'} \in \widehat{\val \restr B}(\term[1]_0)$.
    By \Cref{prop:glang}, there is a "graph homomorphism" $f$ from some "graph" $\graph[2] \in \glang(\term[1]_0)$ to $\gval(\val \restr B, o, o')$, which is also one to $\gval(\val, o, o')$ whose image is included in $B$.
    As $\triangle_{f(\domain{\graph[2]})} \subseteq \triangle_{B} \subseteq \hat{\val}(\term[3])$, by assumption, $f$ is "surjective"; thus $B = A$, and hence $f'$ is "surjective".
\end{proof}
Note that $\algclass$ itself is often not "submodel-closed" (e.g., $\algclass_{\tcyc}$ and $\algclass_{\tgr}$ below), so we cannot take $\algclass' = \algclass$ in general.
For the two counterexamples to the direction $\Longrightarrow$ of \eqref{equation:loop-HE} (given before \Cref{subsection:condition-1}), \ref{cond:surj} fails:
in both cases, $\glang(\eps)$ consists of the single-vertex "graph", and
the "graph homomorphism" mapping it to the left "vertex" $o$ of the given "relational model" satisfies $\tuple{o, o} \in \hat{\val}(\term[3])$ (for $\term[3] = a b$ and $\term[3] = a$, respectively) but is not "surjective".

Below we give some examples of \Cref{thm:hypothesis-elimination-using-graph-loops}.

\subsubsection{Example 1: On cycles}\label{section:loop-hypotheses-on-cycles}
Let 
\begin{align*}
\algclass_{\tcyc} &\defeq \REL_{a^{\smile} a \le \eps, a^* \ge \top},
\quad \tag*{(functionality of $a$, and strongly connected with respect to $a$)}\\
\term[1]_{\tcyc} &\defeq (a^{+})^{\lop}. \quad \tag*{(cycles of $a$ edges)}
\end{align*}
We then have the following.
\begin{lem}\label{lem:cycle-surjective}
    Let $\val \in \algclass_{\tcyc}$,
    let $o$ be an element of the universe of $\val$,
    and let $f$ be a "graph homomorphism" from a "graph" $\graph[2] \in \glang(\term[1]_{\tcyc})$ to $\gval(\val, o, o)$.
    Then, $f$ is "surjective".
\end{lem}
\begin{proof}
    Let $\val \colon \vsig \to \wp(X^2)$.
    For each $x \in f(\domain{\graph[2]})$, we have 
    $\card{\set{z \mid \tuple{x, z} \in \val(a)}} = 1$ by the functionality of $a$ (for $\le 1$) and the form of $\graph[2]$ (for $\ge 1$).
    From this, there (uniquely) exists a map $g \colon \nat \to X$ such that, for all $i \in \nat$, the following hold:
    \begin{gather*}
        g(0) = o, \qquad \tuple{g(i), g(i+1)} \in \val(a), \qquad g(i) \in f(\domain{\graph[2]}).
    \end{gather*}
    By the functionality of $a$,
    we have $\set{g(i) \mid i \in \nat} = \set{z \mid \tuple{o, z} \in \hat{\val}(a^*)}$.
    By $\val \models a^* \ge \top$, the right-hand side is equal to $X$, and thus $g$ is "surjective".
    Hence, $f$ is "surjective".
\end{proof}
Hence, we have obtained the following "hypothesis elimination", as an example of \Cref{thm:hypothesis-elimination-using-graph-loops}.
\begin{ex}\label{ex:cycle-surjective}
    By \Cref{cor:condition-inheritance,lem:cycle-surjective} (note that, in \ref{cond:surj} for $\term[1]_{\tcyc}$, we have $o = o'$, as the source and the target of each "graph" in $\glang(\term[1]_{\tcyc})$ coincide),
    for all $\algclass \subseteq \algclass_{\tcyc}$ and $\PCoR_{\set{\bl^{*}, \com{\eps}, \com{x}}}$ "terms" $\term[1]$, $\term[2]$, and $\term[3]$
    such that $\algclass' \models \term[1] \le \term[1]_{\tcyc}$ for some "submodel-closed" $\algclass' \supseteq \algclass$,
    we have:
    \[\algclass_{\term[3] \ge \eps} \models \term[1] \le \term[2] \quad\Longleftrightarrow\quad \algclass \models \Tr_{\term[3]^{\lop}}(\term[1]) \le \term[2].\]
\end{ex}

\subsubsection{Example 2: On periodic grids}\label{section:loop-hypotheses-on-grids}
We recall the "hypotheses" $\tgrid$ in \Cref{section:grids}.
Let
\begin{align*}
    \algclass_{\tgr} &\defeq \REL_{\tfunc^{\set{\aE, \aN, \aW, \aS}}, (\aE \union \aN \union \aW \union \aS)^* \ge \top},    \\
    \term[1]_{\tgr} &\defeq (((\aN^{+})^{\lop} \aE)^{+})^{\lop},\\
    \term[3]_{\tgr} &\defeq (\aE \aW)^{\lop} (\aW \aE)^{\lop} (\aN \aS)^{\lop} (\aS \aN)^{\lop} (\aS \aE \aN \aW)^{\lop}.
\end{align*}
Note that $(\algclass_{\tgr})_{\term[3]_{\tgr} \ge \eps} = \REL_{\tgrid}$ holds.
We then have the condition for \Cref{thm:hypothesis-elimination-using-graph-loops}, as follows.
This is the key lemma in \Cref{thm:undecidable-while-emptiness}.
\begin{lem}\label{lem:grid-surjective}
    Let $\val \in \algclass_{\tgr}$,
    let $o$ be an element of the universe of $\val$,
    and let $f$ be a "graph homomorphism" from a "graph" $\graph[2] \in \glang(\term[1]_{\tgr})$ to $\gval(\val, o, o)$ such that $\triangle_{f(\domain{\graph[2]})} \subseteq \hat{\val}(\term[3]_{\tgr})$.
    Then, $f$ is "surjective".
\end{lem}
\begin{proof}
    Let $\val \colon \vsig \to \wp(X^2)$.
    \begin{cla}\label{cla:function-and-converse'}
        Let $\tuple{a, b} \in \set{\tuple{\aE, \aW}, \tuple{\aW, \aE}, \tuple{\aN, \aS}, \tuple{\aS, \aN}}$.
        For all $\tuple{x, y} \in X^2$, if $x \in f(\domain{\graph[2]})$, then we have the following two properties:
        \begin{description}
            \item[\label{cla:function'}(Function)] $\card{\set{z \mid \tuple{x, z} \in \val(a)}} = 1$,
            \item[\label{cla:converse'}(Converse)] if $\tuple{x, y} \in \val(a)$, then $\tuple{y, x} \in \val(b)$.
        \end{description}
    \end{cla}
    \begin{claproof}
        By $\tuple{x, x} \in \hat{\val}((a b)^{\lop})$ ($\supseteq \hat{\val}(\term[3]_{\tgr})$),
        there is $z$ such that $\tuple{x, z} \in \val(a)$ and $\tuple{z, x} \in \val(b)$.
        As $\val(a)$ is a partial function (by $\val \models \tfunc^{\set{\aE, \aN, \aW, \aS}}$), $z$ is the unique $a$-successor of $x$.
        Thus, we have \textbf{\ref{cla:function'}} and \textbf{\ref{cla:converse'}}.
    \end{claproof}
    By the form of $\graph[2]$ (consisting of the $\aE$-cycle through the source and the $\aN$-cycles through its "vertices") and $f$,
    there is a map $g \colon \nat^2 \to X$ such that, for all $x, y \in \nat$, the following hold:
    \begin{align*}
        g(0, 0) &= o, & \tuple{g(x, 0), g(x+1, 0)} &\in \val(\aE),
        & \tuple{g(x, y), g(x, y+1)} &\in \val(\aN), & g(x, y) &\in f(\domain{\graph[2]}).
    \end{align*}
    Moreover, we have the following.
    \begin{cla}\label{cla:aE'}
        For all $x, y \in \nat$, we have $\tuple{g(x, y), g(x + 1, y)} \in \val(\aE)$.
    \end{cla}
    \begin{claproof}
        By induction on $y$.

        \proofcase{Case $y = 0$}
        Clear.

        \proofcase{Case $y + 1$}
        By \textbf{\ref{cla:converse'}} with $\tuple{g(x, y), g(x, y + 1)} \in \val(\aN)$,
        we have $\tuple{g(x, y + 1), g(x, y)} \in \val(\aS)$.
        By $\tuple{g(x, y + 1), g(x, y + 1)} \in \hat{\val}((\aS \aE \aN \aW)^{\lop})$ ($\supseteq \hat{\val}(\term[3]_{\tgr})$)
        with \textbf{\ref{cla:function'}} with respect to $\aS$, $\aE$, and $\aN$,
        and with $\tuple{g(x, y), g(x + 1, y)} \in \val(\aE)$ (IH) and $\tuple{g(x + 1, y), g(x + 1, y + 1)} \in \val(\aN)$,
        we have $\tuple{g(x + 1, y + 1), g(x, y + 1)} \in \val(\aW)$.
        By \textbf{\ref{cla:converse'}}, $\tuple{g(x, y + 1), g(x + 1, y + 1)} \in \val(\aE)$.
    \end{claproof}
    Moreover, by the form of $\graph[2]$,
    there are $h \ge 1$ and $v_x \ge 1$ such that $g(h, 0) = o$ and $g(x, v_x) = g(x, 0)$ for each $x$;
    hence, by \textbf{\ref{cla:function'}} and \textbf{\ref{cla:converse'}}, the set $\set{g(x, y) \mid x, y \in \nat}$ is closed under $\val(\aW)$ and $\val(\aS)$ as well.
    Thus, by \textbf{\ref{cla:function'}} with respect to $\aE$ and $\aN$,
    we have $\set{g(x, y) \mid x, y \in \nat}
    = \set{z \mid \tuple{o, z} \in \hat{\val}((\aE \union \aN)^*)}
    = \set{z \mid \tuple{o, z} \in \hat{\val}((\aE \union \aN \union \aW \union \aS)^*)}$.
    By $\val \models (\aE \union \aN \union \aW \union \aS)^* \ge \top$, the right-hand side is equal to the universe $X$,
    and thus $g$ is "surjective".
    Hence, $f$ is "surjective".
\end{proof}

\begin{ex}\label{ex:grid-surjective}
    By \Cref{cor:condition-inheritance,lem:grid-surjective} (note that, in \ref{cond:surj} for $\term[1]_{\tgr}$, we have $o = o'$, as the source and the target of each "graph" in $\glang(\term[1]_{\tgr})$ coincide),
    for all $\algclass \subseteq \algclass_{\tgr}$ and $\PCoR_{\set{\bl^{*}, \com{\eps}, \com{x}}}$ "terms" $\term[1]$, $\term[2]$, and $\term[3]$ such that
    $\algclass' \models \term[1] \le \term[1]_{\tgr}$ for some "submodel-closed" $\algclass' \supseteq \algclass$ and
    $\algclass \models \term[3] \le \term[3]_{\tgr}$,
    we have:
    \[\algclass_{\term[3] \ge \eps} \models \term[1] \le \term[2] \quad\Longleftrightarrow\quad \algclass \models \Tr_{\term[3]^{\lop}}(\term[1]) \le \term[2].\]
\end{ex}

\subsection{Condition 2: Taking Submodels for Surjectivity}\label{subsection:condition-2}
Additionally, we give another sufficient condition for eliminating "loop hypotheses",
by taking "submodels" to obtain a "surjective" map.
\begin{thm}\label{thm:hypothesis-elimination-using-graph-loops-2}
    Let $\algclass \subseteq \REL$ and let $\term[1]$, $\term[2]$, and $\term[3]$ be $\PCoR_{\set{\bl^{*}, \com{\eps}, \com{x}}}$ "terms".
    When $\algclass$ is "submodel-closed" and $\algclass \models \term[3]^{\lop} \le \Tr_{\term[3]^{\lop}}(\term[3]^{\lop})$, 
    we have:
    \[\algclass_{\term[3] \ge \eps} \models \term[1] \le \term[2] \quad\Longleftrightarrow\quad \algclass \models \Tr_{\term[3]^{\lop}}(\term[1]) \le \term[2].\]
\end{thm}
\begin{proof}
    \proofcase{($\Longleftarrow$)}
    By \Cref{prop:loop-transformation}~(\ref{prop:loop-transformation-sup}).

    \proofcase{($\Longrightarrow$)}
    Let $\val \colon \vsig \to \wp(A^2)$ be in $\algclass$.
    Suppose $\tuple{x, y} \in \hat{\val}(\Tr_{\term[3]^{\lop}}(\term[1]))$.
    We inductively define $B_n \subseteq A$ (where $n \ge 0$).

    \proofcase{Case $n = 0$}
    Let $\graph[2] \in \glang(\term[1])$ be a "graph" such that
    there is a "graph homomorphism" $h \colon \graph[2] \homo \gval(\val, x, y)$
    such that $\triangle_{h(\domain{\graph[2]})} \subseteq \hat{\val}(\term[3]^{\lop})$.
    Such an $\graph[2]$ exists by \Cref{prop:glang-u}.
    We then let $B_0 \defeq h(\domain{\graph[2]})$.

    \proofcase{Case $n + 1$}
    For each $z \in B_n$, by $\tuple{z, z} \in \hat{\val}(\term[3]^{\lop})$, we have $\tuple{z, z} \in \hat{\val}(\Tr_{\term[3]^{\lop}}(\term[3]^{\lop}))$.
    Let $\graph[3]_{z} \in \glang(\term[3]^{\lop})$ be a "graph" such that
    there is a "graph homomorphism" $h_z \colon \graph[3]_{z} \homo \gval(\val, z, z)$
    satisfying $\triangle_{h_z(\domain{\graph[3]_{z}})} \subseteq \hat{\val}(\term[3]^{\lop})$.
    Such a $\graph[3]_{z}$ exists by \Cref{prop:glang-u}.
    We then let $B_{n+1} \defeq \bigcup_{z \in B_{n}} h_{z}(\domain{\graph[3]_{z}})$.

    Finally, let $B \defeq \bigcup_{n \ge 0} B_n$.
    We then have:
    \begin{itemize}
        \item $\val \restr B \models \term[3] \ge \eps$ (by the construction of $B$: for each $n \ge 0$ and $z \in B_n$, we have $\tuple{z, z} \in \widehat{\val \restr B}(\term[3]^{\lop})$ by $h_z$, as $h_z(\domain{\graph[3]_{z}}) \subseteq B_{n+1} \subseteq B$),
        and hence $\val \restr B \in \algclass_{\term[3] \ge \eps}$ (as $\algclass$ is "submodel-closed"),
        \item $\tuple{x, y} \in \widehat{\val \restr B}(\term[1])$ (by the construction of $B_0$).
    \end{itemize}
    Thus, by the assumption $\algclass_{\term[3] \ge \eps} \models \term[1] \le \term[2]$, we have $\tuple{x, y} \in \widehat{\val \restr B}(\term[2])$,
    and hence $\tuple{x, y} \in \hat{\val}(\term[2])$, as $\val \restr B$ is a "submodel" of $\val$.
    Hence, this completes the proof.
\end{proof}
Below we give some examples of \Cref{thm:hypothesis-elimination-using-graph-loops-2}.
\begin{ex}\label{ex:size-1}
    Suppose that,
    for all $\val \in \algclass$,
    elements $o$ of the universe of $\val$,
    "graphs" $\graph[2] \in \glang(\term[3]^{\lop})$,
    and "graph homomorphisms" $f \colon \graph[2] \homo \gval(\val, o, o)$,
    the map $f$ maps all "vertices" to the same "vertex" (i.e., $\card{f(\domain{\graph[2]})} = 1$).
    Then, we have $\algclass \models \term[3]^{\lop} \le \Tr_{\term[3]^{\lop}}(\term[3]^{\lop})$ by the map $f$.
    Thus, when $\algclass \subseteq \REL$ is "submodel-closed", 
    we have:
    \[\algclass_{\term[3] \ge \eps} \models \term[1] \le \term[2] \quad\Longleftrightarrow\quad \algclass \models \Tr_{\term[3]^{\lop}}(\term[1]) \le \term[2].\]
    (In particular,
    when $\algclass \models \term[3]^{\lop} = \term[3]$,
    the above is equivalent to $\algclass \models \Tr_{\term[3]}(\term[1]) \le \term[2]$
    (\Cref{prop:loop-transformation}~(\ref{prop:loop-transformation-monotone})).)
\end{ex}

\begin{ex}\label{ex:tests}
    We recall the sets $\psig$ and $\ttest$ in \Cref{section:PWP-with-loop}.
    Using \Cref{ex:size-1}, we can give another hypothesis elimination for the hypotheses $\ttest$.
    Let $\term[1], \term[2]$ be $\PCoR_{\set{\bl^{*}, \com{\eps}, \com{x}}}$ "terms" in which $\com{p}$ does not occur for $p \in \psig$.
    As the "primitive tests" not occurring in $\term[1] \le \term[2]$ do not affect its satisfaction, we may assume that $\psig = \set{p_0, \dots, p_{n-1}, \tilde{p}_0, \dots, \tilde{p}_{n-1}}$ (consisting of $2n$ distinct elements).
    Let us consider the following hypotheses (the set is similar to \cite[Section 4.2]{pousToolsCompletenessKleene2024}):
    \begin{align*}
        \ttest'_0 &\defeq \set{ p_i \tilde{p}_i \le \emp \mid i < n},\\
        \ttest'_1 &\defeq \set{p \le \eps \mid p \in \psig},\\
        \ttest'_2 &\defeq \set{p_i \union \tilde{p}_i \ge \eps \mid i < n}.
    \end{align*}
    Also, let $\ttest'_{S} \defeq \bigcup_{j \in S} \ttest'_{j}$ for $S \subseteq \set{0, 1, 2}$ and let $\ttest' \defeq \ttest'_{\set{0, 1, 2}}$ (clearly, $\REL_{\ttest'} = \REL_{\ttest}$ holds).
    Let $\term[3]_0 \defeq \top (\sum_{i = 0}^{n-1} p_i \tilde{p}_i) \top$.
    By combining hypothesis eliminations, we can eliminate all the hypotheses as follows:
    \begin{align*}
        &\REL_{\ttest'} \models \term[1] \le \term[2] \Leftrightarrow \REL_{\ttest'_{\set{0, 1, 2}}} \models \term[1] \le \term[2] \\
        &\Leftrightarrow \REL_{\ttest'_{\set{1, 2}}} \models \term[1] \le \term[2] \union \term[3]_0 \tag{\Cref{prop:Hoare-hypotheses}} \\
        &\Leftrightarrow \REL_{\ttest'_{\set{1}}} \models \Tr_{\bigcompo_{i = 0}^{n-1} (p_i \union \tilde{p}_i)}(\term[1]) \le \term[2] \union  \term[3]_0 \tag{\Cref{ex:size-1}}                                                      \\
        &\Leftrightarrow \REL \models
        \Tr_{\bigcompo_{i = 0}^{n-1} (p_i \union \tilde{p}_i)}(\term[1]) [p_0^{\lop}, \dots, p_{n-1}^{\lop}, \tilde{p}_0^{\lop}, \dots, \tilde{p}_{n-1}^{\lop}/p_0, \dots, p_{n-1}, \tilde{p}_0, \dots, \tilde{p}_{n-1}] \le \term[2] \union \term[3]_0. \tag*{(\Cref{ex:substitution-id-by-loop}, applied to each $p_i$ and $\tilde{p}_i$)}
    \end{align*}
    Thus, we can encode "tests" in $\PCoR_{\set{\bl^{*}}}$.
    Hence, we can reduce the "equational theory" of $\PCoR_{\set{\bl^{*}}}$ with "tests" to the "equational theory" of $\PCoR_{\set{\bl^{*}}}$ (without "tests"), via "hypothesis eliminations".
    Additionally, we can replace $\bl^{\lop}$ with $\bl^{\dom}$ by using \Cref{ex:substitution-id-by-dom} in the last step instead, where the substitution is then applied to both sides.
    The same equivalences also hold for $\DREL$ (as the conditions used above also hold for $\DREL$; in the last step, \Cref{ex:substitution-id-by-loop} also works for $\DREL$, as $\DREL[p_i^{\lop}/p_i] \subseteq \DREL$).
\end{ex}

\begin{ex}[another sufficient condition, cf.\ \Cref{ex:size-1}]\label{ex:symmetric-graph}
    Let $\term[3]$ be a $\PCoR_{\set{\bl^{*}, \com{\eps}, \com{x}}}$ "term" such that $\glang(\term[3]^{\lop})$ is closed under changing the source/target "vertex".
    We then have $\REL \models \term[3]^{\lop} \le \Tr_{\term[3]^{\lop}}(\term[3]^{\lop})$.
    (Proof: Suppose $\tuple{x, x} \in \hat{\val}(\term[3]^{\lop})$.
    Let $h$ be a "graph homomorphism" from a "graph" $\graph[2] \in \glang(\term[3]^{\lop})$ to $\gval(\val, x, x)$.
    For each $z \in h(\domain{\graph[2]})$,
    by letting $\graph[3]_{z}$ be the "graph" $\graph[2]$ in which the source/target "vertex" has been changed to a "vertex" $w$ such that $h(w) = z$,
    there exists a "graph homomorphism" $h_z$ from $\graph[3]_{z}$ to $\gval(\val, z, z)$.
    By $\graph[3]_{z} \in \glang(\term[3]^{\lop})$, we have $\tuple{x, x} \in \hat{\val}(\Tr_{\term[3]^{\lop}}(\term[3]^{\lop}))$ (\Cref{prop:glang-u}).)
    Thus, by \Cref{thm:hypothesis-elimination-using-graph-loops-2},
    when $\algclass \subseteq \REL$ is "submodel-closed",
    we have:
    \[\algclass_{\term[3] \ge \eps} \models \term[1] \le \term[2] \quad\Longleftrightarrow\quad \algclass \models \Tr_{\term[3]^{\lop}}(\term[1]) \le \term[2].\]
    For instance, $\term[3] = aa$, $aaa$, $a^{+}$, and $(a \union b)^{+}$ satisfy the condition.
\end{ex}

\subsection{Eliminating Loop Hypotheses in Propositional While Programs}\label{section:hypothesis-elimination-loop-while}
Similarly to $\Tr$ (\Cref{defi:loop-transformation}), we define the translation $\TrPWP$ on "$\PWPlop$ terms", as follows:
\begin{defi}\label{defi:loop-transformation-GKAT}
    For "$\PWPlop$ terms" $r$ and $\term[3]$,
    let $\intro*\TrPWP_{\term[3]}(r)$ be the "$\PWPlop$ term" defined as follows:
    \begin{align*}
        \TrPWP_{\term[3]}(r) &\defeq \term[3] r \term[3] \ \mbox{ for $r \in \Tests \cup \asig$}, \\
        \TrPWP_{\term[3]}(\term[1] \compo \term[2]) &\defeq \TrPWP_{\term[3]}(\term[1]) \compo \TrPWP_{\term[3]}(\term[2]), &
        \TrPWP_{\term[3]}(\term[1] \unionb{b} \term[2]) &\defeq \TrPWP_{\term[3]}(\term[1]) \unionb{b} \TrPWP_{\term[3]}(\term[2]), \\
        \TrPWP_{\term[3]}(\term[1]^{\starb{b}}) &\defeq \TrPWP_{\term[3]}(\term[1])^{\starb{b}} \term[3], &
        \TrPWP_{\term[3]}(\term[1]^{\lop}) &\defeq \TrPWP_{\term[3]}(\term[1])^{\lop}.
    \end{align*}
    (The first equation applies whenever $r \in \Tests \cup \asig$.)
\end{defi}
\begin{prop}\label{prop:loop-transformation-GKAT}
    For all "$\PWPlop$ terms" $\term[1]$ and $\term[3]$,
    we have $\REL_{\ttest} \models \TrPWP_{\term[3]^{\lop}}(\term[1]) = \Tr_{\term[3]^{\lop}}(\term[1])$.
\end{prop}
\begin{proof}
    By easy induction on $\term[1]$,
    using that $\REL \models \Tr_{\term[3]^{\lop}}(\term[1]) = \term[3]^{\lop} \Tr_{\term[3]^{\lop}}(\term[1]) \term[3]^{\lop}$ and that $\REL_{\ttest} \models \term[3]^{\lop} b = b \term[3]^{\lop}$ and $\REL_{\ttest} \models \Tr_{\term[3]^{\lop}}(b) = \term[3]^{\lop} b$ for "tests" $b$ (as $\term[3]^{\lop}$ and $b$ are subsets of the identity relation).
\end{proof}
Hence, we can use the hypothesis eliminations given in \Cref{subsection:condition-1,subsection:condition-2} also inside "$\PWPlop$ terms".
\section{Undecidability of Deterministic Propositional While Programs with Loop}\label{section:while}
For a class $\algclass$, the ""emptiness problem"" on $\algclass$ is the following problem:
given a "term" $\term$, does $\algclass \models \term \le \emp$?
In particular, the "emptiness problem" of deterministic "$\PWPlop$ terms" is the one on $\DREL_{\ttest}$ for "$\PWPlop$ terms".
In this section, we prove that this problem is $\PiOne$-hard (\Cref{thm:undecidable-while-emptiness}).

\subsection{Encoding Periodic Tilings Using Loop Hypotheses}
Recall that $\asig$ and $\psig$ are the sets of "actions" and "primitive tests", respectively (\Cref{section:PWP-with-loop}).
For the reduction in this section, assume that $\asig$ contains at least four "actions" and that $\psig$ contains infinitely many involution pairs.
Let $\aE, \aN, \aW, \aS \in \asig$ be pairwise distinct.
Let $\mathcal{D} = \tuple{C, H, V}$ be a "domino system" where $C \subseteq \psig$ (and $\tilde{c} \not\in C$ for all $c \in C$).
For $c \in C$, let $H_c \defeq \set{d \mid \tuple{c, d} \in H}$ and $V_c \defeq \set{d \mid \tuple{c, d} \in V}$.
Let $C = \set{c_1, \dots, c_n}$ where $c_1, \dots, c_n$ are pairwise distinct.
Let $\Gamma_{\mathcal{D}}$ be the set of "equations" in \Cref{figure:hypotheses-GKAT}.
\begin{figure*}[ht]
    \begin{align*}
        \label{equation:serial-converse'} a b &\ge \eps  &&          \mbox{for $\tuple{a, b} \in \set{\tuple{\aE, \aW}, \tuple{\aW, \aE}, \tuple{\aN, \aS}, \tuple{\aS, \aN}}$} \span \tag{2-cycle}  \\
        \label{equation:cycle'} \aS \aE \aN \aW &\ge \eps  \span \tag{4-cycle}                                                                                                                      \\
        \label{equation:connected} (\aE \union \aN \union \aW \union \aS)^* &\ge \top \tag{strongly connected}\\
        \label{equation:color'} \term[3]_{c_1} \unionb{c_1} (\term[3]_{c_2} \unionb{c_2} \dots (\term[3]_{c_n} \unionb{c_n} \emp) \dots) &\ge \eps     \tag{col$\ge$1 $+$ HV-con}\\
        \label{equation:color'-le} \tilde{c}_i \union \tilde{c}_j &\ge \eps  && \mbox{for distinct $i, j \in \range{1, n}$}    \tag{col$\le$1}
    \end{align*}
    Here, $\term[3]_{c} \defeq (\aE (\sum H_c) \aW)^{\lop} (\aN (\sum V_c) \aS)^{\lop}$ for $c \in C$.
    \caption{The \kl{equation} set $\Gamma_{\mathcal{D}}$, where $\mathcal{D}$ is a \kl{domino system} $\tuple{C, H, V}$.}
    \label{figure:hypotheses-GKAT}
\end{figure*}

Intuitively, each "equation" in $\Gamma_{\mathcal{D}}$ plays the following role, where each "primitive test" $c \in C$ is read as ``the current "vertex" has the color $c$'':
\begin{itemize}
    \item (\ref{equation:serial-converse'}), (\ref{equation:cycle'}), and (\ref{equation:connected}):
    these "equations" correspond to the "hypotheses" $\tgrid$ (\Cref{section:grids}):
    (\ref{equation:cycle'}) and (\ref{equation:connected}) are the same as those in $\tgrid$,
    and (\ref{equation:serial-converse'}) is (\ref{equation:bijective}) with $=$ weakened to $\ge$, where the remaining direction $a b \le \eps$ follows from $\tfunc$ (see below).
    Hence, for every "valuation" in $\DREL_{\ttest, \Gamma_{\mathcal{D}}}^{\fin}$, the grid $\Z^2$ can be mapped into it $\tuple{h,v}$-periodically for some $h, v \ge 1$ (\Cref{prop:grid}).
    \item (\ref{equation:color'}):
    each "vertex" has some color $c$ (col$\ge$1), and
    the east neighbor (reached by $\aE$ and returned by $\aW$) has a color in $H_c$ and the north neighbor (reached by $\aN$ and returned by $\aS$) has a color in $V_c$ (HV-con), i.e., $\term[3]_{c} \ge \eps$ holds at the "vertex".
    Here, the left-hand side is a nested $\mathbf{if}$-$\mathbf{then}$-$\mathbf{else}$, which executes $\term[3]_{c_i}$ for the first color $c_i$ that holds (and $\emp$ if no color holds).
    \item (\ref{equation:color'-le}):
    each "vertex" has at most one color (at each "vertex", $c_i$ or $c_j$ fails for any distinct colors $c_i$ and $c_j$).
\end{itemize}

In the following, we consider the class $\DREL_{\ttest, \Gamma_{\mathcal{D}}}^{\fin}$.
Note that $\DREL_{\ttest, \Gamma_{\mathcal{D}}}^{\fin} \subseteq \REL^{\fin}_{\tgrid}$ (recall \Cref{section:grids}) holds, because $\DREL_{ab \ge \eps} \models ab \le \eps$ holds by $\tfunc$ as follows:
\begin{align*}
    \DREL_{ab \ge \eps} \models a b \le (a b)^{\smile} a b = b^{\smile} a^{\smile} a b \le \eps.
\end{align*}
Let $p, q \in \psig$ be "primitive tests" such that $p, \tilde{p}, q, \tilde{q}$ are pairwise distinct and do not occur in $\Gamma_{\mathcal{D}}$, and let 
\[\term[1]_{\tgr}' \;\defeq\; (p ((q \aN^{\plusb{\tilde{q}}})^{\lop} \aE)^{\plusb{\tilde{p}}})^{\lop}.\]
The "term" $\term[1]_{\tgr}'$ is a deterministic version of the "term" $\term[1]_{\tgr}$ (\Cref{section:loop-hypotheses-on-grids}), traversing the whole grid:
starting from a "vertex" where $p$ holds,
it repeatedly goes around the current column along $\aN$ (from a "vertex" where $q$ holds until returning to it) and then moves to the next column along $\aE$,
until it returns to the starting "vertex".
The additional "primitive tests" $p$ and $q$ only mark where these iterations stop, so that the "term" is a "$\PWPlop$ term".
We then have the following.
\begin{lem}\label{lem:connectivity'}
    Let $\mathcal{D}$ be a "domino system".
    We then have:
    \[\mbox{$\mathcal{D}$ has a "periodic tiling"} \quad\iff\quad \mbox{$\DREL_{\ttest, \Gamma_{\mathcal{D}}}^{\fin} \not\models \term[1]_{\tgr}' \le \emp$}.\]
\end{lem}
\begin{proof}
    \proofcase{($\Leftarrow$)}
    By assumption, $\DREL_{\ttest, \Gamma_{\mathcal{D}}}^{\fin} \neq \emptyset$.
    Let $\val \in \DREL_{\ttest, \Gamma_{\mathcal{D}}}^{\fin}$.
    By $\val \in \REL^{\fin}_{\tgrid}$,
    let $h, v, f$ be the ones obtained by \Cref{prop:grid}.
    Let $\tau \colon \Z^2 \to C$ be the unique map such that $\tuple{f(x, y), f(x, y)} \in \val(\tau(x, y))$, by (\ref{equation:color'}) and (\ref{equation:color'-le}).
    By (\ref{equation:color'}) with \Cref{prop:grid}, we have $\tuple{\tau(x, y), \tau(x + 1, y)} \in H$ and $\tuple{\tau(x, y), \tau(x, y + 1)} \in V$.
    Hence, $\tau$ is a "tiling".
    Also, $\tau$ is "periodic" by $h$ and $v$.
    
    \proofcase{($\Rightarrow$)}
    Let $\tau$ be an $\tuple{h,v}$-"periodic tiling".
    By the form of $\val_{\mathcal{D}, \tau}$ (recall \Cref{defi:valuation-of-periodic-tiling}), we have $\val_{\mathcal{D}, \tau} \in \DREL_{\ttest, \Gamma_{\mathcal{D}}}^{\fin}$ by a routine verification.
    Let \[\val'(a) \defeq \begin{cases}
            \triangle_{\set{\tuple{0, 0}}} & (a = p)    \\
            \triangle_{\set{\tuple{x, 0} \mid x \in \range{0, h - 1}}} & (a = q)    \\
            \hat{\val}_{\mathcal{D}, \tau}(\eps) \setminus \triangle_{\set{\tuple{0, 0}}} & (a = \tilde{p})    \\
            \hat{\val}_{\mathcal{D}, \tau}(\eps) \setminus \triangle_{\set{\tuple{x, 0} \mid x \in \range{0, h - 1}}} & (a = \tilde{q})    \\
            \val_{\mathcal{D}, \tau}(a)                                                              & (\mbox{otherwise}).
        \end{cases}\]
    Since $p, \tilde{p}, q, \tilde{q}$ do not occur in $\Gamma_{\mathcal{D}}$ (and $\val'$ still satisfies $\ttest$), we have $\val' \in \DREL_{\ttest, \Gamma_{\mathcal{D}}}^{\fin}$.
    Also, we have $\tuple{\tuple{0, 0}, \tuple{0, 0}} \in \hat{\val}'((p \aE^{\plusb{\tilde{p}}})^{\lop})$ by traversing "vertices" of $\tuple{x, 0}$ as $x = 0, 1, \dots, h - 1, 0$.
    Similarly, for each $x$, we have 
    $\tuple{\tuple{x, 0}, \tuple{x, 0}} \in \hat{\val}'((q \aN^{\plusb{\tilde{q}}})^{\lop})$ by traversing "vertices" of $\tuple{x, y}$ as $y = 0, 1, \dots, v - 1, 0$.
    Thus, we have $\tuple{\tuple{0, 0}, \tuple{0, 0}} \in \hat{\val}'((p ((q \aN^{\plusb{\tilde{q}}})^{\lop} \aE)^{\plusb{\tilde{p}}})^{\lop}) = \hat{\val}'(\term[1]_{\tgr}')$.
    Hence, $\val' \not\models \term[1]_{\tgr}' \le \emp$.
\end{proof}

\subsection{Eliminating Loop Hypotheses}
By summarizing the "loop hypotheses" occurring in $\Gamma_{\mathcal{D}}$, we define
\begin{align*}
\term[3]_{\mathcal{D}} \defeq{} & (\aW \aE)^{\lop} (\aS \aN)^{\lop} (\aS \aE \aN \aW)^{\lop} (\term[3]_{c_1} \unionb{c_1} (\term[3]_{c_2} \unionb{c_2} \dots (\term[3]_{c_n} \unionb{c_n} \emp) \dots))^{\lop} \compo \bigcompo_{1 \le i < j \le n} (\tilde{c}_j \unionb{c_i} \eps)^{\lop}.
\end{align*}
Note that $\DREL_{\ttest, \Gamma_{\mathcal{D}}}^{\fin} = \DREL^{\fin}_{\ttest, (\aE \union \aN \union \aW \union \aS)^* \ge \top, \term[3]_{\mathcal{D}} \ge \eps}$, as $\REL_{\ttest} \models (\tilde{c}_j \unionb{c_i} \eps)^{\lop} = \tilde{c}_i \union \tilde{c}_j$,
$\REL_{\ttest} \models (\term[3]_{c_1} \unionb{c_1} (\term[3]_{c_2} \unionb{c_2} \dots (\term[3]_{c_n} \unionb{c_n} \emp) \dots))^{\lop} \le (\aE \aW)^{\lop} (\aN \aS)^{\lop}$ (so the "loops" $(\aE \aW)^{\lop}$ and $(\aN \aS)^{\lop}$ are omitted in $\term[3]_{\mathcal{D}}$),
and multiple "loop hypotheses" can be combined into one (see the beginning of \Cref{section:Hypothesis-elimination-using-graph-loops}).
We can then eliminate the hypothesis $\term[3]_{\mathcal{D}} \ge \eps$, as follows.

\kldeactivatequotation
\begin{figure}[t]
    \begin{align*}
        &  \begin{tikzpicture}[baseline = -1.5cm]
            \graph[grow down= 1.25cm, branch right = 1.25cm, nodes={, font = \scriptsize}]{
            {00/{}[mynode, draw, circle, yshift = 0ex], 01/{}[mynode, draw, circle, yshift = -2ex], 01R/, 02/{$\dots$}[xshift=-1cm]  , 03L/[xshift=-2cm], 03/{}[mynode, draw, circle, xshift=-2cm, yshift = 2ex]} -!-
            {10/{$\vdots$}[yshift = 0ex], 11/{$\vdots$}[, yshift = -1ex], 11R/, 12/{$\dots$}[xshift=-1cm,]  , 13L/[xshift=-2cm], 13/{$\vdots$}[xshift=-2cm, yshift = 2ex]} -!-
            {20/{}[mynode, draw, circle], 21/{}[mynode, draw, circle], 21R/, 22/{$\dots$}[xshift=-1cm]  , 23L/[xshift=-2cm], 23/{}[mynode, draw, circle, xshift=-2cm]} -!-
            {30/{}[mynode, draw, circle], 31/{}[mynode, draw, circle], 31R/, 32/{$\dots$}[xshift=-1cm]  , 33L/[xshift=-2cm], 33/{}[mynode, draw, circle, xshift=-2cm]}
            };
            \node[left = .3cm of 30](src){};
            \node[below = .3cm of 30](tgt){};
            \graph[use existing nodes, edges={color=black, pos = .5, earrow}, edge quotes={fill = white, inner sep=1pt,font= \scriptsize}]{
            src -> 30 -> tgt;
            \foreach \y in {3}{
            (\y0) ->["$\aE$"] (\y1) -> ["$\aE$"] (\y1R); (\y3L) -> ["$\aE$"] (\y3)->["$\aE$", bend left = 15] (\y0);
            };
            \foreach \x in {0, 1, 3}{
            (0\x) <-["$\aN$"] (1\x) <-["$\aN$"] (2\x) <- ["$\aN$"] (3\x) <-["$\aN$", bend left = 15] (0\x);
            };
            \foreach \x in {0, 1, 3}{\foreach \y in {0, 2, 3}{
            \y\x ->["$\graph_{\term[3]_{\mathcal{D}}}$", out = 0, in = 60, looseness = 25] \y\x;
            };};
            };
        \end{tikzpicture} 
        \qquad
        \begin{aligned}
           & \mbox{where each $\graph_{\term[3]_{\mathcal{D}}}$ is of}\\
           & \mbox{the form} \left(\begin{tikzpicture}[baseline = {([yshift=-.5ex]current bounding box.center)}]
                \graph[grow right= 1.cm, branch down = 1.cm, nodes={, font = \scriptsize}]{
                {/,W/[mynode, draw, circle], /} -!- {N/[mynode, draw, circle],O/[mynode, draw, circle],S/[mynode, draw, circle]} -!- {/,E/[mynode, draw, circle]}
                };
                \node[below left = .3cm and .2cm of O](src){};
                \node[below left = .2cm and .3cm of O](tgt){};
                \node[above right = -.2cm and .8cm of E,mynode, draw, circle](E'){};
                \node[below left = .1cm and .1cm of S, mynode, draw, circle](S'){};
                \path let \p1=(E'), \p2= (S') in node at (\x1,\y2) (S'E')[mynode, draw, circle]{};
                \graph[use existing nodes, edges={color=black, pos = .5, earrow}, edge quotes={fill = white, inner sep=1pt,font= \scriptsize}]{
                src -> O -> tgt;
                O ->["$\tilde{c}_1, \dots, \tilde{c}_{i-1},c_i$"{xshift = -2.2em, yshift = .2em}, out = 150, in = 120, looseness = 50] O;
                O ->["$\aW$", bend right = 20] W -> ["$\aE$", bend right = 20] O;
                O ->["$\aE$", bend right = 20] E -> ["$\aW$", bend right = 20] O;
                O ->["$\aN$", bend right = 20] N -> ["$\aS$", bend right = 20] O;
                O ->["$\aS$", bend right = 20] S -> ["$\aN$", bend right = 20] O;
                E ->["$d$", out = -60, in = -120, looseness = 20] E;
                N ->["$d'$", out = -30, in = 30, looseness = 20] N;
                O ->["$\aS$",pos = .7, bend right] S' -> ["$\aE$"] S'E' -> ["$\aN$", pos = .4] E' ->["$\aW$",pos = .3, bend right = 35] O;
                };
            \end{tikzpicture}\right)\\
            & \mbox{for some $i \in \range{1, n}$, $d \in H_{c_i}$, and $d' \in V_{c_i}$.}
        \end{aligned}
    \end{align*}
    \caption{The form of \kl{graphs} in $\protect\glang_{\term[3]_{\mathcal{D}}}(\term[1]_{\tgr}')$
    where the pairs of \kl{vertices} connected with some $\psig$-labelled edge are all merged (each such pair is mapped to the same \kl{vertex} in $\DREL_{\protect\ttest}$).
    The loops labelled $p, \tilde{p}, q, \tilde{q}$ in the grid and the loops of the form $(\tilde{c}_j \unionb{c_i} \eps)^{\lop}$ are omitted.}
    \label{figure:graphs}
\end{figure}
\klactivatequotation
\begin{thm}\label{thm:undecidable-while-emptiness}
    The "emptiness problem" of "$\PWPlop$ terms" on $\DREL_{\ttest}$ (and also on $\DREL_{\ttest}^{\fin}$) is $\PiOne$-complete.
\end{thm}
\begin{proof}
    \proofcase{(In $\PiOne$)}
    By \Cref{prop:upper-bound} with the fact that $\DREL_{\ttest}$ is "submodel-closed" (and by \Cref{section:graph-languages} for $\DREL_{\ttest}^{\fin}$).

    \proofcase{(Hardness)}
    We have:
    \begin{align*}
         & \mbox{$\mathcal{D}$ does not have a "periodic tiling"}                                                                 \\
         & \Leftrightarrow\quad \DREL^{\fin}_{\ttest, (\aE \union \aN \union \aW \union \aS)^* \ge \top, \term[3]_{\mathcal{D}} \ge \eps} \models \term[1]_{\tgr}' \le \emp \tag{\Cref{lem:connectivity'}}                                                                                 \\
         & \Leftrightarrow\quad \DREL^{\fin}_{\ttest, (\aE \union \aN \union \aW \union \aS)^* \ge \top} \models \TrPWP_{\term[3]_{\mathcal{D}}}(\term[1]_{\tgr}') \le \emp   \tag{\Cref{ex:grid-surjective} ($\bigstar_1$) with \Cref{prop:loop-transformation-GKAT}}                                                                                           \\
         & \Leftrightarrow\quad \DREL_{\ttest} \models \TrPWP_{\term[3]_{\mathcal{D}}}(\term[1]_{\tgr}') \le \emp.   \tag{\Cref{ex:surjective} ($\bigstar_2$)}
    \end{align*}
    In the first step, we also use the equality between $\DREL_{\ttest, \Gamma_{\mathcal{D}}}^{\fin}$ and $\DREL^{\fin}_{\ttest, (\aE \union \aN \union \aW \union \aS)^* \ge \top, \term[3]_{\mathcal{D}} \ge \eps}$ noted above;
    in the second step, we also use $\REL \models \term[3]_{\mathcal{D}}^{\lop} = \term[3]_{\mathcal{D}}$ (as $\term[3]_{\mathcal{D}}$ is a composition of "loops").
    For ($\bigstar_1$),
    by
    $\DREL^{\fin}_{\ttest, (\aE \union \aN \union \aW \union \aS)^* \ge \top} \subseteq \algclass_{\tgr}$,
    $\DREL_{\ttest} \models \term[1]_{\tgr}' \le \term[1]_{\tgr}$ (by \Cref{prop:annotation-elimination} and $\REL_{\ttest} \models p, q \le \eps$; where $\DREL_{\ttest}$ is "submodel-closed"), and
    $\DREL_{\ttest} \models \term[3]_{\mathcal{D}} \le \term[3]_{\tgr}$, we can apply \Cref{ex:grid-surjective}.
    For ($\bigstar_2$),
    we have \[\Sur_{\glang_{\term[3]_{\mathcal{D}}}(\term[1]_{\tgr}')}(\DREL_{\ttest})
    \subseteq \DREL^{\fin}_{\ttest, (\aE \union \aN \union \aW \union \aS)^* \ge \top},\]
    because the "graph" of each "valuation" in $\Sur_{\glang_{\term[3]_{\mathcal{D}}}(\term[1]_{\tgr}')}(\DREL_{\ttest})$ is strongly connected with respect to $\aE$, $\aN$, $\aW$, and $\aS$
    (as "tests" are interpreted as subsets of the identity relation; see also \Cref{figure:graphs}).
    Thus, we can apply \Cref{ex:surjective} with $\glang_{\term[3]_{\mathcal{D}}}(\term[1]_{\tgr}')$ in place of $\glang(\TrPWP_{\term[3]_{\mathcal{D}}}(\term[1]_{\tgr}'))$,
    as they have the same semantics on $\DREL_{\ttest}$ by \Cref{prop:glang-u,prop:loop-transformation-GKAT} (see the note at the end of \Cref{ex:surjective}).
    Hence, we can reduce from the complement of "the periodic domino problem", which is $\PiOne$-hard (\Cref{prop:domino-tiling-undecidable}).
    Additionally, $\DREL_{\ttest} \models \TrPWP_{\term[3]_{\mathcal{D}}}(\term[1]_{\tgr}') \le \emp$
    iff $\DREL_{\ttest}^{\fin} \models \TrPWP_{\term[3]_{\mathcal{D}}}(\term[1]_{\tgr}') \le \emp$, because $\Sur_{\glang_{\term[3]_{\mathcal{D}}}(\term[1]_{\tgr}')}(\DREL_{\ttest})
    \subseteq \DREL^{\fin}_{\ttest} \subseteq \DREL_{\ttest}$.
\end{proof}

\subsection{Comparison to the Undecidability of SDIPDL}
In this subsection, we compare \Cref{thm:undecidable-while-emptiness} with the undecidability of "SDIPDL" due to Goldblatt and Jackson \cite{goldblattWellstructuredProgramEquivalence2012},
where "SDIPDL" is strict deterministic propositional dynamic logic (SDPDL) \cite{halpernPropositionalDynamicLogic1983} extended with intersection.
After recalling "SDIPDL" as a fragment of $\PCoR_{\set{\bl^{*}, \bl^{-}}}$,
we show that \Cref{thm:undecidable-while-emptiness} implies that the "validity problem" is still $\PiOne$-complete for a fragment of "SDIPDL" without nested modal operators (\Cref{cor:undecidable-PDL}).
Finally, we summarize related complexity results (\Cref{table:complexity-SDIPDL}).

We first give the syntax of "SDIPDL" \cite{goldblattWellstructuredProgramEquivalence2012} as a fragment of $\PCoR_{\set{\bl^{*}, \bl^{-}}}$.
We use the following abbreviations (the "antidomain operator" $\bl^{\adom}$ is used for encoding the negation of formulas):
\begin{align*}
    \introalign\cfalse &\defeq \emp,&
    \introalign\ctrue &\defeq \eps,&
    \fml[1] \land \fml[2] &\defeq \fml[1] \compo \fml[2],&
    \fml[1] \lor \fml[2] &\defeq \fml[1] \union \fml[2],&
    \lnot \fml[1] &\defeq \fml[1]^{\adom},\\
    \langle \term \rangle \fml[1] &\defeq (\term \compo \fml[1])^{\dom},&
    [\term] \fml[1] &\defeq (\term \compo \fml[1]^{\adom})^{\adom},&
    \term[1] \unionb{\fml} \term[2] &\defeq \fml \term[1] \union (\lnot \fml) \term[2],&
    \term[1]^{\starb{\fml}} &\defeq (\fml \term[1])^* (\lnot \fml), &
    \fml[1]? &\defeq \fml[1].
\end{align*}

The sets of ""SDIPDL"" formulas and programs, written $\intro*\sdF$ and $\intro*\sdP$,
are mutually defined as the minimal sets of $\PCoR_{\set{\bl^{*}, \bl^{-}}}$ "terms" satisfying the following:
\begin{gather*}
    \begin{prooftree}[separation = .5em]
        \hypo{p \in \psig}
        \infer1{p \in \sdF}
    \end{prooftree}
    \hspace{.8em}
    \begin{prooftree}[separation = .5em]
        \hypo{\mathstrut}
        \infer1{\cfalse \in \sdF}
    \end{prooftree}
    \hspace{.8em}
    \begin{prooftree}[separation = .5em]
        \hypo{\mathstrut}
        \infer1{\ctrue \in \sdF}
    \end{prooftree}
    \hspace{.8em}
    \begin{prooftree}[separation = .5em]
        \hypo{\fml[1] \in \sdF}
        \hypo{\fml[2] \in \sdF}
        \infer2{\fml[1] \land \fml[2] \in \sdF}
    \end{prooftree}
    \hspace{.8em}
    \begin{prooftree}[separation = .5em]
        \hypo{\fml[1] \in \sdF}
        \hypo{\fml[2] \in \sdF}
        \infer2{\fml[1] \lor \fml[2] \in \sdF}
    \end{prooftree}
    \hspace{.8em}
    \begin{prooftree}[separation = .5em]
        \hypo{\fml[1] \in \sdF}
        \infer1{\lnot \fml[1] \in \sdF}
    \end{prooftree}
    \hspace{.8em}
    \begin{prooftree}[separation = .5em]
        \hypo{\term[1] \in \sdP}
        \hypo{\fml[1] \in \sdF}
        \infer2{\langle \term[1] \rangle \fml[1] \in \sdF}
    \end{prooftree}
    \hspace{.8em}
    \begin{prooftree}[separation = .5em]
        \hypo{\term[1] \in \sdP}
        \hypo{\fml[1] \in \sdF}
        \infer2{[\term[1]]\fml[1] \in \sdF}
    \end{prooftree}\\
    \begin{prooftree}[separation = .5em]
        \hypo{a \in \asig}
        \infer1{a \in \sdP}
    \end{prooftree}
    \hspace{.8em}
    \begin{prooftree}[separation = .5em]
        \hypo{\term[1] \in \sdP}
        \hypo{\term[2] \in \sdP}
        \infer2{\term[1] \compo \term[2] \in \sdP}
    \end{prooftree}
    \hspace{.8em}
    \begin{prooftree}[separation = .5em]
        \hypo{\fml \in \sdF}
        \hypo{\term[1] \in \sdP}
        \hypo{\term[2] \in \sdP}
        \infer3{\term[1] \unionb{\fml} \term[2] \in \sdP}
    \end{prooftree}
    \hspace{.8em}
    \begin{prooftree}[separation = .5em]
        \hypo{\fml \in \sdF}
        \hypo{\term[1] \in \sdP}
        \infer2{\term[1]^{\starb{\fml}} \in \sdP}
    \end{prooftree}
    \hspace{.8em}
    \begin{prooftree}[separation = .5em]
        \hypo{\term[1] \in \sdP}
        \hypo{\term[2] \in \sdP}
        \infer2{\term[1] \cap \term[2] \in \sdP}
    \end{prooftree}
    \hspace{.8em}
    \begin{prooftree}[separation = .5em]
        \hypo{\fml \in \sdF}
        \infer1{\fml[1]? \in \sdP}
    \end{prooftree}.
\end{gather*}
Additionally, we use the following abbreviations:
$\fml[1] \to \fml[2] \defeq \lnot \fml[1] \lor \fml[2]$ and $\fml[1] \leftrightarrow \fml[2] \defeq (\fml[1] \to \fml[2]) \land (\fml[2] \to \fml[1])$ in "SDIPDL" formulas.

For "SDIPDL" formulas $\fml$, we write $\DREL_{\ttest} \models \fml$ if $\DREL_{\ttest} \models \fml = \ctrue$ holds.
We then have $\DREL_{\ttest} \models \term \le \emp$ iff $\DREL_{\ttest} \models [\term] \cfalse$ (and similarly for $\DREL_{\ttest}^{\fin}$).
Thus, \Cref{thm:undecidable-while-emptiness} can be rephrased as follows
(where each "$\PWPlop$ term" is regarded as an "SDIPDL" program by reading each occurrence of $\tilde{p}$ as $\lnot p$, $\unionb{b}$ and $\bl^{\starb{b}}$ as those of "SDIPDL", and $\term^{\lop}$ as $\term \cap \ctrue?$,
which is sound as $\REL_{\ttest} \models \tcompl{b} = \lnot b$ for every "test" $b$ (in particular, $\tilde{p} = \lnot p$) and $\REL \models \term^{\lop} = \term \cap \ctrue?$).
\begin{cor}\label{cor:undecidable-PDL}
    The "validity problem" (and also the "finite validity problem") is $\PiOne$-complete for formulas of the form $[\term] \cfalse$ in "SDIPDL" where $\term$ is a "$\PWPlop$ term" \emph{without any modal operators}.
\end{cor}
Hence, \Cref{thm:undecidable-while-emptiness} refines the undecidability result of "SDIPDL" \cite{goldblattWellstructuredProgramEquivalence2012} with respect to the nesting of modal operators.
Notably, the diamond operator $\langle \bl \rangle$ does not appear in the fragment of \Cref{cor:undecidable-PDL}.

\begin{rem}[On nested modal operators in {\cite{goldblattWellstructuredProgramEquivalence2012}}]\label{rem:nested-modal-operators}
    The reduction of \cite{goldblattWellstructuredProgramEquivalence2012} (in particular, the formula ``$\rho_{1}$'' \cite[p.\ 5]{goldblattWellstructuredProgramEquivalence2012}) contains nested box and diamond operators, by unfolding notations.
    Note that:
    (1) the formula ``$\mathsf{fix}(\term)$'' \cite[p.\ 3]{goldblattWellstructuredProgramEquivalence2012} is equivalent to the formula $\langle\term^{\lop}\rangle \ctrue$ on deterministic structures,
    and (2) the formula $[\term^*] \fml$ abbreviates the formula $[\term^{\starb{\fml}}] \cfalse$,
    so the used formula contains "tests" containing modal operators.
\end{rem}

\begin{rem}[On the complexity difference from {\cite{goldblattWellstructuredProgramEquivalence2012}}]\label{rem:difference}
    The "validity problem" is $\PiOne$-complete for the fragment of \Cref{cor:undecidable-PDL},
    whereas the "validity problem" is $\PiOneOne$-complete for full "SDIPDL"
    \cite[Theorem 4.1]{goldblattWellstructuredProgramEquivalence2012}
    (where the upper bound follows from the translation into transitive closure logic \cite[Section 6.2]{gradelUndecidabilityResultsTwovariable1999}).
    The $\PiOne$ upper bound (\Cref{prop:upper-bound}) is derived from the finite model property (\Cref{prop:bounded-model-property}~(\ref{prop:bounded-model-property-1})).
    On the other hand, "SDIPDL" does not have the finite model property.
    For instance, when we consider only deterministic structures,
    the formula $p \land [a][a^{*}] \lnot p \land [a^{*}] \langle (a b)^{\lop} \rangle \ctrue$ is satisfiable only on infinite structures (where $[\term^{*}]\fml$ abbreviates $[\term^{\starb{\fml}}]\cfalse$ as in \Cref{rem:nested-modal-operators}, and $\term^{\lop}$ is read as $\term \cap \ctrue?$).
    (Proof: Starting from a "vertex" $x_0$ where $p$ holds, the formula requires an $a$-path $x_0 x_1 x_2 \dots$ such that $b$ leads back from each $x_{i+1}$ to $x_i$, by $\langle (a b)^{\lop} \rangle \ctrue$ and the functionality of $a$.
    Since $b$ is functional, each $x_i$ is determined by $x_{i+1}$, so the first "vertex" revisited by the path (if any) must be $x_0$;
    however, $x_0$ is never revisited, as $p$ holds at $x_0$ but fails at $x_1, x_2, \dots$, by $[a][a^{*}] \lnot p$.
    Hence, $x_0, x_1, \dots$ are pairwise distinct, and the formula has no finite model.
    In contrast, the formula is satisfiable on $\DREL_{\ttest}$, e.g., on $\nat$ where $a$ is the successor, $b$ is the predecessor, and $p$ holds only at $0$.)
    Goldblatt and Jackson \cite[Theorem 4.4]{goldblattWellstructuredProgramEquivalence2012} showed that the set of finitely valid "SDIPDL" formulas is not recursively enumerable.
    The "finite validity problem" for full "SDIPDL" is $\PiOne$-complete.
    Hardness follows from \Cref{cor:undecidable-PDL}, and the upper bound follows by enumerating finite models.
\end{rem}

Finally, we summarize related complexity results in \Cref{table:complexity-SDIPDL}.
(Note that each "equation" $\term[1] = \term[2]$ can be embedded into the formula $[\term[1]]p \leftrightarrow [\term[2]]p$ preserving validity, where $p$ is a fresh "primitive test"; see, e.g., \cite[Section 5]{fischerPropositionalDynamicLogic1979}.
Similarly, each $\term[1] = \emp$ can be embedded into the formula $[\term[1]]\cfalse$ preserving validity.
Also, as an easy consequence of \Cref{prop:glang}, $\REL_{\ttest} \models \term[1] = \term[2]$ iff $\DREL_{\ttest} \models \term[1] = \term[2]$ for "propositional while programs" $\term[1]$ and $\term[2]$; see \Cref{section:DREL-KAT} for a proof.)
\begin{table}
    \newcommand{\unimportant}[1]{\textcolor{gray}{\footnotesize #1}}
    \centering
    \setlength{\tabcolsep}{4pt}
    \begin{tabular}{c||c|c|c}
        programs $\backslash$ formulas & \shortstack{$[\term[1]]\cfalse$\\ {\footnotesize (emptiness: $\term[1] = 0$)}} & \shortstack{$[\term[1]]p \leftrightarrow [\term[2]]p$\\ {\footnotesize (equivalence: $\term[1] = \term[2]$)}} & any \\
        \hline\hline
        $\PWP$ & \multicolumn{2}{c|}{$\coNP$-c\tablefootnote{Cf.\ \cite[Section 5]{kozenHoareLogicKleene2000} (``Cohen, E.\ 1999. Personal communication'' \cite[p.\ 73]{kozenHoareLogicKleene2000}, for KAT) and \cite[Proposition 5.11]{smolkaGuardedKleeneAlgebra2019} ($\coNP$-hardness and the $\PSPACE$ upper bound).
        For the $\coNP$ upper bound, note that $\DREL_{\ttest} \models \term[1] \le \term[2]$ iff $\REL_{\ttest} \models \term[1] \le \term[2]$ (\Cref{prop:KAT-DREL-and-REL}) iff $\GS(\term[1]) \subseteq \GS(\term[2])$, where $\GS(\cdot)$ denotes the guarded string language \cite[Theorem 6]{kozenKleeneAlgebraTests1996}.
        Hence, it suffices to guess a guarded string $w \in \GS(\term[1]) \setminus \GS(\term[2])$ of polynomial length and to check it in polynomial time by computing the (Brzozowski-style) derivatives \cite[Figure 2]{schmidGuardedKleeneAlgebra2021} of the GKAT expressions $\term[1]$ and $\term[2]$ along $w$.
        Indeed, a shortest such $w$ has polynomial length, as these derivatives are deterministic (i.e., each guarded string determines a unique sequence of pairs of derivatives of $\term[1]$ and $\term[2]$) and the derivatives of $\term[1]$ (resp.\ $\term[2]$) range over polynomially many expressions.}} & \unimportant{$\PSPACE$-c (SDPDL; \cite{halpernPropositionalDynamicLogic1983})}\\
        \cline{1-1}\Xcline{2-4}{1.2pt}%
        $\PWPlop$ & \multicolumn{2}{c|}{\multirow{2}{*}{$\PiOne$-c (\Cref{thm:undecidable-while-emptiness}, \Cref{prop:upper-bound})}} & \unimportant{$\PiOneOne$-c (\cite{goldblattWellstructuredProgramEquivalence2012})}\\
        \cline{1-1}\cline{4-4}
        $\PWPcap$ & \multicolumn{2}{c|}{} & \unimportant{$\PiOneOne$-c (SDIPDL; \cite{goldblattWellstructuredProgramEquivalence2012})}
    \end{tabular}
    \caption{Complexity of the validity problem on the class of \kl{valuations} of \emph{deterministic} \kl{relational models} $\DREL_{\protect\ttest}$
    (where the number of \kl{primitive tests} is not fixed).
    Here, $\protect\PWP$, $\protect\PWPlop$, and $\protect\PWPcap$ denote \kl{propositional while programs}, those with the \kl{graph loop operator}, and those with intersection, respectively;
    $\term[1]$ and $\term[2]$ do not contain $\langle \bl\rangle$ or $[\bl]$, and $p$ is a \kl{primitive test} such that $p$ and $\tilde{p}$ do not occur in $\term[1]$ and $\term[2]$.
    In the column ``any'', arbitrary formulas are allowed (in particular, \kl{tests} may contain modal operators, as in SDPDL and \kl{SDIPDL}).
    Gray entries are known results.}
    \label{table:complexity-SDIPDL}
\end{table}
\section{On the Number of Primitive Tests and Actions}\label{section:on-the-number-of-primitive-tests-and-actions}
In this section,
we consider the fragments of "$\PWPlop$ terms" obtained by restricting the number of "primitive tests" and "actions".
(In the sequel, we count the number of "primitive tests" as the number of involution pairs of "primitive tests";
for instance, the "term" $p \tilde{p} q$ has $2$ "primitive tests" $p$ / $\tilde{p}$ and $q$.)

We say that a "$\PWPlop$ term" $\term$ is ""primitive"" if every "test" occurring in $\term$ is $\eps$, $\emp$, or a "primitive test", and every guard $b$ of $\unionb{b}$ and $\bl^{\starb{b}}$ is a "primitive test".
Namely, the set of ""primitive $\PWPlop$ terms"" is the minimal subset $\Progs$ of $\PCoR_{\set{\bl^{*}}}$ satisfying the following:
\begin{gather*}
    \begin{prooftree}[separation = .5em]
        \hypo{\mathstrut}
        \infer1{\eps \in \Progs}
    \end{prooftree}
    \hspace{.4em}
    \begin{prooftree}[separation = .5em]
        \hypo{\mathstrut}
        \infer1{\emp \in \Progs}
    \end{prooftree}
    \hspace{.4em}
    \begin{prooftree}[separation = .5em]
        \hypo{p \in \psig}
        \infer1{p \in \Progs}
    \end{prooftree}
    \hspace{.4em}
    \begin{prooftree}[separation = .5em]
        \hypo{x \in \asig}
        \infer1{x \in \Progs}
    \end{prooftree}
    \hspace{.4em}
    \begin{prooftree}[separation = .5em]
        \hypo{\term[1] \in \Progs}
        \hypo{\term[2] \in \Progs}
        \infer2{\term[1] \compo \term[2] \in \Progs}
    \end{prooftree}
    \hspace{.4em}
    \begin{prooftree}[separation = .5em]
        \hypo{\term[1] \in \Progs}
        \hypo{\term[2] \in \Progs}
        \hypo{p \in \psig}
        \infer3{\term[1] \unionb{p} \term[2] \in \Progs}
    \end{prooftree}
    \hspace{.4em}
    \begin{prooftree}[separation = .5em]
        \hypo{\term[1] \in \Progs}
        \hypo{p \in \psig}
        \infer2{\term[1]^{\starb{p}} \in \Progs}
    \end{prooftree}
    \hspace{.4em}
    \begin{prooftree}[separation = .5em]
        \hypo{\term[1] \in \Progs}
        \infer1{\term[1]^{\lop} \in \Progs}
    \end{prooftree}.
\end{gather*}
Note that the "term" $\TrPWP_{\term[3]_{\mathcal{D}}}(\term[1]_{\tgr}')$ given in \Cref{thm:undecidable-while-emptiness} is a "primitive $\PWPlop$ term", by replacing each $\sum \set{p_1, \dots, p_n}$ with $\eps \unionb{p_1} (\eps \unionb{p_2} (\dots (\eps \unionb{p_{n}} \emp) \dots))$ (which is equivalent to $\sum \set{p_1, \dots, p_n}$ on $\REL_{\ttest}$).
For simplicity, we consider only "primitive $\PWPlop$ terms" in this section.
This restriction will be useful in \Cref{defi:reducing-primitive-tests}.

\subsection{Reducing to One Primitive Test}\label{section:reducing-primitive-tests}
In this subsection, we strengthen the $\PiOne$-hardness result of \Cref{thm:undecidable-while-emptiness}, by fixing the number of "primitive tests".\footnote{%
When the number of colors is fixed, the number of instances of the "periodic domino problem" is finite, and hence the "problem@periodic domino problem" is trivially decidable.
Hence, we cannot immediately bound the number of "primitive tests" from the reduction of \Cref{thm:undecidable-while-emptiness}.}
Let $P_0 = \set{p_1, \dots, p_n}$ be a set of "primitive tests" such that $\tilde{p}_i \notin P_0$ for each $i$.
Let $m \defeq 2n + 3$ (a sufficiently large integer).
Let $\tP$ be a "fresh" "primitive test".
Let $\Uorg$ and $\Kgo{i}$ (for $i \in \range{1, n}$) be the "primitive $\PWPlop$ terms" given by:
\[\intro*\Uorg \defeq (\tP \aE \tP \aW)^{\lop}, \qquad \intro*\Kgo{i} \defeq \Uorg (\aE^{2i+1} \aW^{2i+1})^{\lop} \aE^{2i+1}.\]
We will use $\Uorg$ for distinguishing original "vertices" and newly added "vertices" in the transformation of "valuations" given later (cf.\ \Cref{figure:encoding}),
and $\Kgo{i}$ for moving to the $(2i+1)$-th "vertex" to the east after checking that we can come back by $\aW^{2i+1}$.

We then transform "terms" over "primitive tests" $P_0$ to "those@terms" over one "primitive test" $\tP$, as follows.
\begin{defi}\label{defi:reducing-primitive-tests}
    For "primitive $\PWPlop$ terms" $\term$ over "actions" $\set{\aE,\aN,\aW,\aS}$ and over "primitive tests" $P_0 \cup \set{\tilde{p}_1, \dots, \tilde{p}_n}$,
    let $\intro*\TrOne(\term)$ be the "primitive $\PWPlop$ term" (over one "primitive test" $\tP$) defined by:
    \begin{align*}
        \TrOne(\eps)                       & \defeq \Uorg, &
        \TrOne(\emp)                       & \defeq \emp, \\
        \TrOne(a)                       & \defeq \Uorg a^{m} \Uorg \mbox{ for $a \in \set{\aE,\aN,\aW,\aS}$}, &
        \TrOne(b)                       & \defeq \Uorg (\aE^{2i+1} \tP_{b} \aW^{2i+1})^{\lop}, \\
        \TrOne(\term[1] \compo \term[2]) & \defeq \TrOne(\term[1]) \compo \TrOne(\term[2]),  &
        \TrOne(\term[1] \unionb{b} \term[2]) & \defeq \Kgo{i} (\aW^{2i+1} \TrOne(\term[1]) \unionb{\tP_{b}} \aW^{2i+1} \TrOne(\term[2]) ),  \\
        \TrOne(\term[1]^{\lop})           & \defeq \TrOne(\term[1])^{\lop}, &
        \TrOne(\term[1]^{\starb{b}})          & \defeq \Kgo{i} (\aW^{2i+1} \TrOne(\term[1]) \Kgo{i})^{\starb{\tP_{b}}} \aW^{2i+1},
    \end{align*}
    where $b \in \set{p_i, \tilde{p}_i}$ for $i \in \range{1, n}$, $\tP_{p_i} \defeq \tP$, and $\tP_{\tilde{p}_i} \defeq \tilde{\tP}$.
\end{defi}
Note that $\TrOne(\term)$ is computable in polynomial time and that $\bl^{\starb{\tP}}$ or $\bl^{\starb{\tilde{\tP}}}$ occurs in $\TrOne(\term)$ only if $\bl^{\starb{b}}$ occurs in $\term$.
We then have the following.
\begin{lem}\label{lem:reducing-primitive-tests}
    Let $\algclass \in \set{\DREL_{\ttest}, \DREL^{\fin}_{\ttest}}$.
    For all "primitive $\PWPlop$ terms" $\term$ over the "actions" $\set{\aE, \aN, \aW, \aS}$ and the "primitive tests" in $P_0 \cup \set{\tilde{p}_1, \dots, \tilde{p}_n}$, we have:
    \[\algclass \models \term \le \emp \quad\Longleftrightarrow\quad \algclass \models \TrOne(\term) \le \emp.\]
\end{lem}
\begin{proof}
    \proofcase{($\Leftarrow$)}
    Suppose that $\hat{\val}(\term) \neq \emptyset$ for some $\val \colon \vsig \to \wp(X^2)$ in $\algclass$.
    Let $Y \defeq X \uplus (X \times \set{\aE, \aN, \aW, \aS} \times \range{1, m - 1})$, where $\uplus$ denotes the disjoint union
    (i.e., for each "action" $a$, each "vertex" $x \in X$ is given a chain $\tuple{x, a, 1}, \dots, \tuple{x, a, m - 1}$ of new "vertices").
    We write $\tuple{x, a, 0}$ for $x$ itself.
    Let $\val' \colon \vsig \to \wp(Y^2)$ be the "valuation" defined by
    \begin{align*}
        \val'(a) &\defeq R_a \mbox{ for $a \in \set{\aE, \aN, \aS}$},\\
        \val'(\aW) &\defeq R_{\aW} \cup \set{\tuple{\tuple{x, \aE, k}, \tuple{x, \aE, k - 1}} \mid x \in X,\; k \in \range{1, m-1}},\\
        \val'(\tP) &\defeq
        \triangle_{X} \cup \triangle_{\set{\tuple{x, \aE, 1} \mid x \in X}}
        \cup \triangle_{\set{\tuple{x, \aE, 2i + 1} \mid i \in \range{1, n},\; \tuple{x, x} \in \val(p_i)}},\\
        \val'(\tilde{\tP}) &\defeq \triangle_{Y} \setminus \val'(\tP),
    \end{align*}
    where $R_a \defeq \set{\tuple{\tuple{x, a, k}, \tuple{x, a, k + 1}} \mid x \in X,\; k \in \range{0, m-2}} \cup \set{\tuple{\tuple{x, a, m - 1}, y} \mid \tuple{x, y} \in \val(a)}$
    (the chain from $x$ to the $a$-successor of $x$).
    The other "primitive tests" are interpreted so that $\ttest$ holds, and $\val'(z) \defeq \emptyset$ for the other "variables" $z$.
    Below (\Cref{figure:encoding}) is an illustration of the "graph" of the "valuation" $\val'$ (where some edges are omitted).
\kldeactivatequotation
    \begin{figure*}[h]
        \centering
        \begin{tabular}{ccc}
            (before) $\val$ & & (after) $\val'$\\
            \hfill
            \begin{tikzpicture}[baseline = .0cm]
                \graph[grow up= 1.cm, branch right = 1.6cm, nodes={font=\tiny}]{
                {0/{}[mynode, fill = black, inner sep = .2em], 1/{}[mynode, fill = black, inner sep = .2em]}
                };
                \graph[use existing nodes, edges={color=black, pos = .5, earrow}, edge quotes={fill = white, inner sep=1pt,font= \scriptsize}]{
                (0) ->["$\aE$"] (1); 
                0 ->["$p_2, p_3$", out = 60, in = 120, looseness = 250, loop] 0;
                };
            \end{tikzpicture}
            &                
            \hfill
            {\Large $\leadsto$}
            \hfill
            &
            \begin{tikzpicture}[baseline = .0cm]
                \graph[grow up= 1.cm, branch right = 8.cm, nodes={font=\tiny}]{
                {0/{}[mynode, fill = black, inner sep = .2em], 1/{}[mynode, fill = black, inner sep = .2em]}
                };
                \path (0) edge [draw = white, opacity = 0] node[pos= 0.1, mynode, fill = black, inner sep = 1pt, opacity = 1](011){}(1);            
                \path (0) edge [draw = white, opacity = 0] node[pos= 0.2, mynode, fill = black, inner sep = 1pt, opacity = 1](012){}(1);            
                \path (0) edge [draw = white, opacity = 0] node[pos= 0.3, mynode, fill = black, inner sep = 1pt, opacity = 1](013){}(1);            
                \path (0) edge [draw = white, opacity = 0] node[pos= 0.4, mynode, fill = black, inner sep = 1pt, opacity = 1](014){}(1);            
                \path (0) edge [draw = white, opacity = 0] node[pos= 0.5, mynode, fill = black, inner sep = 1pt, opacity = 1](015){}(1);            
                \path (0) edge [draw = white, opacity = 0] node[pos= 0.6, mynode, fill = black, inner sep = 1pt, opacity = 1](016){}(1);            
                \path (0) edge [draw = white, opacity = 0] node[pos= 0.7, mynode, fill = black, inner sep = 1pt, opacity = 1](017){}(1);            
                \path (0) edge [draw = white, opacity = 0] node[pos= 0.8, mynode, fill = black, inner sep = 1pt, opacity = 1](018){}(1);            
                \path (0) edge [draw = white, opacity = 0] node[pos= 0.9, opacity = 1](019){$\dots$}(1);            
                \draw [decorate,decoration={brace,amplitude=5pt,mirror,raise=4ex}]
  ($(0.west)+(0,.4)$) -- ($(011.east)+(0,.4)$) node[below, midway, yshift=-5ex, font=\tiny]{(for $\Uorg$)};
                \node[below = 1.em of 012,align=left, font=\tiny](012b){(blank)};
                \path (012b) edge[dotted] ($(012.south)+(0,-.1)$); 
                \node[below = 1.em of 013,align=left, font=\tiny](013b){($p_1$)};
                \path (013b) edge[dotted] ($(013.south)+(0,-.1)$); 
                \node[below = 1.em of 014,align=left, font=\tiny](014b){(blank)};
                \path (014b) edge[dotted] ($(014.south)+(0,-.1)$); 
                \node[below = 1.em of 015,align=left, font=\tiny](015b){($p_2$)};
                \path (015b) edge[dotted] ($(015.south)+(0,-.1)$); 
                \node[below = 1.em of 016,align=left, font=\tiny](016b){(blank)};
                \path (016b) edge[dotted] ($(016.south)+(0,-.1)$); 
                \node[below = 1.em of 017,align=left, font=\tiny](017b){($p_3$)};
                \path (017b) edge[dotted] ($(017.south)+(0,-.1)$); 
                \node[below = 1.em of 018,align=left, font=\tiny](018b){(blank)};
                \path (018b) edge[dotted] ($(018.south)+(0,-.1)$); 
                \graph[use existing nodes, edges={color=black, pos = .5, earrow}, edge quotes={fill = white, inner sep=1pt,font= \scriptsize}]{
                (0) ->["$\aE$"] (011) ->["$\aE$"] (012) ->["$\aE$"] (013) ->["$\aE$"] (014)
                ->["$\aE$"] (015) ->["$\aE$"] (016) ->["$\aE$"] (017) ->["$\aE$"] (018);
                0 ->["$\tP$", out = 60, in = 120, looseness = 200, loop] 0;
                011 ->["$\tP$", out = 60, in = 120, looseness = 200, loop] 011;
                015 ->["$\tP$", out = 60, in = 120, looseness = 200, loop] 015;
                017 ->["$\tP$", out = 60, in = 120, looseness = 200, loop] 017;
                };
            \end{tikzpicture}
            \hfill
            \phantom{ }
        \end{tabular}
        \caption{Illustration of $\val'$ in the proof of \Cref{lem:reducing-primitive-tests} (the small vertices are new; some edges are omitted).}
        \label{figure:encoding}
    \end{figure*}
\klactivatequotation
    Clearly, $\val' \in \algclass$ (note that $Y$ is finite if $X$ is finite).
    As $2i + 2 \le 2n + 2 = m - 1$, the $\aE$-successor $\tuple{x, \aE, 2i + 2}$ of each $\tuple{x, \aE, 2i + 1}$ ($i \in \range{0, n}$) does not satisfy $\tP$;
    hence, $\hat{\val}'(\Uorg) = \triangle_{X}$ and $\hat{\val}'(\Kgo{i}) = \set{\tuple{x, \tuple{x, \aE, 2i + 1}} \mid x \in X}$.
    Using them, by easy induction on $\term[2]$, we can show that
    $\hat{\val}'(\TrOne(\term[2])) = \hat{\val}(\term[2])$ (note that $X \subseteq Y$) for all subterms $\term[2]$ of $\term$.
    (For instance, for $\term[2] = \term[2]_1 \unionb{p_i} \term[2]_2$: from $x \in X$, the "term" $\Kgo{i}$ moves to $\tuple{x, \aE, 2i + 1}$, which satisfies $\tP$ iff $\tuple{x, x} \in \val(p_i)$, and then $\aW^{2i+1}$ moves back to $x$.)
    Thus, $\hat{\val}'(\TrOne(\term)) = \hat{\val}(\term) \neq \emptyset$.

    \proofcase{($\Rightarrow$)}
    Suppose that $\hat{\val}'(\TrOne(\term)) \neq \emptyset$ for some $\val' \colon \vsig \to \wp(Y^2)$ in $\algclass$.
    Let $\val \colon \vsig \to \wp(Y^2)$ be the "valuation" defined by
    $\val(a) \defeq \hat{\val}'(\TrOne(a))$ for $a \in \set{\aE, \aN, \aW, \aS}$,
    $\val(p_i) \defeq \hat{\val}'(\TrOne(p_i))$ and $\val(\tilde{p}_i) \defeq \triangle_{Y} \setminus \val(p_i)$ for $i \in \range{1, n}$,
    and $\val(z) \defeq \val'(z)$ for the other "variables" $z$.
    Then, $\val \in \algclass$, as each $\hat{\val}'(\TrOne(a))$ is a partial function and $\hat{\val}'(\TrOne(p_i)) \subseteq \triangle_{Y}$.
    We show $\hat{\val}'(\TrOne(\term[2])) \subseteq \hat{\val}(\term[2])$ for all subterms $\term[2]$ of $\term$, by induction on $\term[2]$.
    The cases of $\eps$, $\emp$, $p_i$, $a$, $\compo$, and $\bl^{\lop}$ are clear (by $\hat{\val}'(\Uorg) \subseteq \triangle_{Y}$).

    \proofcase{Case $\term[2] = \term[2]_1 \unionb{p_i} \term[2]_2$}
    Let $\tuple{z, z'} \in \hat{\val}'(\TrOne(\term[2]))$.
    By $\Kgo{i}$, there is $w$ such that $\tuple{z, w} \in \hat{\val}'(\aE^{2i+1})$ and $\tuple{w, z} \in \hat{\val}'(\aW^{2i+1})$,
    where $w$ is unique and $z$ is the unique $\aW^{2i+1}$-successor of $w$, as $\val'(\aE)$ and $\val'(\aW)$ are partial functions.
    If $\tuple{w, w} \in \val'(\tP)$, then $\tuple{z, z} \in \hat{\val}'(\TrOne(p_i)) = \val(p_i)$ and $\tuple{z, z'} \in \hat{\val}'(\TrOne(\term[2]_1)) \subseteq \hat{\val}(\term[2]_1)$ (by IH).
    Otherwise, $\tuple{z, z} \in \val(\tilde{p}_i)$ (by the uniqueness of $w$) and $\tuple{z, z'} \in \hat{\val}(\term[2]_2)$ (by IH).
    Thus, $\tuple{z, z'} \in \hat{\val}(\term[2]_1 \unionb{p_i} \term[2]_2)$.

    \proofcase{Case $\term[2] = \term[2]_1^{\starb{p_i}}$} Similarly, by applying the argument above to each iteration:
    the guard $\tP$ in $\TrOne(\term[2]) = \Kgo{i} (\aW^{2i+1} \TrOne(\term[2]_1) \Kgo{i})^{\starb{\tP}} \aW^{2i+1}$ is always evaluated at a "vertex" $w$ such that $\tuple{z, w} \in \hat{\val}'(\Kgo{i})$ for some $z$, as in the case above.

    The cases of $\tilde{p}_i$, $\unionb{\tilde{p}_i}$, and $\bl^{\starb{\tilde{p}_i}}$ are shown in the same way (where $\tP$ and $\tilde{\tP}$ are swapped).
    Thus, $\emptyset \neq \hat{\val}'(\TrOne(\term)) \subseteq \hat{\val}(\term)$.
\end{proof}

Using this reduction, we can strengthen \Cref{thm:undecidable-while-emptiness} by reducing the number of "primitive tests" to one, as follows.
This positively settles the conjecture left in the conference version \cite[Section 6]{nakamuraUndecidabilityPositiveCalculus2024}.
\begin{thm}\label{thm:reducing-primitive-tests}
    The "emptiness problem" on $\DREL_{\ttest}$ (resp.\ $\DREL_{\ttest}^{\fin}$) is 
    $\PiOne$-complete for "$\PWPlop$ terms" where the number of "actions" is (at most) $4$ and the number of "primitive tests" is $1$.
\end{thm}
\begin{proof}
    \proofcase{(In $\PiOne$)} By \Cref{thm:undecidable-while-emptiness}.

    \proofcase{(Hardness)}
    We recall the "primitive $\PWPlop$ term" $\TrPWP_{\term[3]_{\mathcal{D}}}(\term[1]_{\tgr}')$ (over the "actions" $\aE, \aN, \aW, \aS$) in the proof of \Cref{thm:undecidable-while-emptiness}.
    We then have:
    \begin{align*}
        \mbox{$\mathcal{D}$ does not have a "periodic tiling"}
        & \;\Leftrightarrow\; \DREL_{\ttest} \models \TrPWP_{\term[3]_{\mathcal{D}}}(\term[1]_{\tgr}') \le \emp \tag{proof of \Cref{thm:undecidable-while-emptiness}}\\
        & \;\Leftrightarrow\; \DREL_{\ttest} \models \TrOne(\TrPWP_{\term[3]_{\mathcal{D}}}(\term[1]_{\tgr}')) \le \emp. \tag{\Cref{lem:reducing-primitive-tests}}
    \end{align*}
    The same holds for $\DREL^{\fin}_{\ttest}$.
    Because $\TrOne(\TrPWP_{\term[3]_{\mathcal{D}}}(\term[1]_{\tgr}'))$ contains four "actions" $\aE, \aN, \aW, \aS$ and one "primitive test" $\tP$,
    this completes the proof.
\end{proof}
Moreover, we can reduce the number of "actions" to two, as follows.
Let $\aA, \aB$ be two distinct "actions" and let $\intro*\sigmaAB$ be the substitution mapping $\aE, \aN, \aS, \aW$ to $\aA \aA$, $\aA \aB$, $\aB \aA$, $\aB \aB$, respectively.
\Cref{figure:reducing-actions} illustrates the transformation of "valuations" used in the proof of \Cref{lem:reducing-actions}.
\begin{figure*}[h]
    \centering
    \begin{tabular}{ccc}
        (before) $\val$ & & (after) $\val'$\\
        \begin{tikzpicture}[earrow, baseline = (x.base)]
            \node[mynode, fill = black, inner sep = .2em] (x) at (0, 0) {};
            \node[mynode, fill = black, inner sep = .2em] (E) at (1.2, 0) {};
            \node[mynode, fill = black, inner sep = .2em] (W) at (-1.2, 0) {};
            \node[mynode, fill = black, inner sep = .2em] (N) at (0, 1.2) {};
            \node[mynode, fill = black, inner sep = .2em] (S) at (0, -1.2) {};
            \draw[->] (x) -- node[font=\scriptsize, fill = white, inner sep = 1pt] {$\aE$} (E);
            \draw[->] (x) -- node[font=\scriptsize, fill = white, inner sep = 1pt] {$\aW$} (W);
            \draw[->] (x) -- node[font=\scriptsize, fill = white, inner sep = 1pt] {$\aN$} (N);
            \draw[->] (x) -- node[font=\scriptsize, fill = white, inner sep = 1pt] {$\aS$} (S);
        \end{tikzpicture}
        &
        \hspace{1em}{\Large $\leadsto$}\hspace{1em}
        &
        \begin{tikzpicture}[earrow, baseline = (x.base)]
            \node[mynode, fill = black, inner sep = .2em] (x) at (0, 0) {};
            \node[mynode, fill = black, inner sep = 1pt] (a) at (0.6, 0.6) {};
            \node[mynode, fill = black, inner sep = 1pt] (b) at (-0.6, -0.6) {};
            \node[mynode, fill = black, inner sep = .2em] (E) at (1.2, 0) {};
            \node[mynode, fill = black, inner sep = .2em] (W) at (-1.2, 0) {};
            \node[mynode, fill = black, inner sep = .2em] (N) at (0, 1.2) {};
            \node[mynode, fill = black, inner sep = .2em] (S) at (0, -1.2) {};
            \draw[->] (x) -- node[font=\scriptsize, fill = white, inner sep = 1pt] {$\aA$} (a);
            \draw[->] (a) -- node[font=\scriptsize, fill = white, inner sep = 1pt] {$\aA$} (E);
            \draw[->] (a) -- node[font=\scriptsize, fill = white, inner sep = 1pt] {$\aB$} (N);
            \draw[->] (x) -- node[font=\scriptsize, fill = white, inner sep = 1pt] {$\aB$} (b);
            \draw[->] (b) -- node[font=\scriptsize, fill = white, inner sep = 1pt] {$\aA$} (S);
            \draw[->] (b) -- node[font=\scriptsize, fill = white, inner sep = 1pt] {$\aB$} (W);
        \end{tikzpicture}
    \end{tabular}
    \caption{Illustration of $\val'$ in the proof of \Cref{lem:reducing-actions} (the small vertices are new).}
    \label{figure:reducing-actions}
\end{figure*}
\begin{lem}\label{lem:reducing-actions}
    Let $\algclass \in \set{\DREL_{\ttest}, \DREL^{\fin}_{\ttest}}$.
    For all "primitive $\PWPlop$ terms" $\term$ over the "actions" $\set{\aE, \aN, \aW, \aS}$, we have:
    \[\algclass \models \term \le \emp \quad\Longleftrightarrow\quad \algclass \models \sigmaAB(\term) \le \emp.\]
\end{lem}
\begin{proof}
    \proofcase{($\Leftarrow$)}
    Suppose that $\hat{\val}(\term) \neq \emptyset$ for some $\val \colon \vsig \to \wp(X^2)$ in $\algclass$.
    Let $\val' \colon \vsig \to \wp((X \uplus (X \times \set{\aA, \aB}))^2)$ be the "valuation" given by
    \begin{align*}
        \val'(\aA) &\defeq \set{\tuple{x, \tuple{x, \aA}} \mid x \in X} \uplus \set{\tuple{\tuple{x, \aA}, y} \mid \tuple{x, y} \in \val(\aE)} \uplus \set{\tuple{\tuple{x, \aB}, y} \mid \tuple{x, y} \in \val(\aS)},\\
        \val'(\aB) &\defeq \set{\tuple{x, \tuple{x, \aB}} \mid x \in X} \uplus \set{\tuple{\tuple{x, \aA}, y} \mid \tuple{x, y} \in \val(\aN)} \uplus \set{\tuple{\tuple{x, \aB}, y} \mid \tuple{x, y} \in \val(\aW)},
    \end{align*}
    the "primitive tests" are interpreted as in $\val$ on $X$ and so that $\ttest$ holds on $X \times \set{\aA, \aB}$, and $\val'(z) \defeq \val(z)$ for the other "variables" $z$.
    Then, $\val' \in \algclass$, $\val(a) \subseteq \hat{\val}'(\sigmaAB(a))$ for $a \in \set{\aE, \aN, \aW, \aS}$, and $\val(p) \subseteq \val'(p)$ for $p \in \psig$.
    As $\term$ is a "$\PWPlop$ term", $\hat{\val}(\term)$ is monotone with respect to $\val$; hence, $\emptyset \neq \hat{\val}(\term) \subseteq \hat{\val}'(\sigmaAB(\term))$.

    \proofcase{($\Rightarrow$)}
    Suppose that $\hat{\val}(\sigmaAB(\term)) \neq \emptyset$ for some $\val \in \algclass$.
    Then, the "valuation" $\val_{\sigmaAB}$ given by $\val_{\sigmaAB}(a) \defeq \hat{\val}(\sigmaAB(a))$ for $a \in \set{\aE, \aN, \aW, \aS}$ (and $\val_{\sigmaAB}(z) \defeq \val(z)$ for the other "variables" $z$)
    is in $\algclass$ and satisfies $\hat{\val}_{\sigmaAB}(\term) = \hat{\val}(\sigmaAB(\term)) \neq \emptyset$.
\end{proof}
\begin{cor}\label{cor:reducing-primitive-tests-conv}
    The "emptiness problem" on $\DREL_{\ttest}$ (resp.\ $\DREL_{\ttest}^{\fin}$) is 
    $\PiOne$-complete for "$\PWPlop$ terms" where the number of "actions" is $2$ and the number of "primitive tests" is $1$.
\end{cor}
\begin{proof}
    \proofcase{(In $\PiOne$)} By \Cref{prop:upper-bound}, as in \Cref{thm:undecidable-while-emptiness}.

    \proofcase{(Hardness)} By \Cref{thm:reducing-primitive-tests,lem:reducing-actions}.
\end{proof}

\begin{rem}\label{rem:reducing-primitive-tests}
    In \Cref{cor:undecidable-PDL} and in the proofs of \Cref{cor:undecidable-PCoR*-func,cor:undecidable-PCoR*-variable-emptiness,cor:undecidable-PCoR*-difference-emptiness},
    we can replace the "term" $\TrPWP_{\term[3]_{\mathcal{D}}}(\term[1]_{\tgr}')$ with $\sigmaAB(\TrOne(\TrPWP_{\term[3]_{\mathcal{D}}}(\term[1]_{\tgr}')))$,
    which is also a "primitive $\PWPlop$ term" whose emptiness on $\DREL_{\ttest}$ is equivalent (\Cref{lem:reducing-primitive-tests,lem:reducing-actions}).
    Hence, each of them is $\PiOne$-complete even if the number of "actions" is $2$ and the number of "primitive tests" is $1$ (in the "$\PWPlop$ term" before eliminating $\ttest$).
\end{rem}

\subsection{On Decidable Fragments}
When the number of "actions" and the number of "primitive tests" are more restricted than in \Cref{cor:reducing-primitive-tests-conv},
is the "emptiness problem" decidable?
Below, we give some decidable fragments.
\begin{prop}\label{prop:test-free-decidable}
    The "equational theory" on $\DREL_{\ttest}$ is decidable for "$\PWPlop$ terms" where the number of "primitive tests" is $0$.
\end{prop}
\begin{proof}
    In this case, every "test" $b$ satisfies $\REL \models b = \eps$ or $\REL \models b = \emp$; so we can eliminate the operator $\bl^{\starb{b}}$, as $\REL \models \term^{\starb{\emp}} = \eps$ and $\REL \models \term^{\starb{\eps}} = \emp$,
    and thus we can eliminate $\bl^{*}$ from the given "term".
    Hence, by \Cref{prop:upper-bound}~(\ref{prop:upper-bound-2}), this completes the proof.
\end{proof}
\begin{prop}\label{prop:one-action-decidable}
    The "equational theory" on $\DREL_{\ttest}$ is decidable for "$\PWPlop$ terms" where the number of "actions" is at most $1$.
\end{prop}
\begin{proof}[Proof Sketch]
    Let $\term \le \term[2]$ be a given "inequation".
    In this case, we may assume that $\vsig$ consists of the "variables" occurring in $\term \le \term[2]$ and $\tilde{p}$ for the "primitive tests" $p$ occurring in $\term \le \term[2]$ (so $\vsig$ is finite);
    then, the Gaifman "graph" (the "graph" in which edge labels and edge directions have been forgotten) of each such "valuation" in $\DREL_{\ttest}$ is a ""pseudoforest"" (a graph in which each connected component has at most one cycle),
    because the single "action" is a functional relation and each "primitive test" is a subset of the identity relation.
    Each "pseudoforest" has ""treewidth"" at most $2$ (as cycles have "treewidth" at most $2$ and trees have "treewidth" at most $1$),
    and thus we have $\DREL_{\ttest} \models \term \le \term[2]$ (iff $\REL_{\ttest, \tfunc} \models \term \le \term[2]$) iff $\REL_{\ttest, \tfunc}^{\fin, \mathrm{tw} \le 2} \models \term \le \term[2]$ (by the finite model property, \Cref{prop:bounded-model-property}).
    Here, we write $\REL^{\mathrm{tw} \le 2} \subseteq \REL$ for the class of "valuations" whose Gaifman "graph" has "treewidth" at most $2$.
    Moreover, we can easily encode the "quantifier-free formula" $(\ttest \land \tfunc) \to \term \le \term[2]$ into a monadic second-order logic (MSO) formula, 
    as an analog of the standard translation from "quantifier-free formulas" of the calculus of relations to first-order logic sentences (the Schr{\"o}der--Tarski translation) \cite{tarskiCalculusRelations1941}.
    Hence, the "equational theory" is decidable,
    because the validity problem of MSO formulas in finite structures of bounded treewidth is decidable \cite{courcelleMonadicSecondorderLogic1990}.
\end{proof}
Hence, for "$\PWPlop$ terms", we have the following dichotomy result.
\begin{cor}\label{cor:reducing-primitive-tests-2}
    For "$\PWPlop$ terms" where the number of "actions" is (at most) $k$ and the number of "primitive tests" is (at most) $l$,
    the "emptiness problem" on $\DREL_{\ttest}$ is 
    decidable if $k \le 1$ or $l = 0$, and
    $\PiOne$-complete otherwise (i.e., $k \ge 2$ and $l \ge 1$).
\end{cor}
\begin{proof}
    By \Cref{cor:reducing-primitive-tests-conv} (for $k \ge 2$ and $l \ge 1$),
    \Cref{prop:test-free-decidable} (for $l = 0$), and
    \Cref{prop:one-action-decidable} (for $k \le 1$).
\end{proof}
The decidability results above (\Cref{prop:test-free-decidable,prop:one-action-decidable}) also hold when "actions" may additionally occur in the form $a^{\smile}$ (atomic converse), by the same proofs; hence, \Cref{cor:reducing-primitive-tests-2} also holds for such "terms".

\subsection{On the coNP-Hardness for the While-Free Fragment (Where the Number of Primitive Tests Is Fixed)}\label{section:NP-hardness}
Additionally,
we show that the "emptiness problem" is $\coNP$-complete
when the while operator $\bl^{\starb{b}}$ does not occur and the number of "primitive tests" is fixed (to one).
\begin{cor}\label{cor:NP-complete-while-emptiness}
    The "emptiness problem" on $\DREL_{\ttest}$ is $\coNP$-complete
    for "$\PWPlop$ terms" without $\bl^{\starb{b}}$ where the number of "actions" is $2$ and the number of "primitive tests" is $1$.
\end{cor}
\begin{proof}
    \proofcase{(In $\coNP$)}
    By \Cref{prop:upper-bound}~(\ref{prop:upper-bound-2}).

    \proofcase{(Hardness)}
    We reduce from the complement of the "bounded periodic domino problem" (\Cref{prop:bounded-periodic-tiling}).
    Let $\mathcal{D}$ be a "domino system" and let $\ell \ge 1$ be the given (unary-encoded) integer.
    We define the $\ell$-th iteration operator $\term^{(\ell)_{b}}$ as follows:
    \begin{align*}
        \term^{(\ell)_{b}} &\quad\defeq\quad \term (b \term)^{\ell-1} \tcompl{b}.
    \end{align*}
    Let $\term[1]_{\tgr}'' \defeq (p ((q \aN^{(\ell)_{\tilde{q}}})^{\lop} \aE)^{(\ell)_{\tilde{p}}})^{\lop}$ (i.e., the "term" $\term[1]_{\tgr}'$ (used in \Cref{thm:undecidable-while-emptiness}) in which $\bl^{\plusb{\tilde{q}}}$ and $\bl^{\plusb{\tilde{p}}}$ have been replaced with $\bl^{(\ell)_{\tilde{q}}}$ and $\bl^{(\ell)_{\tilde{p}}}$).
    As an analog of \Cref{lem:connectivity'}, we have:
    \[\mbox{$\mathcal{D}$ has an $\tuple{\ell,\ell}$-"periodic tiling"}
    \;\Leftrightarrow\; \DREL^{\fin}_{\ttest, (\aE \union \aN \union \aW \union \aS)^* \ge \top, \term[3]_{\mathcal{D}} \ge \eps} \not\models \term[1]_{\tgr}'' \le \emp.\]
    (For ($\Leftarrow$), when $\tuple{o, o} \in \hat{\val}(\term[1]_{\tgr}'')$, we use $f(x, y) \defeq \val(\aE)^{x}(\val(\aN)^{y}(o))$ (as in the proof of \Cref{prop:grid}) in the proof of \Cref{lem:connectivity'}, which is $\tuple{\ell, \ell}$-periodic by the form of $\term[1]_{\tgr}''$.)
    Thus, in the same way as \Cref{thm:undecidable-while-emptiness,thm:reducing-primitive-tests}, we have:
    \[\mbox{$\mathcal{D}$ does not have an $\tuple{\ell,\ell}$-"periodic tiling"}                                                                
    \;\Leftrightarrow\; \DREL_{\ttest} \models
    \sigmaAB(\TrOne(\TrPWP_{\term[3]_{\mathcal{D}}}(\term[1]_{\tgr}''))) \le \emp,\]
    where we also use \Cref{lem:reducing-actions}, and $\bl^{\starb{b}}$ does not occur in $\sigmaAB(\TrOne(\TrPWP_{\term[3]_{\mathcal{D}}}(\term[1]_{\tgr}'')))$.
    This "term" is computable in polynomial time, as $\ell$ is unary-encoded.
    Hence, this completes the proof.
\end{proof}

\begin{rem}\label{rem:NP-complete}
    \Cref{cor:NP-complete-while-emptiness} is trivial when the number of "primitive tests" is not bounded.
    Because $\DREL_{\ttest} \models b \le \emp$ iff the propositional logic formula obtained from $b$ by replacing each $p$, $\tilde{p}$, $\eps$, $\emp$, $\union$, $\compo$ with $p$, $\lnot p$, true, false, $\lor$, $\land$ is unsatisfiable, there is a reduction from the unsatisfiability problem of propositional logic formulas, which is well known to be $\coNP$-complete \cite{cookComplexityTheoremprovingProcedures1971}.
\end{rem}
\section{Undecidability of PCoR* with Difference}\label{section:undecidable-PCoR*-with-difference}
In this section, we show that the "equational theory" of $\PCoR_{\set{\bl^{*}, \com{\eps}}}$ (on $\REL$) is $\PiOne$-complete (\Cref{cor:undecidable-PCoR*-difference}).
We prove it as a corollary of \Cref{thm:undecidable-while-emptiness} via "hypothesis eliminations".
\begin{cor}\label{cor:undecidable-PCoR*-func}
    The "equational theory" of $\PCoR_{\set{\bl^{*}}}$ on $\DREL$ is $\PiOne$-complete.
    More precisely, the "inclusion problem" $\term[1] \le \term[2]$ on $\DREL$, where $\term[1]$ is a $\set{\compo, \union, \bl^{*}, \bl^{\lop}}$-"term" and $\term[2]$ is a $\set{\compo, \union, \top}$-"term", is $\PiOne$-complete.
\end{cor}
\begin{proof}
    \proofcase{(In $\PiOne$)}
    By \Cref{prop:upper-bound}, as $\DREL$ is "submodel-closed".

    \proofcase{(Hardness)}
    Recall that, in the proof of \Cref{thm:undecidable-while-emptiness}, the last line is $\DREL_{\ttest} \models \TrPWP_{\term[3]_{\mathcal{D}}}(\term[1]_{\tgr}') \le \emp$,
    where $\term[1]_{\tgr}'$ and $\term[3]_{\mathcal{D}}$ are the "terms" given in \Cref{section:while} for a "domino system" $\mathcal{D}$.
    By eliminating $\ttest$ using \Cref{ex:tests}, the last line in \Cref{thm:undecidable-while-emptiness} is equivalent to the following, where $p_0, \dots, p_{n+1}$ denote the "primitive tests" $c_1, \dots, c_n, p, q$, respectively:
    \begin{align*}
        \DREL \models{} & \Tr_{\bigcompo_{i = 0}^{n+1} (p_i \union \tilde{p}_i)}(\TrPWP_{\term[3]_{\mathcal{D}}}(\term[1]_{\tgr}')) [p_0^{\lop}, \dots, p_{n+1}^{\lop}, \tilde{p}_0^{\lop}, \dots, \tilde{p}_{n+1}^{\lop}/p_0, \dots, p_{n+1}, \tilde{p}_0, \dots, \tilde{p}_{n+1}]\\
        & \le \top (\sum_{i = 0}^{n+1} p_i \tilde{p}_i) \top.
    \end{align*}
    Here, the left-hand side can be regarded as a $\set{\compo, \union, \bl^{*}, \bl^{\lop}}$-"term", as the occurrences of $\eps$ and $\emp$ can be removed by easy simplifications.
\end{proof}
Moreover, by replacing $\tfunc$ with other "hypotheses", we also have the following.
\begin{cor}\label{cor:undecidable-PCoR*-difference}
    The "equational theory" of $\PCoR_{\set{\bl^{*}, \com{\eps}}}$ on $\REL$ is $\PiOne$-complete.
    More precisely, the "inclusion problem" $\term[1] \le \term[2]$ on $\REL$, where $\term[1]$ is a $\set{\compo, \union, \bl^{*}, \bl^{\lop}}$-"term" and $\term[2]$ is a $\set{\compo, \union, \bl^{\lop}, \com{\eps}}$-"term", is $\PiOne$-complete.
\end{cor}
\begin{proof}
    \proofcase{(In $\PiOne$)} By \Cref{prop:upper-bound}.

    \proofcase{(Hardness)}
    Let $\term[1] \le \term[2]$ be the "inequation" in the proof of \Cref{cor:undecidable-PCoR*-func} and let $\tfunc' \defeq \set{\aE \aW \le \eps, \aW \aE \le \eps, \aN \aS \le \eps, \aS \aN \le \eps}$.
    By \Cref{ex:surjective}, $\DREL \models \term[1] \le \term[2]$ iff $\REL_{\tfunc'} \models \term[1] \le \term[2]$,
    as, in both classes, $\aE, \aN, \aW, \aS$ are mutually inverse bijections on the image of each "graph homomorphism" from a "graph" of $\term[1]$
    (for each "vertex" $g$ in the image of the grid part, the "loops" of $\term[3]_{\mathcal{D}}$ hold at $g$;
    hence, e.g., $g$ has an $\aE$-successor $z$ with $\tuple{z, g} \in \val(\aW)$ by $(\aE (\sum H_c) \aW)^{\lop}$ for some color $c$, and $z$ is unique by $\tfunc$ (in $\DREL$) or by $\aW \aE \le \eps$ (in $\REL_{\tfunc'}$), and similarly for $\aN$, $\aW$, $\aS$;
    thus, the image of the grid part is closed under $\aE$, $\aN$, $\aW$, and $\aS$, in the same way as in the proof of \Cref{lem:grid-surjective}, and hence it is the whole image),
    and each counterexample in $\Sur_{\glang(\term[1])}(\REL_{\tfunc'})$ remains a counterexample in $\DREL$ after replacing $\val(p)$ with $\val(p) \cap \hat{\val}(\eps)$ for each "primitive test" $p$
    (and $\val(z)$ with $\emptyset$ for the "variables" $z$ not occurring in $\term[1] \le \term[2]$).
    By eliminating $\tfunc'$ (\Cref{prop:Hoare-hypotheses}, with $\REL \models \aE \aW \le \eps \leftrightarrow (\aE \aW \com{\eps})^{\lop} \le \emp$ (\Cref{ex:u-le-I}), and so on), we obtain:
    \[\DREL \models \term[1] \le \term[2] \quad\Longleftrightarrow\quad \REL \models \term[1] \le \term[2] \union \top ((\aE \aW \com{\eps})^{\lop} \union (\aW \aE \com{\eps})^{\lop} \union (\aN \aS \com{\eps})^{\lop}\union (\aS \aN \com{\eps})^{\lop}) \top,\]
    where $\top$ can be replaced with $\eps \union \com{\eps}$ (and then the occurrences of $\eps$ can be removed by distributivity).
\end{proof}
Additionally, for the star-free fragment on $\DREL$, we have the following (cf.\ on $\REL$, the "equational theory" of $\PCoR$ is already $\coNP$-complete \cite[Theorem 24]{nakamuraExistentialCalculiRelations2023}).
\begin{cor}\label{cor:coNP-PCoR-func}
    The "equational theory" of $\PCoR$ on $\DREL$ is $\coNP$-complete.
    More precisely, the "inclusion problem" $\term[1] \le \term[2]$ on $\DREL$, where $\term[1]$ is a $\set{\compo, \union, \bl^{\lop}}$-"term" and $\term[2]$ is a $\set{\compo, \union, \top}$-"term", is $\coNP$-complete.
\end{cor}
\begin{proof}
    \proofcase{(In $\coNP$)} By \Cref{prop:upper-bound}~(\ref{prop:upper-bound-2}), as $\DREL$ is "submodel-closed".

    \proofcase{(Hardness)} In the same way as \Cref{cor:undecidable-PCoR*-func}, by using \Cref{cor:NP-complete-while-emptiness} instead of \Cref{thm:undecidable-while-emptiness}.
\end{proof}

\subsection*{Remark: Comparison to the Reduction from the Universality Problem of CFGs}
The "equational theory" is $\PiOne$-complete for $\PCoR_{\set{\bl^{*}, \com{x}}}$ \cite[Theorem 50]{nakamuraExistentialCalculiRelations2023}.
This result can be shown as a corollary of \Cref{cor:undecidable-PCoR*-difference}.
We can give a (polynomial-time) reduction from the "equational theory" of $\PCoR_{\set{\bl^{*}, \com{\eps}}}$ to "that@equational theory" of $\PCoR_{\set{\bl^{*}, \com{x}}}$, by encoding an equivalence relation using "hypotheses" and eliminating them (see \Cref{section:difference-encoding} for details, where $\term^{\lop}$ is regarded as $\term \cap \eps$).
Hence, in certain settings, extending with the "difference constant" does not make the "equational theory" strictly harder than extending with "variable complements".

The reduction of \cite[Theorem 50]{nakamuraExistentialCalculiRelations2023} is a reduction from the universality problem of context-free grammars (CFGs).
In a nutshell, each rewriting rule $x \leftarrow w$ (where $x$ is a "variable" and $w$ is a $\set{\eps, \compo}$-"term") of a given CFG is encoded as the "hypothesis" $w \le x$.
By $\REL \models w \le x \leftrightarrow w \cap \com{x} \le \emp$ (\Cref{ex:u-le-x}), we can eliminate these "Hoare hypotheses".
Using this approach, we can show the following undecidability.
\begin{prop}[Corollary of {\cite[Theorem 50]{nakamuraExistentialCalculiRelations2023}}, cf.\ \Cref{cor:undecidable-PCoR*-difference}]\label{prop:undecidable-PCoR*-variable}
    The "inclusion problem" $\term[1] \le \term[2]$ on $\REL$, where $\term[1]$ is a $\set{\compo, \union, \bl^{*}}$-"term" and $\term[2]$ is a $\set{\compo, \union, \bl^{\lop}, \com{x}}$-"term", is $\PiOne$-complete.
\end{prop}
\begin{proof}
    \proofcase{(In $\PiOne$)}
    By \Cref{prop:upper-bound}.

    \proofcase{(Hardness)}
    We recall the reduction of \cite[Theorem 50]{nakamuraExistentialCalculiRelations2023}.
    Let $C$ be a CFG with rules $x_i \leftarrow w_i$ ($i \in \range{1, n}$) and start symbol $x_0$, where the nonterminals and the terminals (in a finite set $A$) are "variables".
    Then, $C$ generates all words over $A$ iff $\REL_{\set{w_i \le x_i \mid i \in \range{1, n}}} \models A^{*} \le x_0$ \cite[Lemma 48]{nakamuraExistentialCalculiRelations2023}, where $A^{*}$ denotes the "term" $(\sum A)^{*}$.
    By applying \Cref{ex:u-le-x-2} (instead of \Cref{ex:u-le-x}) to each $w_i \le x_i$ with pairwise distinct "fresh" "variables" $z_1, \dots, z_n$, this is equivalent to:
    \[\REL \models A^{*} \le x_0 \union \top \Big(\sum_{i = 1}^{n} ((w_i z_i)^{\lop} \union (\com{x}_i\, \com{z}_i)^{\lop} \union (\com{z}_i\, \com{x}_i)^{\lop} \union (z_i x_i)^{\lop} \union (x_i z_i)^{\lop})\Big) \top,\]
    where $\top$ can be replaced with $a \union \com{a}$ for some $a \in A$.
\end{proof}
Below we give an open question arising from \Cref{cor:undecidable-PCoR*-difference} and \Cref{prop:undecidable-PCoR*-variable}.
\begin{ques}[cf.\ \Cref{cor:undecidable-PCoR*-difference,prop:undecidable-PCoR*-variable}]
    Is the "inclusion problem" $\term[1] \le \term[2]$ on $\REL$ decidable, where $\term[1]$ is a $\set{\compo, \union, \bl^{*}}$-"term" and $\term[2]$ is a $\set{\compo, \union, \bl^{\lop}, \com{\eps}}$-"term"?
\end{ques}
\section{On the Emptiness Problem on REL and DREL}\label{section:emptiness}
In this section, we focus on the "emptiness problem" of fragments of $\PCoR_{\set{\bl^{*}, \com{\eps}, \com{x}}}$
on $\DREL$ and $\REL$.
(When specifying the operators of "terms" precisely, we ignore the occurrences of $\eps$ and $\emp$ in the "terms" constructed below, as they can be removed by easy simplifications.)

\subsection{Undecidable Fragments on DREL}\label{section:emptiness-undecidability-DREL}
Recall that, in the proof of \Cref{thm:undecidable-while-emptiness}, the last line is $\DREL_{\ttest} \models \TrPWP_{\term[3]_{\mathcal{D}}}(\term[1]_{\tgr}') \le \emp$,
where $\term[1]_{\tgr}'$ and $\term[3]_{\mathcal{D}}$ are the "terms" given in \Cref{section:while} for a "domino system" $\mathcal{D}$.
By \Cref{thm:undecidable-while-emptiness},
we can show the following undecidability result.
\begin{cor}\label{cor:undecidable-PCoR*-variable-emptiness}
    The "emptiness problem" of $\PCoR_{\set{\bl^{*}, \com{x}}}$ (precisely, $\set{\compo, \union, \bl^{*}, \bl^{\lop}, \com{x}}$-"terms") on $\DREL$ is $\PiOne$-complete.
\end{cor}
\begin{proof}
    \proofcase{(In $\PiOne$)}
    By \Cref{prop:upper-bound}, as $\DREL$ is "submodel-closed".

    \proofcase{(Hardness)}
    Let $p_0, \dots, p_{k-1}, \tilde{p}_0, \dots, \tilde{p}_{k-1}$ be the "primitive tests" occurring in the last line of \Cref{thm:undecidable-while-emptiness}.
    By eliminating $\ttest$ (\Cref{ex:substitution-tests}), the last line in \Cref{thm:undecidable-while-emptiness} is equivalent to the following:
    \[\DREL \models \TrPWP_{\term[3]_{\mathcal{D}}}(\term[1]_{\tgr}')[q_0^{\lop}, \dots, q_{k-1}^{\lop}, \com{q}_0^{\lop}, \dots, \com{q}_{k-1}^{\lop}/p_0, \dots, p_{k-1}, \tilde{p}_0, \dots, \tilde{p}_{k-1}] \le \emp,\]
    where $q_0, \dots, q_{k-1}$ are "fresh" "variables".
    Hence, this completes the proof.
\end{proof}

Furthermore, by eliminating $\ttest$ using the "difference constant", we also have the following undecidability result.
\begin{cor}\label{cor:undecidable-PCoR*-difference-emptiness}
    The "emptiness problem" of $\PCoR_{\set{\bl^{*}, \com{\eps}}}$ (precisely, $\set{\compo, \union, \bl^{*}, \bl^{\lop}, \com{\eps}}$-"terms") on $\DREL$ is $\PiOne$-complete.
\end{cor}
\begin{proof}
    \proofcase{(In $\PiOne$)}
    By \Cref{prop:upper-bound}, as $\DREL$ is "submodel-closed".

    \proofcase{(Hardness)}
    As "primitive tests" occur only positively in "$\PWPlop$ terms",
    we can eliminate $\ttest$ for the "emptiness problem" on $\DREL$ by substituting $q_i^{\lop}$ and $(q_i \com{\eps})^{\lop}$ for $p_i$ and $\tilde{p}_i$, respectively
    (i.e., $(q_i \com{\eps})^{\lop}$ is used instead of $\com{q}_i^{\lop}$ in the proof of \Cref{cor:undecidable-PCoR*-variable-emptiness});
    this transformation correctly works, as $\val(q_i)$ can be chosen to map each $x$ to $x$ if $p_i$ holds at $x$ and to another element otherwise, where we add a fresh element if necessary;
    conversely, for each $\val \in \DREL$, the relations $\hat{\val}(q_i^{\lop})$ and $\hat{\val}((q_i \com{\eps})^{\lop})$ are disjoint, as $\val(q_i)$ is a partial function,
    so we can interpret $p_i$ and $\tilde{p}_i$ as $\hat{\val}(q_i^{\lop})$ and $\hat{\val}(\eps) \setminus \hat{\val}(q_i^{\lop})$, respectively.
    Thus, the last line in \Cref{thm:undecidable-while-emptiness} is equivalent to the following:
    \[\DREL \models \TrPWP_{\term[3]_{\mathcal{D}}}(\term[1]_{\tgr}')[q_0^{\lop}, \dots, q_{k-1}^{\lop}, (q_0 \com{\eps})^{\lop}, \dots, (q_{k-1} \com{\eps})^{\lop}/p_0, \dots, p_{k-1}, \tilde{p}_0, \dots, \tilde{p}_{k-1}] \le \emp.\]
    Hence, this completes the proof.
\end{proof}

\subsection{Decidable Fragments on DREL}\label{section:emptiness-decidability-DREL}
Below we give two decidable fragments on $\DREL$.
The following shows that the "emptiness problem" is decidable for purely \emph{positive} fragments.
\begin{prop}\label{prop:decidable-PCoR*-DREL-emptiness}
    The "emptiness problem" on $\DREL$ (resp.\ $\REL$) is $\NC^{1}$-complete for $\PCoR_{\set{\bl^{*}}}$.
\end{prop}
\begin{proof}
    \proofcase{(In $\NC^{1}$)}
    Let $\algclass = \REL$ or $\DREL$.
    Let $\val_1 \colon \vsig \to \wp(\set{o}^2)$ be the "valuation" defined as $\val_1(a) \defeq \set{\tuple{o, o}}$ for $a \in \vsig$ (i.e., the "graph" of $\val_1$ is the singleton "graph" \begin{tikzpicture}[baseline = -.5ex]
        \graph[grow up= 3.cm, branch right = 3.cm, nodes={font=\tiny}]{
        {0/{$o$}[mynode, inner sep = .1em]}
        };
        \graph[use existing nodes, edges={color=black, pos = .5, earrow}, edge quotes={fill = white, inner sep=1pt,font= \scriptsize}]{
        0 ->["$\vsig$", out = 150, in = 210, looseness = 10] 0; 
        };
    \end{tikzpicture}).
    Clearly, $\val_1 \in \DREL$.
    Because every $\vsig$-labelled "graph" is "graph homomorphic" to $\gval(\val_1, o, o)$, we have that $\algclass \models \term \le \emp$ iff $\glang(\term) = \emptyset$,
    as a corollary of \Cref{prop:glang}.
    The right-hand side can be inductively calculated by the following properties:
    \begin{gather*}
        \glang(a) \neq \emptyset,\quad
        \glang(\eps) \neq \emptyset,\quad
        \glang(\top) \neq \emptyset,\quad
        \glang(\emp) = \emptyset,\quad
        \glang(\term[1]^*) \neq \emptyset,\\
        \glang(\term[1] \compo \term[2]) \neq \emptyset \iff \glang(\term[1]) \neq \emptyset \mbox{ and } \glang(\term[2]) \neq \emptyset,\quad
        \glang(\term[1] \union \term[2]) \neq \emptyset \iff \glang(\term[1]) \neq \emptyset \mbox{ or } \glang(\term[2]) \neq \emptyset, \\
        \glang(\term[1]^{\smile}) \neq \emptyset \iff \glang(\term[1]) \neq \emptyset,\quad
        \glang(\term[1] \cap \term[2]) \neq \emptyset \iff \glang(\term[1]) \neq \emptyset \mbox{ and } \glang(\term[2]) \neq \emptyset.
    \end{gather*}
    Thus, we can reduce this problem to the evaluation of a fixed finite algebra (with $2$ values of ``$= \emptyset$'' and ``$\neq \emptyset$''),
    which is in $\NC^{1}$ \cite[Corollary 1]{lohreyParallelComplexityTree2001}, \cite[Theorem 8]{bussBooleanFormulaValue1987}.

    \proofcase{(Hardness)}
    This is because the "emptiness problem" of $\set{\emp, \eps, \union, \compo}$-"terms" is equal to the complement of the Boolean formula value problem \cite{bussBooleanFormulaValue1987}, which is $\NC^{1}$-hard.
\end{proof}
The following shows that the "emptiness problem" is decidable for the intersection-free fragments.
\begin{prop}\label{prop:decidable-KA-emptiness}
    For $\set{\eps, \emp, \compo, \union, \bl^{*}, \bl^{\smile}, \com{\eps}, \com{x}}$-"terms", the "emptiness problem" on $\DREL$ (resp.\ $\REL$) is $\NC^{1}$-complete.
\end{prop}
\begin{proof}
    \proofcase{(In $\NC^{1}$)}
    Let $\algclass = \REL$ or $\DREL$.
    Let $\val_2 \colon \vsig \to \wp(\set{o,o'}^2)$ be the "valuation" defined as $\val_2(a) \defeq \set{\tuple{o, o}, \tuple{o',o'}}$ for $a \in \vsig$, where $o \neq o'$
    (i.e., the "graph" of $\val_2$ is the "graph" \begin{tikzpicture}[baseline = -.5ex]
        \graph[grow up= 3.cm, branch right = .8cm, nodes={font=\tiny}]{
        {0/{$o$}[mynode, inner sep = .1em, minimum width = 1.2em], 1/{$o'$}[mynode, inner sep = .1em, minimum width = 1.2em]}
        };
        \graph[use existing nodes, edges={color=black, pos = .5, earrow}, edge quotes={fill = white, inner sep=1pt,font= \scriptsize}]{
        0 ->["$\vsig$", out = 150, in = 210, looseness = 8] 0; 
        1 ->["$\vsig$", out = 30, in = -30, looseness = 8] 1; 
        };
    \end{tikzpicture}).
    Clearly, $\val_2 \in \DREL$.
    Then, for all $\set{\eps, \emp, \compo, \union, \bl^{*}, \bl^{\smile}, \com{\eps}, \com{x}}$-"terms" $\term$,
    every "graph" in $\glang(\term)$ is "graph homomorphic" to $\gval(\val_2, o, o)$ or $\gval(\val_2, o, o')$,
    by the form of the "graph".
    Thus, we have that $\algclass \models \term \le \emp$ iff $\glang(\term) = \emptyset$, as a corollary of \Cref{prop:glang}.
    Hence, in the same way as \Cref{prop:decidable-PCoR*-DREL-emptiness}, this completes the proof.

    \proofcase{(Hardness)}
    By the same proof as \Cref{prop:decidable-PCoR*-DREL-emptiness}.
\end{proof}

\subsection{Decidability on REL}\label{section:emptiness-decidability-REL}
Below, we show that the "emptiness problem" for $\PCoR_{\set{\bl^{*}, \com{\eps}, \com{x}}}$ is decidable on $\REL$.
\begin{thm}\label{thm:decidable-ECoR*-emptiness}
    The "emptiness problem" on $\REL$ is $\PSPACE$-complete for $\PCoR_{\set{\bl^{*}, \com{\eps}, \com{x}}}$ and $\PCoR_{\set{\bl^{*}, \com{x}}}$.
\end{thm}
\begin{proof}
    \proofcase{(In $\PSPACE$)}
    We recall the characterization using \emph{edge saturations} given in \cite[Theorem 18]{nakamuraExistentialCalculiRelations2023}.
    We say that a "graph" $\graph[2]$ is ""consistent"" if,
    for every $x, y \in \domain{\graph[2]}$,
    \begin{itemize}
        \item for each $a \in \vsig$, either $\tuple{x, y} \not\in a^{\graph[2]}$ or $\tuple{x, y} \not\in \com{a}^{\graph[2]}$ holds,
        \item if $x = y$, then $\tuple{x, y} \not\in \com{\eps}^{\graph[2]}$.
    \end{itemize}
    As a corollary of \cite[Theorem 18]{nakamuraExistentialCalculiRelations2023},
    we have that 
    \[\REL \not\models \term \le \emp \quad\iff\quad
    \mbox{there is some "consistent" "graph" in $\glang(\term)$.}\]
    To determine the right-hand side,
    it suffices to keep only the source and target "vertices" in each inductive step of the construction of $\glang(\term)$.
    For a "graph" $\graph$, let $\intro*\trunc{\graph}$ be the subgraph of $\graph$ induced by the "source" and "target" "vertices",
    with additional $\com{\eps}$-labelled edges between the "source" and the "target" in both directions if merging them in $\graph$ yields a non-"consistent" "graph"
    ($\trunc{\graph}$ is undefined if $\graph$ is not "consistent").
    For instance, if $\graph[1] = \left(\hspace{-.5em}\begin{tikzpicture}[baseline= 0.ex]
        \graph[grow up= 1.cm, branch right = .8cm, nodes={font=\tiny}]{
            {0/{}[mynode, fill = black, inner sep = .1em], 1/{}[mynode, fill = black, inner sep = .1em]};
        };
        \node[left = 4pt of 0](s){};\node[right = 4pt of 1](t){};
        \graph[use existing nodes, edges={color=black, pos = .5, earrow}, edge quotes={fill = white, inner sep=1pt,font= \scriptsize}]{
        (s) -> (0); (1) -> (t);
        (0) ->["$a$"] (1);
        0 ->["$a$", out = 60, in = 120, looseness = 50, loop] 0;
        1 ->["$b$", out = 60, in = 120, looseness = 50, loop] 1;
        };
    \end{tikzpicture}\hspace{-.5em}\right)$ and $\graph[2] = \left(\hspace{-.5em}\begin{tikzpicture}[baseline= 0.ex]
        \graph[grow up= 1.cm, branch right = .8cm, nodes={font=\tiny}]{
            {0/{}[mynode, fill = black, inner sep = .1em], 1/{}[mynode, fill = black, inner sep = .1em]};
        };
        \node[left = 4pt of 0](s){};\node[right = 4pt of 1](t){};
        \graph[use existing nodes, edges={color=black, pos = .5, earrow}, edge quotes={fill = white, inner sep=1pt,font= \scriptsize}]{
        (s) -> (0); (1) -> (t);
        (0) ->["$b$"] (1);
        0 ->["$b$", out = 60, in = 120, looseness = 50, loop] 0;
        1 ->["$\com{a}$", out = 60, in = 120, looseness = 50, loop] 1;
        };
    \end{tikzpicture}\hspace{-.5em}\right)$, then
    $\trunc{\graph[1] \compo \graph[2]} = \left(\hspace{-.5em}\begin{tikzpicture}[baseline= 0.ex]
            \graph[grow up= 1.cm, branch right = 1.3cm, nodes={font=\tiny}]{
                {0/{}[mynode, fill = black, inner sep = .1em], 1/{}[mynode, fill = black, inner sep = .1em]};
            };
            \node[left = 4pt of 0](s){};\node[right = 4pt of 1](t){};
            \graph[use existing nodes, edges={color=black, pos = .5, earrow}, edge quotes={fill = white, inner sep=1pt,font= \scriptsize}]{
            (s) -> (0); (1) -> (t);
            (0) <->["$\com{\eps}$"] (1);
            0 ->["$a$", out = 60, in = 120, looseness = 50, loop] 0;
            1 ->["$\com{a}$", out = 60, in = 120, looseness = 50, loop] 1;
            };
        \end{tikzpicture}\hspace{-.5em}\right)$ and $\trunc{\graph[2] \compo \graph[1]}$ is undefined.
    We then define $\tilde{\glang}(\term)$ in the same way as $\glang(\term)$ (\Cref{section:alternative-semantics}), except that $\trunc{\cdot}$ is applied after each step,
    e.g., $\tilde{\glang}(\term[1] \compo \term[2]) \defeq \set{\trunc{\graph[1] \compo \graph[2]} \mid \graph[1] \in \tilde{\glang}(\term[1]), \graph[2] \in \tilde{\glang}(\term[2]), \mbox{ and $\trunc{\graph[1] \compo \graph[2]}$ is defined}}$.
    As the "vertices" other than the "source" and "target" are never merged afterwards,
    we have $\trunc{\graph[1] \compo \graph[2]} = \trunc{\trunc{\graph[1]} \compo \trunc{\graph[2]}}$ and $\trunc{\graph[1] \cap \graph[2]} = \trunc{\trunc{\graph[1]} \cap \trunc{\graph[2]}}$ (where both sides are undefined simultaneously).
    Similarly, $\trunc{\graph^{\smile}} = \trunc{\trunc{\graph}^{\smile}}$, as the swapping of the "source" and the "target" preserves the induced subgraph and whether merging them yields a non-"consistent" "graph"
    (this is why the $\com{\eps}$-labelled edges are added in both directions).
    Hence, by induction on $\term$, $\tilde{\glang}(\term) = \set{\trunc{\graph} \mid \mbox{$\graph \in \glang(\term)$ is "consistent"}}$;
    in particular, $\glang(\term)$ has some "consistent" "graph" iff $\tilde{\glang}(\term) \neq \emptyset$.
    Each "graph" in $\tilde{\glang}(\term)$ has one or two "vertices" and hence can be stored in space polynomial in the number of "variables" occurring in $\term$.
    For $r^*$, start at the truncated "graph" of $\eps$ and repeatedly compose the current state with a truncated "graph" generated recursively for $r$, truncating after each step.
    A shortest path never repeats a state, so we can bound the search by the number of states using a polynomial-bit counter at each recursion level.
    Thus, we can determine whether $\REL \not\models \term \le \emp$ in nondeterministic polynomial space.
    By Savitch's theorem ($\NPSPACE = \PSPACE$) \cite{savitchRelationshipsNondeterministicDeterministic1970}, this completes the proof.

    \proofcase{(Hardness)}
    From the proof of \cite[Theorem 5.1]{kozenHoareLogicKleene2000}
    (for showing the $\PSPACE$-hardness of propositional Hoare logic), the following problem is $\PSPACE$-complete:
    Given "tests" $b_1, \dots, b_n, b, c_1, \dots, c_n, c$ over the "primitive tests" $p_0, \dots, p_{m-1}, \tilde{p}_0, \dots, \tilde{p}_{m-1}$,
    to decide whether $\REL_{\ttest} \models b_1 \mathtt{x} c_1 \le \emp \to \dots \to b_n \mathtt{x} c_n \le \emp \to b \mathtt{x}^* c \le \emp$,\footnote{Note that $\REL_{\protect\ttest} \models b \mathtt{x}^{\starb{\tcompl{c}}} \le \emp
    \leftrightarrow b \mathtt{x}^* c \le \emp$, cf.\ \cite[Theorem 5.1]{kozenHoareLogicKleene2000}.} where $\mathtt{x}$ is a fixed "action".
    W.l.o.g., we can assume $n \ge 1$.
    Let $\term[3] \defeq \bigcap_{i = 1}^{n} (\tcompl{b}_i \mathtt{x} \union \mathtt{x} \tcompl{c}_i)$.
    We then have:
    \begin{align*}
        &\REL_{\ttest} \models b_1 \mathtt{x} c_1 \le \emp \to \dots \to b_n \mathtt{x} c_n \le \emp \to b \mathtt{x}^* c \le \emp\\
        &\Longleftrightarrow
        \REL_{\ttest} \models \mathtt{x} \le \tcompl{b}_1 \mathtt{x} \union \mathtt{x} \tcompl{c}_1 \to \dots \to \mathtt{x} \le \tcompl{b}_n \mathtt{x} \union \mathtt{x} \tcompl{c}_n \to b \mathtt{x}^* c \le \emp
        \tag{By $\REL_{\protect\ttest} \models b \mathtt{x} c \le \emp \leftrightarrow \mathtt{x} \le \tcompl{b} \mathtt{x} \union \mathtt{x} \tcompl{c}$}\\
        &\Longleftrightarrow
        \REL_{\ttest} \models \mathtt{x} \le \term[3] \to b \mathtt{x}^* c \le \emp \tag{By $\REL \models (\mathtt{x} \le \term[1] \land \mathtt{x} \le \term[2]) \leftrightarrow \mathtt{x} \le \term[1] \cap \term[2]$}\\
        &\Longleftrightarrow
        \REL_{\ttest} \models \mathtt{x} = \term[3] \to b \mathtt{x}^* c \le \emp \tag{By $\REL_{\protect\ttest} \models \mathtt{x} \ge \term[3]$}\\
        &\Longleftrightarrow
        \REL_{\ttest} \models b \term[3]^* c \le \emp \tag{By eliminating the hypothesis (\Cref{prop:REL-axiom-=-substitution})}\\
        &\Longleftrightarrow
        \REL \models (b \term[3]^* c)[q_0^{\lop}, \dots, q_{m-1}^{\lop}, \com{q}_0^{\lop}, \dots, \com{q}_{m-1}^{\lop}/p_0, \dots, p_{m-1}, \tilde{p}_0, \dots, \tilde{p}_{m-1}] \le \emp. \tag{By eliminating $\protect\ttest$ (\Cref{ex:substitution-tests})}
    \end{align*}
    Here, $q_0, \dots, q_{m-1}$ are "fresh" "variables".
    In the fourth step, \Cref{prop:REL-axiom-=-substitution} applies, as $\mathtt{x}$ is "fresh" in $\REL_{\ttest}$ and $\REL_{\ttest} \models \term[3] = \term[3][\term[3]/\mathtt{x}]$
    (for $\le$: if $\tuple{y, z} \in \hat{\val}(\term[3])$, then, for each $i$, $\tcompl{b}_i$ holds at $y$ or $\tcompl{c}_i$ holds at $z$;
    for $\ge$: by $\REL_{\ttest} \models \term[3] \le \mathtt{x}$ and monotonicity).
    As the left-hand side term is a $\PCoR_{\set{\bl^{*}, \com{x}}}$ "term", this completes the proof.
\end{proof}
For the $\bl^{*}$-free fragment,
the "emptiness problem" is $\coNP$-complete, as follows.
\begin{cor}\label{cor:decidable-ECoR*-emptiness}
    The "emptiness problem" on $\REL$ is $\coNP$-complete for $\PCoR_{\set{\com{\eps}, \com{x}}}$.
\end{cor}
\begin{proof}
    \proofcase{(In $\coNP$)}
    By the same algorithm as \Cref{thm:decidable-ECoR*-emptiness}.
    When the Kleene star $\bl^{*}$ does not occur, the nondeterministic algorithm terminates in polynomial time.

    \proofcase{(Hardness)}
    This is because the "emptiness problem" of $\set{\union, \cap, \com{x}}$-"terms" is equal to the Boolean unsatisfiability problem (by reading $x, \com{x}, \cap, \union$ as $x, \lnot x, \land, \lor$), which is $\coNP$-hard \cite{cookComplexityTheoremprovingProcedures1971}.
\end{proof}
For the $\com{x}$-free fragment,
the "emptiness problem" is $\NC^{1}$-complete, as follows.
\begin{cor}\label{cor:decidable-PCoR*-emptiness-difference}
    The "emptiness problem" on $\REL$ is $\NC^{1}$-complete for $\PCoR_{\set{\bl^{*}, \com{\eps}}}$.
\end{cor}
\begin{proof}
    \proofcase{(In $\NC^{1}$)}
    We recall the proof of \Cref{thm:decidable-ECoR*-emptiness}.
    Similarly,
    we define $\tilde{\glang}'(\term)$ in the same way as $\tilde{\glang}(\term)$,
    where the definition of $\tilde{\glang}'(x)$ has been replaced with $\tilde{\glang}'(x) \defeq \set{\hspace{-.5em}\begin{tikzpicture}[baseline= -.5ex]
        \graph[grow up= 1.cm, branch right = .8cm, nodes={font=\tiny}]{
            {0/{}[mynode, fill = black, inner sep = .1em], 1/{}[mynode, fill = black, inner sep = .1em]};
        };
        \node[left = 4pt of 0](s){};\node[right = 4pt of 1](t){};
        \graph[use existing nodes, edges={color=black, pos = .5, earrow}, edge quotes={fill = white, inner sep=1pt,font= \scriptsize}]{
        (s) -> (0); (1) -> (t);
        };
    \end{tikzpicture}\hspace{-.5em}}$ for $x \in \vsig$.
    Because $\com{x}$ does not appear in $\PCoR_{\set{\bl^{*}, \com{\eps}}}$, we have that ($\REL \not\models \term \le \emp$ iff) $\tilde{\glang}(\term) \neq \emptyset$ iff $\tilde{\glang}'(\term) \neq \emptyset$.
    The universe of $\tilde{\glang}'(\term)$ is bounded by a fixed finite set (of $2^{1^2} + 2^{2^2} = 18$ $\com{\eps}$-labelled "graphs" up to graph isomorphisms).
    Thus, we can reduce this problem to the evaluation of a fixed finite algebra (with $2^{18}$ values),
    which is in $\NC^{1}$ \cite[Corollary 1]{lohreyParallelComplexityTree2001}.

    \proofcase{(Hardness)}
    By the same proof as \Cref{prop:decidable-PCoR*-DREL-emptiness}.
\end{proof}
By modifying the hardness proof of \Cref{thm:decidable-ECoR*-emptiness}, we also have the following lower bound for the "inclusion problem".
\begin{prop}\label{prop:PSPACE-hard-PCoR*}
    The "inclusion problem" $\term[1] \le \term[2]$ on $\REL$, where $\term[1]$ is a $\PCoR_{\set{\bl^{*}}}$ "term" and $\term[2]$ is a $\PCoR$ "term", is $\PSPACE$-hard.
\end{prop}
\begin{proof}
    We recall the $\PSPACE$-complete problem in the proof of \Cref{thm:decidable-ECoR*-emptiness}.
    Let $\term[3] \defeq \bigcompo_{j = 0}^{m-1} (p_j \union \tilde{p}_j)$.
    We then have:
    \begin{align*}
        &\REL_{\ttest} \models b_1 \mathtt{x} c_1 \le \emp \to \dots \to b_n \mathtt{x} c_n \le \emp \to b \mathtt{x}^* c \le \emp\\
        &\Longleftrightarrow
        \REL_{\ttest} \models b \mathtt{x}^* c \le \top \Big(\sum_{i = 1}^{n} b_i \mathtt{x} c_i\Big) \top \tag{\Cref{prop:Hoare-hypotheses}}\\
        &\Longleftrightarrow
        \REL \models \Tr_{\term[3]}(b \mathtt{x}^* c)[p_0^{\lop}, \dots, p_{m-1}^{\lop}, \tilde{p}_0^{\lop}, \dots, \tilde{p}_{m-1}^{\lop}/p_0, \dots, p_{m-1}, \tilde{p}_0, \dots, \tilde{p}_{m-1}]\\
        &\hspace{8em} \le \top \Big(\sum_{i = 1}^{n} b_i \mathtt{x} c_i\Big) \top \union \top \Big(\sum_{j = 0}^{m-1} p_j \tilde{p}_j\Big) \top. \tag{\Cref{ex:tests}}
    \end{align*}
    As the left-hand side term is a $\PCoR_{\set{\bl^{*}}}$ "term" and the right-hand side term is a $\PCoR$ "term", this completes the proof.
\end{proof}
When $\com{\eps}$ and $\com{x}$ occur only in the left-hand side, they can be eliminated, as follows.
\begin{prop}\label{prop:eliminating-complements-left}
    The "inclusion problem" $\term[1] \le \term[2]$ on $\REL$, where $\term[1]$ is a $\PCoR_{\set{\bl^{*}, \com{\eps}, \com{x}}}$ "term" and $\term[2]$ is a $\PCoR_{\set{\bl^{*}}}$ (resp.\ $\PCoR$) "term",
    is polynomial-time reducible to the case where $\term[1]$ is a $\PCoR_{\set{\bl^{*}}}$ "term" and $\term[2]$ is a $\PCoR_{\set{\bl^{*}}}$ (resp.\ $\PCoR$) "term".
    Hence, the former problem for $\PCoR_{\set{\bl^{*}}}$ "terms" $\term[2]$ is $\EXPSPACE$-complete.
\end{prop}
\begin{proof}
    Let $z$ and $x'$ (for each "variable" $x$ such that $\com{x}$ occurs in $\term[1]$) be pairwise distinct "fresh" "variables",
    let $\term[1]'$ be the $\PCoR_{\set{\bl^{*}}}$ "term" obtained from $\term[1]$ by replacing $\com{\eps}$ with $z$ and each $\com{x}$ with $x'$,
    and let $\Gamma \defeq \set{z^{\lop} \le \emp} \cup \set{x \cap x' \le \emp \mid \mbox{$\com{x}$ occurs in $\term[1]$}}$.
    We then have:
    \begin{align*}
        \REL \models \term[1] \le \term[2]
        &\;\Longleftrightarrow\; \REL_{\Gamma} \models \term[1]' \le \term[2]
        \tag{see below}\\
        &\;\Longleftrightarrow\; \REL \models \term[1]' \le \term[2] \union \top \Big(z^{\lop} \union \sum_{x} (x \cap x')\Big) \top.
        \tag{\Cref{prop:Hoare-hypotheses}}
    \end{align*}
    We prove the first equivalence in both directions.

    \proofcase{($\Rightarrow$)}
    By $\REL_{\Gamma} \models \term[1]' \le \term[1]$, which holds as $\REL_{\Gamma} \models z \le \com{\eps}$ and $\REL_{\Gamma} \models x' \le \com{x}$ and all the operators of $\PCoR_{\set{\bl^{*}, \com{\eps}, \com{x}}}$ are monotone.

    \proofcase{($\Leftarrow$)}
    By replacing $\val(z)$ with $\hat{\val}(\com{\eps})$ and $\val(x')$ with $\hat{\val}(\com{x})$ for each $\val \in \REL$ (as $z$ and $x'$ are "fresh").

    As $z^{\lop} \union \sum_{x} (x \cap x')$ is a $\PCoR$ "term", this completes the proof.
    The $\EXPSPACE$-completeness follows from that of the "equational theory" of $\PCoR_{\set{\bl^{*}}}$
    (the upper bound \cite{nakamuraDerivativesGraphsPositive2025} and the lower bound \cite{brunetPetriAutomata2017,nakamuraPartialDerivativesGraphs2017}).
\end{proof}

\subsection{Remark on the coNP-Hardness of the Star-Free Fragments on DREL}
Additionally, we have the following $\coNP$-hardness results by a similar argument to \Cref{cor:NP-complete-while-emptiness}.
(When the number of "variables" is not fixed, they are trivially $\coNP$-hard in the same way as \Cref{rem:NP-complete}, and in $\coNP$ by \Cref{prop:upper-bound}~(\ref{prop:upper-bound-2}).)
\begin{cor}\label{cor:undecidable-PCoR-variable-emptiness}
    The "emptiness problem" on $\DREL$ is $\coNP$-complete
    for $\PCoR_{\set{\com{x}}}$ "terms" (more precisely, $\set{\compo, \union, \bl^{\lop}, \com{x}}$-"terms") such that
    the number of "variables" is $3$ (corresponding to $2$ "actions" and $1$ "primitive test" in \Cref{cor:NP-complete-while-emptiness}).
\end{cor}
\begin{proof}
    \proofcase{(In $\coNP$)}
    By \Cref{prop:upper-bound}~(\ref{prop:upper-bound-2}).

    \proofcase{(Hardness)}
    By \Cref{cor:NP-complete-while-emptiness} and by eliminating $\ttest$ using \Cref{ex:substitution-tests},
    i.e., by substituting $q^{\lop}$ and $\com{q}^{\lop}$ for $\tP$ and $\tilde{\tP}$, respectively, where $q$ is a "fresh" "variable".
\end{proof}

\begin{cor}\label{cor:undecidable-PCoR-difference-emptiness}
    The "emptiness problem" on $\DREL$ is $\coNP$-complete
    for $\PCoR_{\set{\com{\eps}}}$ "terms" (more precisely, $\set{\compo, \union, \bl^{\lop}, \com{\eps}}$-"terms") such that
    the number of "variables" is $3$ (corresponding to $2$ "actions" and $1$ "primitive test" in \Cref{cor:NP-complete-while-emptiness}).
\end{cor}
\begin{proof}
    \proofcase{(In $\coNP$)}
    By \Cref{prop:upper-bound}~(\ref{prop:upper-bound-2}).

    \proofcase{(Hardness)}
    By \Cref{cor:NP-complete-while-emptiness} and by eliminating $\ttest$ as in the proof of \Cref{cor:undecidable-PCoR*-difference-emptiness},
    i.e., by substituting $q^{\lop}$ and $(q \com{\eps})^{\lop}$ for $\tP$ and $\tilde{\tP}$, respectively, where $q$ is a "fresh" "variable".
\end{proof}

\section{Conclusion and Future Directions}\label{section:conclusion}
We have shown that
the "emptiness problem" of deterministic "propositional while programs with loop" is $\PiOne$-complete (\Cref{thm:undecidable-while-emptiness}), and that this remains so even if the number of "actions" is $2$ and the number of "primitive tests" is $1$ (\Cref{thm:reducing-primitive-tests,cor:reducing-primitive-tests-conv}).
This yields a dichotomy: the "emptiness problem" is decidable if the number of "actions" is at most $1$ or the number of "primitive tests" is $0$, and $\PiOne$-complete otherwise (\Cref{cor:reducing-primitive-tests-2}).
Using this, we have shown that the "equational theory" of $\PCoR_{\set{\bl^{*}, \com{\eps}}}$ "terms" on $\REL$ is $\PiOne$-complete (\Cref{cor:undecidable-PCoR*-difference}) via "hypothesis eliminations".
Moreover, we have shown that the "emptiness problem" of $\PCoR_{\set{\bl^{*}, \com{\eps}, \com{x}}}$ "terms" with respect to $\REL$ is decidable and $\PSPACE$-complete (\Cref{thm:decidable-ECoR*-emptiness}).
In \Cref{table:complexity}, we summarize results for fragments of $\PCoR_{\set{\bl^{*}, \com{\eps}, \com{x}}}$ with respect to $\REL$ and $\DREL$.
For the cases marked with ``$?$'', the tight complexity is open.
\begin{table}[!t]
    \newcommand{\unimportant}[1]{\textcolor{gray}{\footnotesize #1}}
    \Crefname{cor}{Cor.}{Cor.}\Crefname{prop}{Prop.}{Prop.}\Crefname{thm}{Thm.}{Thm.}%
    \newcommand{\cell}[1]{\begin{tabular}{@{}c@{}}#1\end{tabular}}%
    \newcommand{\ucell}[1]{{\renewcommand{\arraystretch}{.8}\cell{#1}}}%
    \newcommand{\rsa}{\rule[-.85em]{0pt}{2.2em}}%
    \newcommand{\rsb}{\rule[-.62em]{0pt}{1.75em}}%
    \newcommand{\bdw}{1.2pt}%
    \newcommand{\bdr}[1]{\multicolumn{1}{c!{\vrule width \arrayrulewidth\hskip\dimexpr\doublerulesep+\arrayrulewidth-\bdw\relax\vrule width \bdw}}{#1}}%
    \centering
    \subfloat[][$\algclass = \REL$]{
        \scalebox{.70}{
        \begin{tabular}{c||c||c|c||c|c||c|c|c|c}
            $t$ $\backslash$ $s$ & $s = \emp$ & ${\set{}}$ & ${\set{\bl^{*}}}$ & ${\set{\com{\eps}}}$ & ${\set{\bl^{*}, \com{\eps}}}$  & ${\set{\com{x}}}$  & ${\set{\bl^{*}, \com{x}}}$ & ${\set{\com{\eps}, \com{x}}}$  & ${\set{\bl^{*}, \com{\eps}, \com{x}}}$\\
            \hline
            \hline
            \rsa${\set{}}$ & \multirow{2}{*}{\cell{$\NC^{1}$-c\\(\Cref{cor:decidable-PCoR*-emptiness-difference})}} & \multicolumn{8}{c}{\multirow{4}{*}{\ucell{\unimportant{$\coNP$-c}\\\unimportant{\cite[Theorem 24]{nakamuraExistentialCalculiRelations2023}, \Cref{prop:upper-bound}}}}} \\
            \cline{1-1}
            \rsa${\set{\com{\eps}}}$ & \\
            \cline{1-2}
            \rsa${\set{\com{x}}}$ & \multirow{2}{*}{\cell{$\coNP$-c\\(\Cref{cor:decidable-ECoR*-emptiness})}} \\
            \cline{1-1}
            \rsa${\set{\com{\eps}, \com{x}}}$ & \\
            \cline{1-4}\Xcline{5-10}{\bdw}
            \rsa${\set{\bl^{*}}}$ & \multirow{2}{*}{\cell{$\NC^{1}$-c\\(\Cref{cor:decidable-PCoR*-emptiness-difference})}} & \cell{$\PSPACE$-h (\Cref{prop:PSPACE-hard-PCoR*})\\\unimportant{in $\EXPSPACE$} (?)} & \bdr{\ucell{\unimportant{$\EXPSPACE$-c}\\\unimportant{\cite{brunetPetriAutomataKleene2015,brunetPetriAutomata2017,nakamuraPartialDerivativesGraphs2017,nakamuraDerivativesGraphsPositive2025}}}} & \multicolumn{2}{c||}{\multirow{4}{*}{\cell{$\PiOne$-c\\(\Cref{cor:undecidable-PCoR*-difference}, \Cref{prop:upper-bound})}}} & \multicolumn{4}{c}{\multirow{4}{*}{\ucell{\unimportant{$\PiOne$-c}\\\unimportant{\cite[Theorem 50]{nakamuraExistentialCalculiRelations2023}, \Cref{prop:undecidable-PCoR*-variable}, \Cref{prop:upper-bound}}}}} \\
            \cline{1-1}\cline{3-4}
            \rsa${\set{\bl^{*}, \com{\eps}}}$ & & \multirow{3}{*}{\cell{$\PSPACE$-h (\Cref{prop:PSPACE-hard-PCoR*})\\\unimportant{in $\EXPSPACE$} (?)}} & \bdr{\multirow{3}{*}{\cell{$\EXPSPACE$-c\\(\Cref{prop:eliminating-complements-left})}}} & \multicolumn{2}{c||}{} & \multicolumn{4}{c}{}\\
            \cline{1-2}
            \rsa${\set{\bl^{*}, \com{x}}}$ & \multirow{2}{*}{\cell{$\PSPACE$-c\\(\Cref{thm:decidable-ECoR*-emptiness})}} & & \bdr{} & \multicolumn{2}{c||}{} & \multicolumn{4}{c}{}\\
            \cline{1-1}
            \rsa${\set{\bl^{*}, \com{\eps}, \com{x}}}$ & & & \bdr{} & \multicolumn{2}{c||}{} & \multicolumn{4}{c}{}
        \end{tabular}}
    }\\
    \subfloat[][$\algclass = \DREL$]{
        \scalebox{.78}{
    \begin{tabular}{c||c||c|c||c|c||c|c|c|c}
        $t$ $\backslash$ $s$ & $s = \emp$ & $\set{}$ & ${\set{\bl^{*}}}$ & ${\set{\com{\eps}}}$ & ${\set{\bl^{*}, \com{\eps}}}$  & ${\set{\com{x}}}$  & ${\set{\bl^{*}, \com{x}}}$ & ${\set{\com{\eps}, \com{x}}}$  & ${\set{\bl^{*}, \com{\eps}, \com{x}}}$\\
        \hline\hline
        \rsb${\set{}}$ & $\NC^{1}$-c (\Cref{prop:decidable-PCoR*-DREL-emptiness}) & \multicolumn{8}{c}{\multirow{4}{*}{\cell{$\coNP$-c\\(\Cref{cor:coNP-PCoR-func}, \Cref{prop:upper-bound})}}} \\
        \cline{1-2}
        \rsb${\set{\com{\eps}}}$ & $\coNP$-c (\Cref{cor:undecidable-PCoR-difference-emptiness}, \Cref{prop:upper-bound}) \\
        \cline{1-2}
        \rsb${\set{\com{x}}}$ & \multirow{2}{*}{$\coNP$-c (\Cref{cor:undecidable-PCoR-variable-emptiness}, \Cref{prop:upper-bound})} \\
        \cline{1-1}
        \rsb${\set{\com{\eps}, \com{x}}}$ & \\
        \cline{1-2}\Xcline{3-10}{\bdw}
        \rsb${\set{\bl^{*}}}$ & \bdr{$\NC^{1}$-c (\Cref{prop:decidable-PCoR*-DREL-emptiness})} & \multicolumn{8}{c}{\multirow{4}{*}{$\PiOne$-c (\Cref{cor:undecidable-PCoR*-func}, \Cref{prop:upper-bound})}} \\
        \cline{1-1}\Xcline{2-2}{\bdw}
        \rsb${\set{\bl^{*}, \com{\eps}}}$ & $\PiOne$-c (\Cref{cor:undecidable-PCoR*-difference-emptiness}, \Cref{prop:upper-bound}) \\
        \cline{1-2}
        \rsb${\set{\bl^{*}, \com{x}}}$ & \multirow{2}{*}{$\PiOne$-c (\Cref{cor:undecidable-PCoR*-variable-emptiness}, \Cref{prop:upper-bound})} \\
        \cline{1-1}
        \rsb${\set{\bl^{*}, \com{\eps}, \com{x}}}$ &
    \end{tabular}}
    }
    \caption{On the complexity of $\algclass \models \term[1] \le \term[2]$ (where the number of \kl{variables} is not fixed).
    Here, each $X \subseteq \set{\bl^{*}, \com{\eps}, \com{x}}$ denotes the \kl{term} set $\PCoR_{X}$, and the cases marked with ``$?$'' are open.}
    \label{table:complexity}
\end{table}

Additionally, for deterministic "propositional while programs with loop",
it would be interesting to consider whether the undecidability result holds for more restricted syntaxes of "propositional while programs with loop", e.g., by bounding the nesting of the "graph loop operator" and the "Kleene star", cf.\ the folk theorem \cite{harelFolkTheorems1980,kozenKleeneAlgebraTests1997} and its failure in the propositional setting \cite{kozenBohmJacopiniTheorem2008}.

\bibliographystyle{fundam}
\bibliography{main}

\appendix 
\crefalias{section}{appendix}
\crefalias{subsection}{appendix}

\section{\texorpdfstring{On \Cref{rem:Hoare-KA}: $\top$ vs $\Sigma^*$ in Hoare Hypothesis Eliminations}{On Remark: T vs Sigma* in Hoare Hypothesis Eliminations}}\label{section:Hoare-hypothesis-KA}
In this section, we give the "Hoare hypothesis" elimination using $\Sigma^*$ (of \Cref{rem:Hoare-KA}) via \Cref{prop:Hoare-hypotheses}.
In the sequel, we consider ""Kleene lattice"" ("KL") "terms": the class of $\set{\eps, \emp, \compo, \union, \bl^{*}, \cap}$-"terms".
Throughout this section, let $\Sigma$ be a finite set of "variables", and we write $\Sigma^*$ for the "term" $(\sum_{a \in \Sigma} a)^*$.
We use the following property of "KL" "terms".
\begin{prop}\label{prop:KL-graph}
    Let $\term$ be a "KL" "term" with $\vsig(\term) \subseteq \Sigma$, and let $\graph[2] \in \glang(\term)$ be a "graph".
    Then, every "vertex" in $\graph[2]$ has a $\Sigma^*$-path from the "source" and has a $\Sigma^*$-path to the "target".
\end{prop}
\begin{proof}
    By easy induction on $\term$.
\end{proof}
For a "valuation" $\val \colon \vsig \to \wp(Z^2)$, "vertices" $x_0, y_0 \in Z$,
and a "term" $\term$, 
the $\term$-""interval"" $\intro*\interval{x_0}{y_0}_{\val, \term}$ is defined as
\[ \interval{x_0}{y_0}_{\val, \term} \defeq \set{z \in Z \mid \tuple{x_0, z} \in \hat{\val}(\term) \land \tuple{z, y_0} \in \hat{\val}(\term)}.\]
Namely, $\interval{x_0}{y_0}_{\val, \term}$ is the set of "vertices" having a $\term$-path from $x_0$ and a $\term$-path to $y_0$.
Using \Cref{prop:KL-graph}, we have the following.
\begin{prop}\label{prop:interval}
    Let $Z$ be a set.
    Let $\val \colon \vsig \to \wp(Z^2)$ in $\REL$.
    Let $x_0, y_0 \in Z$ and $z_1, z_2 \in \interval{x_0}{y_0}_{\val, \Sigma^*}$.
    For all "KL" "terms" $\term[1]$ with $\vsig(\term[1]) \subseteq \Sigma$,
    we have:
    \[\tuple{z_1, z_2} \in \hat{\val}(\term[1]) \quad\Longrightarrow\quad \tuple{z_1, z_2} \in \reallywidehat{(\val \restr \interval{x_0}{y_0}_{\val, \Sigma^*})}(\term[1]).\]
\end{prop}
\begin{proof}
    By $\tuple{z_1, z_2} \in \hat{\val}(\term[1])$ with \Cref{prop:glang}, there is a "graph homomorphism" $h \colon \graph[2] \homo \gval(\val, z_1, z_2)$ for some $\graph[2] \in \glang(\term[1])$.
    By the form of $\graph[2]$ (\Cref{prop:KL-graph}) with $h(\src^{\graph[2]}) = z_1 \in \interval{x_0}{y_0}_{\val, \Sigma^*}$ and $h(\tgt^{\graph[2]}) = z_2 \in \interval{x_0}{y_0}_{\val, \Sigma^*}$,
    we have $h(\domain{\graph[2]}) \subseteq \interval{x_0}{y_0}_{\val, \Sigma^*}$.
    Thus, $\tuple{z_1, z_2} \in \reallywidehat{(\val \restr \interval{x_0}{y_0}_{\val, \Sigma^*})}(\term[1])$ still holds, by the same $h$.
\end{proof}
We say that $\algclass$ is ""interval-closed"" with respect to $\Sigma^*$ (i.e., closed under taking "submodels" with respect to $\Sigma^*$-"intervals") if $\val \restr \interval{x_0}{y_0}_{\val, \Sigma^*} \in \algclass$ for all $\val \in \algclass$ and $\tuple{x_0, y_0} \in \hat{\val}(\Sigma^*)$ (i.e., $\interval{x_0}{y_0}_{\val, \Sigma^*} \neq \emptyset$).
By applying \Cref{prop:Hoare-hypotheses} to such "submodels" with \Cref{prop:interval}, we have the "Hoare hypothesis" elimination of \Cref{rem:Hoare-KA}.
\begin{thm}\label{thm:Hoare-hypothesis-elimination-interval}
    Let $\algclass \subseteq \REL$ be such that $\algclass$ is "interval-closed" with respect to $\Sigma^*$.
    Let $\term[1]$ be a "KL" "term" with $\vsig(\term[1]) \subseteq \Sigma$ and let $\term[2], \term[3]$ be $\PCoR_{\set{\bl^{*}, \com{\eps}, \com{x}}}$ "terms".
    Then,
    \[\algclass_{\term[3] \le \emp} \models \term[1] \le \term[2] \quad\Longleftrightarrow\quad \algclass \models \term[1] \le \term[2] \union \Sigma^* \term[3] \Sigma^*.\]
\end{thm}
\begin{proof}
    \proofcase{($\Longleftarrow$)}
    By $\algclass_{\term[3] \le \emp} \models \Sigma^* \term[3] \Sigma^* \le \emp$.

    \proofcase{($\Longrightarrow$)}
    Let $\val \in \algclass$ and $\tuple{z_1, z_2} \in \hat{\val}(\term[1])$.
    As $\REL \models \term[1] \le \Sigma^*$ (by easy induction on $\term[1]$), we have $\tuple{z_1, z_2} \in \hat{\val}(\Sigma^*)$, and hence $z_1, z_2 \in \interval{z_1}{z_2}_{\val, \Sigma^*}$.
    Let $\val' \defeq \val \restr \interval{z_1}{z_2}_{\val, \Sigma^*}$; we have $\val' \in \algclass$, as $\algclass$ is "interval-closed" with respect to $\Sigma^*$.
    By \Cref{prop:interval}, $\tuple{z_1, z_2} \in \hat{\val}'(\term[1])$.
    By \Cref{prop:Hoare-hypotheses}, $\tuple{z_1, z_2} \in \hat{\val}'(\term[2] \union \top \term[3] \top)$.
    By construction of $\val'$ (as every $\Sigma^*$-path in $\val$ from $z_1$ to $z_2$ is included in $\interval{z_1}{z_2}_{\val, \Sigma^*}$),
    $\tuple{z_1, w} \in \hat{\val}'(\top) \Leftrightarrow \tuple{z_1, w} \in \hat{\val}'(\Sigma^*)$ and $\tuple{w, z_2} \in \hat{\val}'(\top) \Leftrightarrow \tuple{w, z_2} \in \hat{\val}'(\Sigma^*)$ for all $w$,
    and thus $\tuple{z_1, z_2} \in \hat{\val}'(\top \term[3] \top)$ iff $\tuple{z_1, z_2} \in \hat{\val}'(\Sigma^* \term[3] \Sigma^*)$,
    and hence $\tuple{z_1, z_2} \in \hat{\val}'(\term[2] \union \Sigma^* \term[3] \Sigma^*)$.
    Hence, as $\val'$ is a "submodel" of $\val$, we have $\tuple{z_1, z_2} \in \hat{\val}(\term[2] \union \Sigma^* \term[3] \Sigma^*)$.
\end{proof}

\begin{rem}
    \Cref{thm:Hoare-hypothesis-elimination-interval} 
    does not hold for an arbitrary class $\algclass \subseteq \REL$, cf.\ \Cref{prop:Hoare-hypotheses}.
    For example, let us consider the class $\algclass \defeq \set{\val}$ where the "valuation" $\val$ is given by the following "graph":
    $\begin{tikzpicture}[baseline = -.5ex]
            \graph[grow up= 1.cm, branch right = 1.cm, nodes={font=\tiny}]{
            {0/{$0$}[mynode, fill = gray!30], 1/{$1$}[mynode, fill = gray!30], 2/{$2$}[mynode, fill = gray!30]}
            };
            \graph[use existing nodes, edges={color=black, pos = .5, earrow}, edge quotes={fill = white, inner sep=1pt,font= \scriptsize}]{
            (0) ->["$a$"] (1) ->["$a$"] (2);
            };
        \end{tikzpicture}$.
    Then, $\algclass_{a a \le \emp} \models a \le \emp$ (by $\algclass_{a a \le \emp} = \emptyset$) and $\algclass \not\models a \le \emp \union \Sigma^* a a \Sigma^*$ where $\Sigma =\set{a}$ (by $\tuple{0,1} \in \hat{\val}(a) \setminus \hat{\val}(\Sigma^* a a \Sigma^*)$).
\end{rem}

\subsection{On REL with Hypotheses}\label{subsection:Hoare-Hypothesis-elimination-hypotheses}
\Cref{thm:Hoare-hypothesis-elimination-interval} can be applied to the "Hoare hypothesis" elimination for $\REL_{\Gamma}$, as follows.
\begin{cor}\label{cor:Hoare-hypothesis-elimination}
    Let $\term[2]_i', \term[1]$ be "KL" "terms" with $\vsig(\term[2]_i'), \vsig(\term[1]) \subseteq \Sigma$
    and let $\term[1]_i', \term[2], \term[3]$ be $\PCoR_{\set{\bl^{*}, \com{\eps}, \com{x}}}$ "terms" where $i$ ranges over an index set $I$,
    and let $\Gamma \defeq \set{\term[1]_i' \le \term[2]_i' \mid i \in I}$.
    We have:
    \[\REL_{\Gamma, \term[3] \le \emp} \models \term[1] \le \term[2] \quad\Longleftrightarrow\quad \REL_{\Gamma} \models \term[1] \le \term[2] \union \Sigma^* \term[3] \Sigma^*.\]
\end{cor}
\begin{proof}
    Let $\val \in \REL_{\Gamma}$, $\tuple{x_0, y_0} \in \hat{\val}(\Sigma^*)$, and $\val' \defeq \val \restr \interval{x_0}{y_0}_{\val, \Sigma^*}$.
    For each $i \in I$ and $\tuple{z_1, z_2} \in \hat{\val}'(\term[1]_i')$,
    we have $\tuple{z_1, z_2} \in \hat{\val}(\term[1]_i')$, as $\val'$ is a "submodel" of $\val$,
    so $\tuple{z_1, z_2} \in \hat{\val}(\term[2]_i')$ by $\val \models \term[1]_i' \le \term[2]_i'$,
    and thus $\tuple{z_1, z_2} \in \hat{\val}'(\term[2]_i')$ by \Cref{prop:interval}.
    Hence, $\val' \in \REL_{\Gamma}$; i.e., $\REL_{\Gamma}$ is "interval-closed" with respect to $\Sigma^*$.
    This completes the proof by \Cref{thm:Hoare-hypothesis-elimination-interval}.
\end{proof}
\Cref{cor:Hoare-hypothesis-elimination} is an extension of the "hypothesis elimination" known for "Kleene algebra" ("KA") "terms" (see \cite[Theorem 27]{hardinProofTheoryKleene2005}, \cite[Theorem 3.2 (3.3)]{hardinModularizingElimination$r0$2005} for $\REL_{\Gamma}$); below is a slightly simplified version of \Cref{cor:Hoare-hypothesis-elimination} for "KL" "terms".
\begin{cor}\label{cor:Hoare-hypothesis-elimination-KL}
    Let $\Gamma$ be a set of "KL" "equations" and let $\term[1], \term[2], \term[3]$ be "KL" "terms".
    When $\vsig(\Gamma), \vsig(\term[1]) \subseteq \Sigma$, we have:
    \[\REL_{\Gamma, \term[3] \le \emp} \models \term[1] \le \term[2] \quad\Longleftrightarrow\quad \REL_{\Gamma} \models \term[1] \le \term[2] \union \Sigma^* \term[3] \Sigma^*.\]
\end{cor}
\begin{proof}
    By \Cref{cor:Hoare-hypothesis-elimination}.
\end{proof}

\subsection{On the Converse Extension}\label{subsection:Hoare-Hypothesis-elimination-converse}
As a corollary of \Cref{thm:Hoare-hypothesis-elimination-interval},
when we extend "KL" "terms" $\term$ with converse ($\bl^{\smile}$) (so-called ""Kleene allegory"" "terms"), we have the following "Hoare hypothesis" elimination.
\begin{cor}\label{cor:Hoare-hypothesis-elimination-interval-KAL}
    Let $\breve{\Sigma} \defeq \set{a, a^{\smile} \mid a \in \Sigma}$.
    Let $\algclass \subseteq \REL$ be such that $\algclass$ is "interval-closed" with respect to $\breve{\Sigma}^*$.
    Let $\term[1]$ be a "Kleene allegory" "term" with $\vsig(\term[1]) \subseteq \Sigma$ and let $\term[2], \term[3]$ be $\PCoR_{\set{\bl^{*}, \com{\eps}, \com{x}}}$ "terms".
    Then,
    \[\algclass_{\term[3] \le \emp} \models \term[1] \le \term[2] \quad\Longleftrightarrow\quad \algclass \models \term[1] \le \term[2] \union \breve{\Sigma}^{*} \term[3] \breve{\Sigma}^{*}.\]
\end{cor}
\begin{proof}[Proof Sketch]
    By pushing the converse operator inside as much as possible and viewing each $a^{\smile}$ (where $a \in \Sigma$) as a "variable",
    we can apply \Cref{thm:Hoare-hypothesis-elimination-interval}.
\end{proof}

\section{More Examples of Substitution Hypotheses}\label{section:substitution-examples}
In this section, we give more examples of "substitution hypotheses" (\Cref{prop:REL-axiom-=-substitution}).
\begin{ex}[{\cite[Lemma 3.8 (v)]{pousToolsCompletenessKleene2024}}]\label{ex:substitution-transitive}
    Consider $xx \le x$ where $x$ is a "variable".
    Let $\Gamma$ be a set of "equations" such that $x \not\in \vsig(\Gamma)$.
    Then $x$ is "fresh" in $\REL_{\Gamma}$.
    Thus we have:
    \[\REL_{\Gamma, xx \le x} \models \term[1] = \term[2] \quad\Longleftrightarrow\quad \REL_{\Gamma} \models (\term[1] = \term[2])[x^+/x],\]
    because $\REL \models x x \le x \leftrightarrow x = x^{+}$
    and $\REL \models x^{+} = x^{+}[x^{+}/x]$.
\end{ex}
\begin{ex}[{\cite[Lemma 3.8 (iv)]{pousToolsCompletenessKleene2024}}]\label{ex:substitution-transitive-2}
    Consider $x \term[3] \le x$ where $x$ is a "variable" and $x \not\in \vsig(\term[3])$.
    Let $\Gamma$ be a set of "equations" such that $x \not\in \vsig(\Gamma)$ (so $x$ is "fresh" in $\REL_{\Gamma}$). Then we have:
    \[\REL_{\Gamma, x\term[3] \le x} \models \term[1] = \term[2] \quad\Longleftrightarrow\quad \REL_{\Gamma} \models (\term[1] = \term[2])[x \term[3]^*/x],\]
    because $\REL \models x \term[3] \le x \leftrightarrow x = x \term[3]^*$
    and $\REL \models x \term[3]^* = (x \term[3]^*)[x \term[3]^*/x]$.
\end{ex}
\begin{ex}[{\cite[Lemma 3.4]{hardinEliminationHypothesesKleene2002}}]\label{ex:substitution-unused-code}
    Consider $b x = b$ where $x$ is a "variable" and $b$ is a "test" (here, $x \not\in \vsig(b)$).
    Let $\Gamma$ be a set of "equations" such that $x \not \in \vsig(\Gamma) \cup \psig$ (so $x$ is "fresh" in $\REL_{\ttest, \Gamma}$). Then we have:
    \[\REL_{\ttest, \Gamma, b x = b} \models \term[1] = \term[2] \quad\Longleftrightarrow\quad \REL_{\ttest, \Gamma} \models (\term[1] = \term[2])[\tcompl{b} x \union b/x],\]
    because $\REL_{\ttest} \models b x = b \leftrightarrow x = \tcompl{b} x \union b$ \cite[Lemma 2.1]{hardinEliminationHypothesesKleene2002} and $\REL_{\ttest} \models \tcompl{b} x \union b = (\tcompl{b} x \union b)[\tcompl{b} x \union b/x]$.
\end{ex}

\section{A Reduction from the Difference Constant to Variable Complements}\label{section:difference-encoding}
In this section,
we show that we can eliminate the "difference constant" $\com{\eps}$ using "variable complements".
Let $\tQ$ be a special "variable".
Let $\intro*\TrbarI_{\tQ}$ be the translation defined in the same way as $\Tr_{\tQ}$ (\Cref{defi:loop-transformation}), except that $\TrbarI_{\tQ}(\com{\eps}) \defeq \com{\tQ}$.
For each "valuation" $\val$ such that $\val(\tQ)$ is an equivalence relation,
we write $\intro*\quo{x}_{\sim_{\tQ}}$ for the equivalence class of $x$ with respect to the equivalence relation $\val(\tQ)$
and write $\val/\tQ$ for the ""quotient"" of $\val$ with respect to the equivalence relation $\val(\tQ)$,
i.e., $(\val/\tQ)(a) \defeq \set{\tuple{\quo{x}_{\sim_{\tQ}}, \quo{y}_{\sim_{\tQ}}} \mid \tuple{x, y} \in \val(a)}$.
We say that $\val$ is ""compatible"" for a "variable" set $A$ if $\hat{\val}((\tQ a \tQ) \cap (\tQ \com{a} \tQ)) = \emptyset$ for each $a \in A$
(cf.\ the consistency of edges in \cite[Definition 13]{nakamuraExistentialCalculiRelations2023}, where the equivalence closure of $\eps$-labelled edges plays the role of $\tQ$).
\begin{lem}\label{lem:id}
    Let $\val \in \REL$.
    Suppose that $\val(\tQ)$ is an equivalence relation and $\val$ is "compatible" for a set $A$.
    For all $\PCoR_{\set{\bl^{*}, \com{\eps}, \com{x}}}$ "terms" $\term[1]$ such that $\vsig(\term[1]) \subseteq A$, we have:
    \[\tuple{\quo{x}_{\sim_{\tQ}}, \quo{y}_{\sim_{\tQ}}} \in \widehat{\val/{\tQ}}(\term[1]) \quad\iff\quad \tuple{x, y} \in \hat{\val}(\TrbarI_{\tQ}(\term[1])).\]
\end{lem}
\begin{proof}
    By easy induction on $\term$.
    (For instance, for $\com{\eps}$, we have $\quo{x}_{\sim_{\tQ}} \neq \quo{y}_{\sim_{\tQ}}$ iff $\tuple{x, y} \in \hat{\val}(\com{\tQ})$;
    for $\com{a}$, as $\val$ is "compatible", we have $\tuple{\quo{x}_{\sim_{\tQ}}, \quo{y}_{\sim_{\tQ}}} \not\in (\val/\tQ)(a)$ iff $\tuple{x, y} \in \hat{\val}(\tQ \com{a} \tQ)$;
    for $\term^{\lop}$, we have $\tuple{\quo{x}_{\sim_{\tQ}}, \quo{y}_{\sim_{\tQ}}} \in \widehat{\val/{\tQ}}(\term^{\lop})$ iff $\quo{x}_{\sim_{\tQ}} = \quo{y}_{\sim_{\tQ}}$ and $\tuple{\quo{x}_{\sim_{\tQ}}, \quo{x}_{\sim_{\tQ}}} \in \widehat{\val/{\tQ}}(\term)$ iff $\tuple{x, y} \in \hat{\val}(\tQ \TrbarI_{\tQ}(\term)^{\lop} \tQ)$ by IH.)
\end{proof}

\begin{prop}\label{prop:id}
    Let $\algclass \subseteq \REL$ and $A \subseteq \vsig$.
    Suppose that, for all $\val \in \algclass$, $\val(\tQ)$ is an equivalence relation and $\val$ is "compatible" for $A$.
    When $\algclass$ is closed under "quotients" with respect to $\tQ$,
    we have that, for all $\PCoR_{\set{\bl^{*}, \com{\eps}, \com{x}}}$ "terms" $\term[1], \term[2]$ such that $\vsig(\term[1]) \cup \vsig(\term[2]) \subseteq A$,
    \[\algclass_{\tQ = \eps} \models \term[1] \le \term[2] \quad\iff\quad \algclass \models \TrbarI_{\tQ}(\term[1]) \le \TrbarI_{\tQ}(\term[2]).\]
\end{prop}
\begin{proof}
    \proofcase{($\Longleftarrow$)}
    Because $\algclass \models \tQ = \eps \to \TrbarI_{\tQ}(\term[3]) = \term[3]$ for all $\term[3]$ (cf.\ \Cref{prop:loop-transformation}~(\ref{prop:loop-transformation-id})).

    \proofcase{($\Longrightarrow$)}
    By \Cref{lem:id} with $\val/{\tQ} \in \algclass_{\tQ = \eps}$ for all $\val \in \algclass$ (as $\algclass$ is closed under "quotients" with respect to $\tQ$ and $\val/{\tQ} \models \tQ = \eps$).
\end{proof}

\begin{thm}\label{thm:id}
    For $\set{\top, \compo, \union, \bl^{\lop}} \subseteq S \subseteq \set{\eps, \emp, \top, \compo, \union, \bl^{\smile}, \bl^{*}, \cap, \bl^{\lop}, \com{\eps}, \com{x}}$,
    there is a polynomial-time reduction from the "equational theory" for $S$-"terms" to "that@equational theory" for $((S \setminus \set{\com{\eps}}) \cup \set{\com{x}})$-"terms", with respect to $\REL$.
\end{thm}
\begin{proof}
    Let $\term[1], \term[2]$ be $S$-"terms".
    Let $A = \set{a_1, \dots, a_n}$ be a finite set (of "variables") such that $\vsig(\term[1]) \cup \vsig(\term[2]) \subseteq A$
    and let $\tQ$ be a "fresh" "variable".
    Let us consider the following "hypotheses":
    \begin{align*}
        \mathtt{i}_0 &\defeq \set{\tQ \le \eps},\\
        \mathtt{i}_1 &\defeq \set{\tQ \ge \eps},\\
        \mathtt{i}_2 &\defeq \set{(\tQ \com{\tQ})^{\lop} = \emp }, \tag*{($\REL \models (\tQ \com{\tQ})^{\lop} = \emp \leftrightarrow \tQ \le \tQ^{\smile}$)}\\
        \mathtt{i}_3 &\defeq \set{(\tQ \tQ \com{\tQ}^{\smile})^{\lop} = \emp}, \tag*{($\REL \models (\tQ \tQ \com{\tQ}^{\smile})^{\lop} = \emp \leftrightarrow \tQ \tQ \le \tQ$)}\\
        \mathtt{i}_4 &\defeq \set{(a \tQ \com{a}^{\smile} \tQ)^{\lop} = \emp \mid a \in A}. \tag*{($\REL_{\mathtt{i}_2 \cup \mathtt{i}_3} \models (a \tQ \com{a}^{\smile} \tQ)^{\lop} = \emp \leftrightarrow (\tQ a \tQ) \cap (\tQ \com{a} \tQ) = \emp$)}
    \end{align*}
    By construction, for each $\val \in \REL$, $\val \models \mathtt{i}_1 \cup \mathtt{i}_2 \cup \mathtt{i}_3 \cup \mathtt{i}_4$ iff $\val(\tQ)$ is an equivalence relation (by $\mathtt{i}_1 \cup \mathtt{i}_2 \cup \mathtt{i}_3$) and $\val$ is "compatible" for $A$ (by $\mathtt{i}_4$).
    Let $\term[3] \defeq (\tQ \com{\tQ})^{\lop} \union (\tQ \tQ \com{\tQ}^{\smile})^{\lop} \union \sum_{a \in A} (a \tQ \com{a}^{\smile} \tQ)^{\lop}$ (i.e., the sum of the left-hand sides of $\mathtt{i}_2 \cup \mathtt{i}_3 \cup \mathtt{i}_4$).
    We then have:
    \begin{align*}
        & \REL \models \term[1] \le \term[2]\\
        & \Leftrightarrow \REL_{\mathtt{i}_0 \cup \mathtt{i}_1 \cup \mathtt{i}_2 \cup \mathtt{i}_3 \cup \mathtt{i}_4} \models \term[1] \le \term[2] \tag{by replacing $\val(\tQ)$ with the identity relation, as $\tQ$ is "fresh"}\\
        & \Leftrightarrow \REL_{\mathtt{i}_1 \cup \mathtt{i}_2 \cup \mathtt{i}_3 \cup \mathtt{i}_4} \models \TrbarI_{\tQ}(\term[1]) \le \TrbarI_{\tQ}(\term[2]) \tag{\Cref{prop:id}}\\
        & \Leftrightarrow \REL_{\mathtt{i}_1} \models \TrbarI_{\tQ}(\term[1]) \le \TrbarI_{\tQ}(\term[2]) \union \top \term[3] \top \tag{\Cref{prop:Hoare-hypotheses}}\\
        & \Leftrightarrow \REL \models \Tr_{\tQ^{\lop}}(\TrbarI_{\tQ}(\term[1])) \le \TrbarI_{\tQ}(\term[2]) \union \top \term[3] \top. \tag{\Cref{thm:hypothesis-elimination-using-graph-loops-2}}
    \end{align*}
    Here, in the second step, $\REL_{\mathtt{i}_1 \cup \mathtt{i}_2 \cup \mathtt{i}_3 \cup \mathtt{i}_4}$ is closed under "quotients" with respect to $\tQ$,
    and, in the last step, $\REL$ is "submodel-closed" and $\REL \models \tQ^{\lop} \le \Tr_{\tQ^{\lop}}(\tQ^{\lop})$.
    Thus, we can reduce the "equational theory" for $S$-"terms" to the "equational theory" for $((S \setminus \set{\com{\eps}}) \cup \set{\com{x}, \bl^{\smile}})$-"terms",
and thus to "that@equational theory" for $((S \setminus \set{\com{\eps}}) \cup \set{\com{x}})$-"terms" by the reduction of \Cref{ex:converse}.
\end{proof}

\section{\texorpdfstring{Detailed Proof of \Cref{prop:glang-u}}{Detailed Proof of Proposition on Loop Extensions}}\label{section:prop:glang-u}
(The following proof proceeds in the same way as \cite[Propositions 8 and 11]{nakamuraExistentialCalculiRelations2023}.
Similar arguments can be found, e.g., in \cite{andrekaEquationalTheoryUnionfree1995,brunetPetriAutomata2017,pousPositiveCalculusRelations2018}.)

We use the following two notations for "graphs", "series-composition" ($\compo$) and "parallel-composition" ($\cap$):
\kldeactivatequotation
\begin{align*}
       G \compo H
        & \quad\defeq\quad \begin{tikzpicture}[baseline = -.5ex]
                        \graph[grow right = 1.cm, branch down = 2.5ex, nodes={mynode, font = \scriptsize}]{
                        {s1/{}[draw, circle]}
                        -!- {c/{}[draw, circle]}
                        -!- {t1/{}[draw, circle]}
                        };
                        \node[left = 4pt of s1](s1l){} edge[earrow, ->] (s1);
                        \node[right = 4pt of t1](t1l){}; \path (t1) edge[earrow, ->] (t1l);
                        \graph[use existing nodes, edges={color=black, pos = .5, earrow}, edge quotes={fill=white, inner sep=1pt,font= \scriptsize}]{
                        s1 ->["$G$"] c;
                        c ->["$H$"] t1;
                        };
                 \end{tikzpicture}
        & G \cap H
        & \quad\defeq\quad \begin{tikzpicture}[baseline = -.5ex]
                        \graph[grow right = 1.cm, branch down = 2.5ex, nodes={mynode, font = \scriptsize}]{
                        {s1/{}[draw, circle]}
                        -!- {t1/{}[draw, circle]}
                        };
                        \node[left = 4pt of s1](s1l){} edge[earrow, ->] (s1);
                        \node[right = 4pt of t1](t1l){}; \path (t1) edge[earrow, ->] (t1l);
                        \graph[use existing nodes, edges={color=black, pos = .5, earrow}, edge quotes={fill=white, inner sep=1pt,font= \scriptsize}]{
                        s1 ->["$G$", bend left = 25] t1;
                        s1 ->["$H$", bend right = 25] t1;
                        };
                 \end{tikzpicture}.
\end{align*}
\klactivatequotation

In both "series-composition" and "parallel-composition", we glue disjoint copies of the two "graphs".
For each "vertex" $z$ in either copy, write $\quo{z}$ for its image in the resulting "graph".

\begin{prop}\label{lem:glang-val-op}
       Let $\val \in \REL$ and $\graph[1], \graph[2]$ be "graphs".
       \begin{align*}
              \label{lem:glang-val-cap} \hat{\val}(\graph[1] \cap \graph[2])     & = \hat{\val}(\graph[1]) \cap \hat{\val}(\graph[2]) \tag{\Cref{lem:glang-val-op}$\cap$}     \\
              \label{lem:glang-val-compo} \hat{\val}(\graph[1] \compo \graph[2]) & = \hat{\val}(\graph[1]) \compo \hat{\val}(\graph[2]) \tag{\Cref{lem:glang-val-op}$\compo$}
       \end{align*}
\end{prop}
\begin{proof}
       \proofcase{(\ref{lem:glang-val-cap})}
       It suffices to prove that for every $x, y$, the following are equivalent:
       \begin{itemize}
        \item $\exists f, f \colon (\graph[1] \cap \graph[2]) \homo \gval(\val, x, y)$;
        \item $\exists f_{\graph[1]}, f_{\graph[2]},\  f_{\graph[1]} \colon \graph[1] \homo \gval(\val, x, y) \land f_{\graph[2]} \colon \graph[2] \homo \gval(\val, x, y)$.
       \end{itemize}
       \proofcase{($\Rightarrow$)}
       By letting $f_{\graph[1]} = \set{\tuple{x', f(\quo{x'})} \mid x' \in \domain{\graph[1]}}$ and $f_{\graph[2]} = \set{\tuple{x', f(\quo{x'})} \mid x' \in \domain{\graph[2]}}$.

       \proofcase{($\Leftarrow$)}
       By letting $f$ be the map on the resulting "graph" induced by $f_{\graph[1]} \cup f_{\graph[2]}$.
       Note that $f_{\graph[1]}(\src^{\graph[1]}) = x = f_{\graph[2]}(\src^{\graph[2]})$ and $f_{\graph[1]}(\tgt^{\graph[1]}) = y = f_{\graph[2]}(\tgt^{\graph[2]})$; so $f$ is indeed a map.

       \proofcase{(\ref{lem:glang-val-compo})}
       It suffices to prove that for every $x, y$, the following are equivalent:
       \begin{itemize}
        \item $\exists f, f \colon (\graph[1] \compo \graph[2]) \homo \gval(\val, x, y)$;
        \item $\exists z, \exists f_{\graph[1]}, f_{\graph[2]},\  f_{\graph[1]} \colon \graph[1] \homo \gval(\val, x, z) \land f_{\graph[2]} \colon \graph[2] \homo \gval(\val, z, y)$.
       \end{itemize}
       \proofcase{($\Rightarrow$)}
       By letting $z = f(\quo{\tgt^{\graph[1]}})$, $f_{\graph[1]} = \set{\tuple{x', f(\quo{x'})} \mid x' \in \domain{\graph[1]}}$, and $f_{\graph[2]} = \set{\tuple{x', f(\quo{x'})} \mid x' \in \domain{\graph[2]}}$.

       \proofcase{($\Leftarrow$)}
       By letting $f$ be the map on the resulting "graph" induced by $f_{\graph[1]} \cup f_{\graph[2]}$.
       Note that $f_{\graph[1]}(\tgt^{\graph[1]}) = z = f_{\graph[2]}(\src^{\graph[2]})$; so $f$ is indeed a map.
\end{proof}

\begin{prop}\label{lem:glang-val-op-u}
    Let $\val \in \REL$ and $\graph[1], \graph[2]$ be "graphs".
    Let $\glangvar$ be a set of "graphs".
    \begin{align*}
           \label{lem:glang-val-cap-u} \bigcup_{\substack{f \colon \domain{\graph[1]} \to \glangvar,\\
           g \colon \domain{\graph[2]} \to \glangvar}} \hat{\val}(\graph[1][f] \cap \graph[2][g])
           &= \bigcup_{\substack{h \colon \domain{\graph[1] \cap \graph[2]} \to \glangvar}}          \hat{\val}((\graph[1] \cap \graph
           [2])[h]) \tag{\Cref{lem:glang-val-op-u}$\cap$}     \\
           \label{lem:glang-val-compo-u} \bigcup_{\substack{f \colon \domain{\graph[1]} \to \glangvar,\\
           g \colon \domain{\graph[2]} \to \glangvar}} \hat{\val}(\graph[1][f] \compo \graph[2][g])
           &= \bigcup_{\substack{h \colon \domain{\graph[1] \compo \graph[2]} \to \glangvar}}          \hat{\val}((\graph[1] \compo \graph
           [2])[h]) \tag{\Cref{lem:glang-val-op-u}$\compo$}    
    \end{align*}
\end{prop}
\begin{proof}
    \proofcase{($\supseteq$)}
    Let $f(z) = h(\quo{z})$ and $g(z) = h(\quo{z})$.
    At each "vertex" $\quo{z}$, the copies of $h(\quo{z})$ can be mapped to the single copy attached by $h$.

    \proofcase{($\subseteq$)}
    Choose one "vertex" $z$ for each $\quo{z}$ and let $h(\quo{z}) = \begin{cases}
    f(z) & (z \in \domain{\graph[1]})\\
    g(z) & (\mbox{otherwise})
\end{cases}$. The copies selected by $h$ are present in the "graph" obtained by composing $\graph[1][f]$ and $\graph[2][g]$.
\end{proof}

\begin{proof}[Proof of \Cref{prop:glang-u}]
       By easy induction on $\term$.

       \proofcase{Case $a$ where $a \in \vsigpm$}
       For every $\tuple{x, y} \in \hat{\val}(\top)$, we have
\kldeactivatequotation
       \begin{align*}
              \tuple{x, y} \in \hat{\val}(\Tr_{\term[3]^{\lop}}(a)) & \quad\Leftrightarrow\quad \tuple{x, y} \in \hat{\val}(\term[3]^{\lop} a \term[3]^{\lop})\\
              & \quad\Leftrightarrow\quad \tuple{x, y} \in \hat{\val}(\left\{\begin{tikzpicture}[baseline = -.5ex]
                \graph[grow right = 1.2cm, branch down = 6ex, nodes={mynode}]{
                {0/{}[draw, circle]}-!-{1/{}[draw, circle]}
                };
                \node[left = .5em of 0](l){};
                \node[right = .5em of 1](r){};
                \node[above = 1.em of 0](0e0){\scriptsize $\graph[1]$};
                \node[above = 1.em of 1](1e1){\scriptsize $\graph[2]$};
                \graph[use existing nodes, edges={color=black, pos = .5, earrow}, edge quotes={fill=white, inner sep=1pt,font= \scriptsize}]{
                   0 ->["$a$"] 1;
                   0 ->[bend right] 0e0 ->[bend right] 0;
                   1 ->[bend right] 1e1 ->[bend right] 1;
                   l -> 0; 1 -> r;
                };
         \end{tikzpicture} \mid \graph[1],\graph[2] \in \glang(\term[3]^{\lop}) \right\}) \tag{\Cref{prop:glang}}                           \\
                                             & \quad\Leftrightarrow\quad \tuple{x, y} \in \hat{\val}(\glang_{\term[3]}(a)). \tag{Definition of $\protect\glang_{\term[3]}$}
       \end{align*}
\klactivatequotation

       \proofcase{Case $\top$}
       In the same way as the case of $a$ above, where the $a$-labelled edge is removed (as the "graph" in $\glang(\top)$ has no edges).

       \proofcase{Case $\eps$}
       For every $\tuple{x, y} \in \hat{\val}(\top)$, we have
       \begin{align*}
              \tuple{x, y} \in \hat{\val}(\Tr_{\term[3]^{\lop}}(\eps))  & \quad\Leftrightarrow\quad \tuple{x, y} \in \hat{\val}(\term[3]^{\lop}) \tag{as $\Tr_{\term[3]^{\lop}}(\eps) = \term[3]^{\lop} \eps \term[3]^{\lop}$ and $\REL \models \term[3]^{\lop} \term[3]^{\lop} = \term[3]^{\lop}$}\\
              & \quad\Leftrightarrow\quad \tuple{x, y} \in \hat{\val}(\left\{\begin{tikzpicture}[baseline = -.5ex]
                \graph[grow right = 1.2cm, branch down = 6ex, nodes={mynode}]{
                {0/{}[draw, circle]}
                };
                \node[left = .5em of 0](l){};
                \node[right = .5em of 0](r){};
                \node[above = 1.em of 0](0e0){\scriptsize $\graph[1]$};
                \graph[use existing nodes, edges={color=black, pos = .5, earrow}, edge quotes={fill=white, inner sep=1pt,font= \scriptsize}]{
                   0 ->[bend right] 0e0 ->[bend right] 0;
                   l -> 0; 0 -> r;
                };
         \end{tikzpicture} \mid \graph[1] \in \glang(\term[3]^{\lop}) \right\}) \tag{\Cref{prop:glang}}                           \\
                                                & \quad\Leftrightarrow\quad \tuple{x, y} \in \hat{\val}(\glang_{\term[3]}(\eps)). \tag{Definition of $\protect\glang_{\term[3]}$}
       \end{align*}

       \proofcase{Case $\emp$}
       For every $\tuple{x, y}$, we have
       \begin{align*}
              \tuple{x, y} \in \hat{\val}(\Tr_{\term[3]^{\lop}}(\emp)) & \Leftrightarrow \const{false}  \Leftrightarrow \tuple{x, y} \in \hat{\val}(\emptyset) \Leftrightarrow \tuple{x, y} \in \hat{\val}(\glang_{\term[3]}(\emp)). \tag{Definitions of $\Tr$, $\hat{\val}$, and $\protect\glang_{\term[3]}$}
       \end{align*}

       \proofcase{Case $\term[1] \compo \term[2]$}
       \begin{align*}
              &\hat{\val}(\Tr_{\term[3]^{\lop}}(\term[1] \compo \term[2])) \\
              & \quad=\quad \hat{\val}(\Tr_{\term[3]^{\lop}}(\term[1])) \compo \hat{\val}(\Tr_{\term[3]^{\lop}}(\term[2]))  \tag{Definition of $\hat{\val}$}                                                                                                                     \\
                                                  & \quad=\quad \hat{\val}(\glang_{\term[3]}(\term[1])) \compo \hat{\val}(\glang_{\term[3]}(\term[2]))          \tag{IH}                                                                                                                \\
                                                  & \quad=\quad \bigcup_{\substack{ \graph[1] \in \glang(\term[1])\\\graph[2] \in \glang(\term[2])}}
                                                  \bigcup_{\substack{f \colon \domain{\graph[1]} \to \glang(\term[3]^{\lop}),\\
                                                  g \colon \domain{\graph[2]} \to \glang(\term[3]^{\lop})}}
                                                  \hat{\val}(\graph[1][f]) \compo \hat{\val}(\graph[2][g])                 \tag{$\compo$ is distributive with respect to $\cup$} \\
                                                  & \quad=\quad \bigcup_{\substack{ \graph[1] \in \glang(\term[1])\\\graph[2] \in \glang(\term[2])}}
                                                  \bigcup_{\substack{f \colon \domain{\graph[1]} \to \glang(\term[3]^{\lop}),\\
                                                  g \colon \domain{\graph[2]} \to \glang(\term[3]^{\lop})}} \hat{\val}(\graph[1][f] \compo \graph[2][g]) \tag{\ref{lem:glang-val-compo}}                                      \\
                                                  & \quad=\quad \bigcup_{\substack{ \graph[1] \in \glang(\term[1])\\\graph[2] \in \glang(\term[2])}}
                                                  \bigcup_{\substack{h \colon \domain{\graph[1] \compo \graph[2]} \to \glang(\term[3]^{\lop})}}          \hat{\val}((\graph[1] \compo \graph
                                                  [2])[h])
                                                  \tag{\ref{lem:glang-val-compo-u}}\\
                                                  & \quad=\quad \hat{\val}(\glang_{\term[3]}(\term[1] \compo \term[2])). \tag{Definition of $\protect\glang_{\term[3]}$}
       \end{align*}

       \proofcase{Case $\term[1] \cap \term[2]$}
       \begin{align*}
             & \hat{\val}(\Tr_{\term[3]^{\lop}}(\term[1] \cap \term[2])) \\
             & \quad=\quad \hat{\val}(\Tr_{\term[3]^{\lop}}(\term[1])) \cap \hat{\val}(\Tr_{\term[3]^{\lop}}(\term[2]))                     \tag{Definition of $\hat{\val}$}                                                                                                  \\
                                                 & \quad=\quad \hat{\val}(\glang_{\term[3]}(\term[1])) \cap \hat{\val}(\glang_{\term[3]}(\term[2]))            \tag{IH}                                                                                                              \\
                                                 &\quad=\quad \bigcup_{\substack{ \graph[1] \in \glang(\term[1])\\\graph[2] \in \glang(\term[2])}}
                                                 \bigcup_{\substack{f \colon \domain{\graph[1]} \to \glang(\term[3]^{\lop}),\\
                                                 g \colon \domain{\graph[2]} \to \glang(\term[3]^{\lop})}}
                                                 \hat{\val}(\graph[1][f]) \cap \hat{\val}(\graph[2][g])           \tag{$\cap$ is distributive with respect to $\cup$} \\
                                                 & \quad=\quad \bigcup_{\substack{ \graph[1] \in \glang(\term[1])\\\graph[2] \in \glang(\term[2])}}
                                                 \bigcup_{\substack{f \colon \domain{\graph[1]} \to \glang(\term[3]^{\lop}),\\
                                                 g \colon \domain{\graph[2]} \to \glang(\term[3]^{\lop})}} \hat{\val}(\graph[1][f] \cap \graph[2][g]) \tag{\ref{lem:glang-val-cap}}                                       \\
                                                 & \quad=\quad \bigcup_{\substack{ \graph[1] \in \glang(\term[1])\\\graph[2] \in \glang(\term[2])}}
                                                 \bigcup_{\substack{h \colon \domain{\graph[1] \cap \graph[2]} \to \glang(\term[3]^{\lop})}}          \hat{\val}((\graph[1] \cap \graph
                                                 [2])[h])
                                                 \tag{\ref{lem:glang-val-cap-u}}\\
                                                 & \quad=\quad \hat{\val}(\glang_{\term[3]}(\term[1] \cap \term[2])). \tag{Definition of $\protect\glang_{\term[3]}$}
       \end{align*}

       \proofcase{Case $\term[1]^{\lop}$}
       By $\REL \models \term[3]^{\lop} \Tr_{\term[3]^{\lop}}(\term[1])^{\lop} \term[3]^{\lop} = \Tr_{\term[3]^{\lop}}(\term[1]) \cap \term[3]^{\lop} \eps \term[3]^{\lop}$, this case is shown in the same way as the cases of $\term[1] \cap \term[2]$ and $\eps$ (recall $\term[1]^{\lop} = \term[1] \cap \eps$).

       \proofcase{Case $\term[1]^{\smile}$}
       By $\Tr_{\term[3]^{\lop}}(\term[1]^{\smile}) = \Tr_{\term[3]^{\lop}}(\term[1])^{\smile}$, IH, and the fact that $\glang_{\term[3]}(\term[1]^{\smile})$ is obtained from $\glang_{\term[3]}(\term[1])$ by swapping the source and the target.

       \proofcase{Case $\term[1] \union \term[2]$}
       \begin{align*}
              \hat{\val}(\Tr_{\term[3]^{\lop}}(\term[1] \union \term[2])) & \quad=\quad \hat{\val}(\Tr_{\term[3]^{\lop}}(\term[1])) \cup \hat{\val}(\Tr_{\term[3]^{\lop}}(\term[2]))                     \tag{Definition of $\hat{\val}$} \\
                                                 & \quad=\quad \hat{\val}(\glang_{\term[3]}(\term[1])) \cup \hat{\val}(\glang_{\term[3]}(\term[2]))            \tag{IH}             \\
                                                 & \quad=\quad \hat{\val}(\glang_{\term[3]}(\term[1]) \cup \glang_{\term[3]}(\term[2])) = \hat{\val}(\glang_{\term[3]}(\term[1] \union \term[2])). \tag{Definition of $\protect\glang_{\term[3]}$}
       \end{align*}

       \proofcase{Case $\term[1]^*$}
       For a set $\glangvar$ of "graphs", let $\glangvar^{+} \defeq \set{\graph[1]_1 \compo \dots \compo \graph[1]_n \mid n \ge 1 \land \graph[1]_1, \dots, \graph[1]_n \in \glangvar}$.
       \begin{align*}
             \hat{\val}(\Tr_{\term[3]^{\lop}}(\term[1]^{*})) &\quad=\quad \hat{\val}(\Tr_{\term[3]^{\lop}}(\term[1])^{*} \compo \term[3]^{\lop}) \tag{Definition of $\Tr$} \\
             &\quad=\quad \hat{\val}(\Tr_{\term[3]^{\lop}}(\term[1])^{+} \union \term[3]^{\lop})  \tag{as $\REL \models \term[3]^{\lop} \compo \Tr_{\term[3]^{\lop}}(\term[1]) = \Tr_{\term[3]^{\lop}}(\term[1]) = \Tr_{\term[3]^{\lop}}(\term[1]) \compo \term[3]^{\lop}$ by easy induction on $\term[1]$} \\
             &\quad=\quad \hat{\val}(\glang_{\term[3]}(\term[1])^+) \cup \hat{\val}(\term[3]^{\lop})       \tag{IH}                                                                  \\
             &\quad=\quad \hat{\val}(\glang_{\term[3]}(\term[1]^+)) \cup \hat{\val}(\term[3]^{\lop})      \tag{By the same argument as that for Case $\term \compo \term[2]$}    \\
             &\quad=\quad \hat{\val}(\glang_{\term[3]}(\term[1]^*)).      \tag{Definition of $\protect\glang_{\term[3]}$}
       \end{align*}
\end{proof}
\section{Note: The Equational Theories with Respect to DREL and REL Coincide for Propositional While Programs}\label{section:DREL-KAT}
In this section, we note that the "equational theories" with respect to $\DREL_{\ttest}$ and $\REL_{\ttest}$ coincide for "propositional while programs".
We use this fact in \Cref{table:complexity-SDIPDL}.
This equivalence holds even for \emph{"regular programs"}.
We recall the set $\Tests$ of "tests" given in \Cref{section:PWP-syntax}.
The set of ""regular programs"" is the minimal subset $\Progs \subseteq \PCoR_{\set{\bl^{*}}}$ satisfying the following:
\begin{gather*}
    \begin{prooftree}[separation = .5em]
        \hypo{b \in \Tests}
        \infer1{b \in \Progs}
    \end{prooftree}
    \quad
    \begin{prooftree}[separation = .5em]
        \hypo{x \in \asig}
        \infer1{x \in \Progs}
    \end{prooftree}
    \quad
    \begin{prooftree}[separation = .5em]
        \hypo{\term[1] \in \Progs}
        \hypo{\term[2] \in \Progs}
        \infer2{\term[1] \compo \term[2] \in \Progs}
    \end{prooftree}
    \quad
    \begin{prooftree}[separation = .5em]
        \hypo{\term[1] \in \Progs}
        \hypo{\term[2] \in \Progs}
        \infer2{\term[1] \union \term[2] \in \Progs}
    \end{prooftree}
    \quad
    \begin{prooftree}[separation = .5em]
        \hypo{\term[1] \in \Progs}
        \infer1{\term[1]^{*} \in \Progs}
    \end{prooftree}.
\end{gather*}
We then have the following.
\begin{prop}\label{prop:KAT-DREL-and-REL}
For "regular programs" $\term[1]$ and $\term[2]$, we have the following:
\[\DREL_{\ttest} \models \term[1] \le \term[2] \quad\iff\quad \REL_{\ttest} \models \term[1] \le \term[2].\]
\end{prop}
\begin{proof}
\proofcase{($\Leftarrow$)}
By $\DREL_{\ttest} \subseteq \REL_{\ttest}$.

\proofcase{($\Rightarrow$)}
We show the contrapositive.
Let $\tuple{x, y} \in \hat{\val}(\term[1]) \setminus \hat{\val}(\term[2])$ where $\val \in \REL_{\ttest}$.
By \Cref{prop:glang},
there is some "graph homomorphism" $f$ from some "graph" $\graph[2] \in \glang(\term[1])$ into $\gval(\val, x, y)$.
By the definition of $\glang(\term[1])$, the "graph" $\graph[2]$ is a path, more precisely, of the following form:
\[\begin{tikzpicture}[baseline = -.5ex]
       \graph[grow up= 1.cm, branch right = 1.5cm, nodes={font=\tiny}]{
       {0/{$0$}[mynode, fill = gray!30],
       1/{$1$}[mynode, fill = gray!30],
       2/{$2$}[mynode, fill = gray!30],
       /{$\dots$},
       nm1/{}[mynode, fill = gray!30],
       n/{$n$}[mynode, fill = gray!30]}
       };
       \node[left = 4pt of 0](0l){}; \path (0l) edge[earrow, ->] (0);
       \node[right = 4pt of n](nr){}; \path (n) edge[earrow, ->] (nr);
       \graph[use existing nodes, edges={color=black, pos = .5, earrow}, edge quotes={fill = white, inner sep=1pt,font= \scriptsize}]{
       (0) ->["$a_1$"] (1) ->["$a_2$"] (2); (nm1) ->["$a_{n}$"] (n);
       };
\end{tikzpicture}.\]
Let $\graph[2]'$ be the "graph" obtained from $\graph[2]$ by merging all adjacent "vertices" $i-1$ and $i$ such that $a_i \in \psig$.
The "graph" $\graph[2]'$ is of the following form (up to graph isomorphisms):
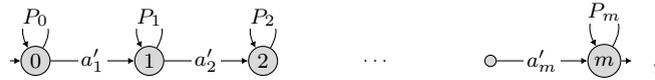
\[\begin{tikzpicture}[baseline = -.5ex]
       \graph[grow up= 1.cm, branch right = 1.5cm, nodes={font=\tiny}]{
       {0/{$0$}[mynode, fill = gray!30],
       1/{$1$}[mynode, fill = gray!30],
       2/{$2$}[mynode, fill = gray!30],
       /{$\dots$},
       nm1/{}[mynode, fill = gray!30],
       n/{$m$}[mynode, fill = gray!30]}
       };
       \node[left = 4pt of 0](0l){}; \path (0l) edge[earrow, ->] (0);
       \node[right = 4pt of n](nr){}; \path (n) edge[earrow, ->] (nr);
       \graph[use existing nodes, edges={color=black, pos = .5, earrow}, edge quotes={fill = white, inner sep=1pt,font= \scriptsize}]{
       (0) ->["$a_1'$"] (1) ->["$a_2'$"] (2); (nm1) ->["$a_{m}'$"] (n);
       0 ->["$P_0$",out = 60, in = 120, looseness = 8] 0;
       1 ->["$P_1$",out = 60, in = 120, looseness = 8] 1;
       2 ->["$P_2$",out = 60, in = 120, looseness = 8] 2;
       n ->["$P_m$",out = 60, in = 120, looseness = 8] n;
       };
\end{tikzpicture},\]
where $P_i \subseteq \psig$ and $a_i' \in \asig$.
By the map $f$ with $\hat{\val}(p) \subseteq \hat{\val}(\eps)$ for each $p \in \psig$,
we have a "graph homomorphism" $f' \colon \graph[2]' \homo \gval(\val, x, y)$ (induced by $f$).
From $\graph[2]'$,
we define the "valuation" $\val' \colon \vsig \to \wp(\range{0, m}^2)$ so that
\begin{itemize}
\item $\val'(p) \defeq \set{\tuple{i, i} \mid i \in \range{0, m} \mbox{ and } \tuple{f'(i), f'(i)} \in \hat{\val}(p)}$ for $p \in \psig$,
\item $\val'(a) \defeq \set{\tuple{i-1, i} \mid i \in \range{1,m} \mbox{ and } a_i' = a}$ for $a \in \asig$, and
\item $\val'(z) \defeq \emptyset$ for the other "variables" $z$.
\end{itemize}
By construction, $\val' \in \DREL_{\ttest}$.
We then have $\graph[2]' \homo \gval(\val', 0, m)$ by the identity map, and hence $\graph[2] \homo \gval(\val', 0, m)$.
We thus have $\tuple{0, m} \in \hat{\val}'(\term[1])$ by \Cref{prop:glang}.
Also,
by $f' \colon \gval(\val', 0, m) \homo \gval(\val, x, y)$ (with respect to the labels in $\vsig \cup \set{\eps}$, which suffices as complements do not occur in $\term[2]$) and $\tuple{x, y} \not\in \hat{\val}(\term[2])$,
we have $\tuple{0, m} \not\in \hat{\val}'(\term[2])$.
Thus, $\tuple{0, m} \in \hat{\val}'(\term[1]) \setminus \hat{\val}'(\term[2])$.
Hence, this completes the proof.
\end{proof}
Consequently, by combining with \cite[Theorem 6]{kozenKleeneAlgebraTests1996}, we have:
\[\DREL_{\ttest} \models \term[1] \le \term[2] 
\quad\iff\quad \REL_{\ttest} \models \term[1] \le \term[2]
\quad\iff\quad \GS(\term[1]) \subseteq \GS(\term[2]).\]
Here, $\intro*\GS(\term)$ is the guarded string language \cite[Section 3.1]{kozenKleeneAlgebraTests1996} of a "regular program" $\term$ where,
precisely, the definition for "primitive tests" is replaced
as follows (assuming $\psig = \set{p_0, \dots, p_{k-1}, \tilde{p}_0, \dots, \tilde{p}_{k-1}}$):
\begin{align*}
\GS(p_i) &\defeq \set{P \subseteq \set{p_0, \dots, p_{k-1}} \mid p_i \in P}, & 
\GS(\tilde{p}_i) &\defeq \set{P \subseteq \set{p_0, \dots, p_{k-1}} \mid p_i \not\in P}
\end{align*}
for $i \in \range{0, k-1}$.
Hence, the "equational theory" above can be characterized 
by "that@equational theory" of Kleene algebra with tests \cite[Theorem 8]{kozenKleeneAlgebraTests1996}.

\end{document}